\documentclass[11pt]{article}
\usepackage[margin=1in]{geometry}
\usepackage{amsmath}
\usepackage{amsthm}
\usepackage{amssymb}
\usepackage{mathtools}
\usepackage{algorithm}
\usepackage{subfig}
\usepackage{color}
\usepackage[english]{babel}
\usepackage{graphicx}
\usepackage{wrapfig,epsfig}
\usepackage{epstopdf}
\usepackage{url}
\usepackage{graphicx}
\usepackage{color}
\usepackage{epstopdf}
\usepackage{algpseudocode}
\usepackage{scrextend}
\usepackage[T1]{fontenc}
\usepackage{bbm}
\usepackage{comment}
\usepackage{authblk}
\usepackage{multicol}{}
\usepackage{wasysym}
\usepackage{tikz}

\DeclareMathAlphabet{\mathpzc}{OT1}{pzc}{m}{it}

\definecolor{b2}{RGB}{51,153,255}
\definecolor{yl}{RGB}{255,80,0}

\usepackage{hyperref}
\hypersetup{colorlinks=true,citecolor=red,linkcolor=blue,urlcolor=black}
\usetikzlibrary{arrows}
\graphicspath{{./figs/}}

\date{}

\title{
A Framework for Building Data Structures \\ from 
Communication Protocols
{\footnote{Research supported in part by NSF grants (CCF-1617955, CCF-1740833, CCF-2008733), and Simons Foundation (\#491119), and by 
NSF CAREER award CCF-1844887, and ISF grant \#3011005535.}
}
\author{
Alexandr Andoni\thanks{\url{andoni@cs.columbia.edu}. Columbia University.}
\quad
Shunhua Jiang\thanks{\url{sj3005@columbia.edu}. Columbia University.}
\quad 
Omri Weinstein\thanks{\url{omri@cs.columbia.edu}. Columbia University and the Hebrew University.}
}
}

\newtheorem{theorem}{Theorem}[section]
\newtheorem{lemma}[theorem]{Lemma}
\newtheorem{definition}[theorem]{Definition}

\newtheorem{corollary}[theorem]{Corollary}

\newtheorem{fact}[theorem]{Fact}

\newcommand{\R}{\mathbb{R}}

\newcommand{\ov}{\overline}

\DeclareMathOperator*{\E}{{\mathbb{E}}}

\DeclareMathOperator{\poly}{\mathrm{poly}}
\DeclareMathOperator{\OV}{\mathsf{OV}}

\DeclareMathOperator{\PM}{\mathsf{PM}}

\DeclareMathOperator{\cG}{\mathcal{G}}
\DeclareMathOperator{\cP}{\mathcal{P}}
\DeclarePairedDelimiterX{\infdivx}[2]{(}{)}{#1 \,\delimsize\|\, #2}
\DeclareMathOperator*{\Div}{D}
\newcommand{\D}{\Div\infdivx}
\DeclareMathOperator*{\smallDiv}{d}

\DeclareMathOperator{\pub}{\mathrm{pub}}
\DeclareMathOperator{\pri}{\mathrm{pri}}

\newcommand{\dd}{\smallDiv\infdivx}

\newcommand{\X}{{\cal X}}
\newcommand{\Y}{{\cal Y}}

\newcommand{\Pat}{P\v{a}tra\c{s}cu }

\begin{document}

\maketitle

\begin{abstract} 
We present a general framework for designing efficient data structures for  high-dimensional pattern-matching problems ($\exists \;? i\in[n], f(x_i,y)=1$) through communication models in which $f(x,y)$ admits \emph{sublinear} communication protocols with \emph{exponentially-small error}. Specifically, we reduce the data structure problem to the Unambiguous Arthur-Merlin (UAM) communication complexity of $f(x,y)$ under \emph{product distributions}.  

We apply our framework to the \emph{Partial Match} problem (a.k.a, matching with wildcards), whose underlying communication problem is \emph{sparse set-disjointness}. When the database consists of $n$ points in dimension $d$, and the number of $\star$'s in the query is at most $w = c\log n \;(\ll d)$, the fastest known linear-space data structure (Cole, Gottlieb and Lewenstein, STOC'04) had query time $t \approx 2^w = n^c$, which is nontrivial only when $c<1$. By contrast, our framework produces a data structure with query time $n^{1-1/(c \log^2 c)}$ and space close to linear.

To achieve this, we develop a one-sided $\epsilon$-error communication protocol for Set-Disjointness under product distributions with $\tilde{\Theta}(\sqrt{d\log(1/\epsilon)})$ complexity, improving on the classical result of Babai, Frankl and Simon (FOCS'86). Building on this protocol, we show that the \emph{Unambiguous AM} communication complexity of \emph{$w$-Sparse} Set-Disjointness with $\epsilon$-error under product distributions is $\tilde{O}(\sqrt{w \log(1/\epsilon)})$,  \emph{independent} of the ambient dimension $d$, which is crucial for the Partial Match result. 
Our framework sheds further light on the power of \emph{data-dependent} data structures, which is instrumental for reducing to the (much easier) case of product distributions.

\end{abstract}

\tableofcontents

\section{Introduction}

A prototypical  problem in information retrieval and pattern-matching is to preprocess a high-dimensional dataset $X_n = \{x_1,x_2,\ldots,x_n\}\subset \R^d$, such that given a query $y\in \R^d$
(e.g., a keyword, a text document, etc.), it is possible to efficiently answer the following, for a predicate $f:\R^d\times\R^d \to \{0,1\}$:
\begin{equation}\label{eq_pattern_matching}
\exists  \; x_i \in X_n \; s.t.\; 
f(x_i,y) = 1 \; ?
\end{equation}

Some canonical examples for the predicate $f$ are near(est)-neighbor search (e.g., $f(x_i,y) = 1$ iff $\|x_i-y\|\le t$), range searching and quantile queries \cite{brodal:rangemedian, AgarwalSurveyRangeCounting97, chan19kd}, searching with ``wildcards'' \cite{CGL04, CFPSS16}, set-containment and intersection  \cite{cip02, KW20, CP17, ak20} to mention a few. Perhaps the most basic pattern-matching query, which dates back to at least as early as Rivest's thesis \cite{r74} and turns out to subsume most aforementioned ones, is the \emph{partial match} (PM) problem: 

\begin{definition}[Partial Match] \label{prob_PM} 
Preprocess a dataset of $n$ points over 
$x_1,\ldots, x_n \in \{0,1\}^d$ into a data structure of size $s$, such that 
given a query $y \in \{0,1,\star\}^d$, 
it can quickly report whether there exists a data 
point $x_i$ that matches $y$ on all non-$\star$ 
coordinates. When the number of $\star$'s in the query is bounded by $w \leq d$, we denote the problem  $\PM_{n,d,w}$. The unbounded case is denoted $\PM_{n,d}$. 
\end{definition}

In the context of information-retrieval,  partial-match  models pattern-matching queries between documents and keywords in search engines and spatial 
databases such as Google, SQL, or Elasticsearch \cite{Data_mining_book}. This problem is  essentially equivalent to $\Theta(d)$-dimensional orthogonal range searching with aligned query rectangles and $\ell_\infty$ nearest neighbor search \cite{chan19kd, cip02}. When $d\gg 1$, the classic low-dimensional geometric data structures such as range-trees \cite{brodal:rangemedian} or quad-trees \cite{spatial_DS_book} become futile.  
PM also plays a key role in routing and IP packet classification, where the goal is to classify a packet depending on its parameters (source and destination address, length etc) \cite{SSV99, fm00, EM01}. For a more thorough exposition on the role of partial match in information-retrieval and data mining, we refer the reader to \cite{cip02, ak20} and references therein.   

From a theoretical standpoint, 
partial match is equivalent to the \emph{online} version of the  
\emph{Orthogonal Vectors} Problem \cite{w05} ($\OV_{n,d}$): 
Here, the inputs are two sets $A,B \subset \{0,1\}^d$ of 
$n$ vectors each, and the goal is to determine whether there is an orthogonal pair of vectors $a\in A, b\in B$.\footnote{Indeed, the transformation of queries $\star \mapsto 00, \;1 \mapsto 10, \; 0 \mapsto 01$ and data points $1 \mapsto 01, \; 0 \mapsto 10$ reduces partial match to online OV up to doubling the dimension. The opposite reduction transforms queries as $0 \mapsto \star, \;1 \mapsto 0$.} The trivial upper bound  for this offline (batch) problem is $O(n^2 d)$, or $O(n^2 w)$  for \emph{Sparse OV} ($\OV_{n,d,w}$), i.e., assuming each vector has $\leq w$ ones.  %
When $d = c \log n$ for an arbitrarily large constant $c$, 
the best \emph{offline} algorithm for (dense) $\OV_{n,c\log n}$ runs in time 
\begin{equation}\label{eq_offlineOV}
    n^{2-1/O(\log c)}
\end{equation}  
using Fast Matrix Multiplication to implement the ``polynomial method'' \cite{awy14, cw16, acw16, alman2025average}. It is unclear how to extend this speedup to the Sparse OV case, where the dimension is assumed to be 
$d = \omega(\log n)$ \cite{GIKW18}.  The \emph{OV Conjecture} (OVC)  postulates that no truly subquadratic algorithm ($n^{1.999}$) exists for \emph{arbitrarily} large but fixed constant $c$ -- indeed, such algorithm would refute SETH  \cite{w05, ipz01} (i.e., yield a $2^{(1-\epsilon) n}$ time algorithm for SAT). OVC is widely considered a cornerstone in fine-grained complexity \cite{Wil19survey, ARW17}. From this perspective, understanding the power of preprocessing for (Bounded) Partial Match, may lead to progress on the \emph{Sparse OV Conjecture} \cite{GIKW18}, as exploiting sparsity seems to require algorithms which are more \emph{combinatorial} in nature, outside the realm of fast matrix multiplication (FMM). 
Indeed, FMM speedups are conjectured to be impossible in online settings \cite{HKNS15}, %
so the partial match problem morally captures the restriction to combinatorial algorithms.

\paragraph{Problem Status.}
Exact search problems are generally believed to suffer from the \emph{curse of dimensionality} \cite{IM98}, which informally postulates that no truly sublinear query time (e.g., $t <n^{0.99}$) is possible unless exponential $s > \exp(d^{1-o(1)})$ 
storage space (and preprocessing time) is used, which are essentially the two trivial data structures for any $d$-dimensional problem. %
In attempt to circumvent this conjectured phenomenon, much of the work on pattern-matching data structures in the past two decades has focused on the easier 
\emph{approximate} versions of partial match such as approximate-search with ``wildcards''  or approximate set-containment (via locality-sensitive filtering  \cite{CFPSS16, CP17, ak20}). Much less is known about exact PM, which is the focus of this paper.

The first data structure for the general (unrestricted) problem ($\PM_{n,d}$) was obtained by Rivest \cite{r74}, who analyzed the average case complexity of the problem. 
After 30+ years, the state-of-art data structure for $\PM_{n,c\log n}$ is due to Chan \cite{chan19kd}, who showed how to achieve expected query time 
\begin{align} \label{eq_Chan_DS}
t=O(n^{1-O(1/c\log c)}) 
\end{align}
with close-to-linear space (say, $s= n^{1.1}$), slightly refuting the 
``curse-of-dimensionality'' of the problem. 

In the sparse case ($\PM_{n,d,w}$), when the number of $\star$'s  is bounded by $w \ll d$, Cole, Gottlieb and Lewenstein \cite{CGL04} gave an $O(n \log^w n)$-space data structure with query time $t = \tilde{O}(2^w)$,\footnote{The notation $\tilde{O}(1)$ hides $\poly(\log d)$ factors.}   independent of the alphabet size $|\Sigma|$, reminiscent of the time-space tradeoffs of range-trees in constant dimension. 
Assuming $w = c\log n$, the query time of this data structure becomes 
$\tilde{O}(n^c)$, which is nontrivial (sublinear) only for $c <1$. \\ 

\subsection{Our results}
\paragraph{Our result 1: Partial Match data structure.} We show that it is possible to achieve sublinear query time even when the number of ``wildcards'' in the query is as large as $w = c\log n$ for \emph{arbitrarily large}  constant $c \gg 1$: 

\begin{theorem}[Partial Match with Bounded Wildcards, Informal version of Theorem \ref{thm:partial_match_ds}] \label{thm_main_informal}
For any large enough constant $c$, there is a word-RAM data structure for $\PM_{n, d, w}$, where $w \leq c\log n$, with  space $s=O(n^{1.1})$
and expected query time $t_q = n^{1- 1/\Theta(c\log^2 c
))}$. The data structure can report \emph{all} matches of the query in the dataset.\footnote{In this case, the query time also has an additive term $\#$matches. For the decision problem, it is possible to build a data structure that avoids this term by first subsampling the dataset.} 
\end{theorem}

This data structure improves over \cite{CGL04} for $w =  c\log n$ for any $c> 1$. %
As mentioned above, Theorem \ref{thm_main_informal} is the first PM data structure with sublinear query time.
For calibration, note that for the unbounded case $\PM_{n,d}$ in logarithmic dimension $d=c\log n$, Theorem \eqref{thm_main_informal} recovers the time-space tradeoff \eqref{eq_Chan_DS} of Chan's data structure \cite{chan19kd} up to $\log c$ factors, hence our result unifies and improves the time-space landscape for partial match. %

\paragraph{Our result 2: A Framework for data structures from communication protocols.}  Theorem \ref{thm_main_informal} is an instance of a general framework 
we develop in this paper. Our framework reduces the problem of designing data structures for high-dimensional pattern-matching problems (as in Eq.~\eqref{eq_pattern_matching}), 
to the problem of designing efficient \emph{communication protocols}    
for the underlying two-party problem $f:\X\times \Y\to \{0,1\}$ under \emph{product distributions}, i.e., where inputs $(x,y)\sim \lambda_x\times \lambda_y$.
The simplest variant of this framework is a reduction to the standard two-party distributional communication model with \emph{exponentially-small false-positive} error, which informally states the following:  

\begin{center}
\em For $\epsilon>0$, suppose the two-party $\epsilon$-error distributional communication complexity of $f:\X\times \Y\to 
\{0,1\}$ under %
 \emph{product distributions} is $C_\epsilon$. 
Then there is a data structure \\ 
supporting queries from Eq.~\eqref{eq_pattern_matching} with space $n\cdot 2^{O(C_{\epsilon})}$ and query time $2^{O(C_\epsilon)}+\epsilon n$.
\end{center}

For the Partial Match application, our framework requires an efficient communication protocol for the \emph{sparse set-disjointness} problem ($\mathsf{Disj}_{d,w}$) where one set has size $\leq w$ and the other set has size $\geq d-w$.\footnote{The aforementioned reduction to OV doesn't preserve sparsity. Instead, we design a communication protocol that reduces the $w$-sparse partial match problem (where the query has at most $w$ wildcards) to the $w$-sparse set containment where the sets have a maximum size of $w$. This is also equivalent to the $w$-sparse set disjointness problem $\mathsf{Disj}_{d,w}$ (where one set has size $\leq w$ and the other set has size $\geq d-w$).}
The basic advantage of this reduction is that, while set-disjointness has maximal $\Omega(d)$ communication complexity in the worst case \cite{ks92,r90}, %
it is well known that the distributional complexity of set-disjointness under product distributions is \emph{sublinear}, albeit only with \emph{constant error}: $O(\sqrt{d}\log d)$ for $\mathsf{Disj}_{d}$ and $O(\sqrt{w}\log d)$ for $\mathsf{Disj}_{d,w}$ \cite{bfs86}. 
Achieving sublinear communciation with \emph{exponentially-small} error is a different ballgame --  
Unlike randomized communication complexity, in the distributional setting one cannot simply amplify the success by repeating the protocol. Moreover, for the above simulation to yield sublinear query time, even a Chernoff-bound dependence $O(\log (1/\epsilon))$ on the error would not suffice -- we need a protocol with \emph{sub}-logarithmic dependence on $\epsilon$. Perhaps surprisingly, we show that this is possible: 
\begin{theorem}
For any $\epsilon > 0$, the two-party $\epsilon$-error distributional communication complexity of $\mathsf{Disj}_{d,w}$ under product distributions is
$O(\sqrt{w \log(1/\epsilon)} \cdot \log d).$
\end{theorem}
Note that, even in the above, the communication complexity has an unacceptable factor $\log d$ (due to the fact that even sending a small set of size $\sqrt{w}$ from a size $d$ universe requires $\sqrt{w} \log d$ bits). Unfortunately, 
even a mild dependence on the ambient dimension is prohibitive for the PM data structure\footnote{We can bound the term $\sqrt{w \log(1/\epsilon)}$ by $0.1 \log n$. The query time is exponential in the communication complexity $C_{\epsilon}$, which becomes $C_{\epsilon}=0.1 \log n \cdot \log d = \Omega(\log n\log\log n)$ even when $d = \log^2 n$, and this makes the query time $2^{C_{\epsilon}} = n^{\Omega(\log\log n)}$.}.  
In order to obtain \emph{dimension-free} communication for product $\mathsf{Disj}_{d,w}$, we make use of \emph{non-determinism}. Intuitively, the benefit of using an oracle (Merlin) 
here is that she can send a \emph{relative encoding} of the smaller set as a \emph{subset} of the larger set, and this requires much fewer bits.  

Because of this non-determinism, the formal statement of the framework is slightly more nuanced in order to allow for a somewhat stronger protocols than a two-party protocol. %
In particular, our framework is strong enough to simulate {\em Unambiguous Artur-Merlin} protocols
\cite{goos2015zero}, which are randomized multi-round interactive proofs in which, for every outcome of the randomness, there is \emph{at most one} 
advice from Merlin which causes Alice and Bob to accept.  
In this model, we can indeed get a dimension-free protocol for product sparse set disjointness and hence product sparse partial match, %
with exponentially small soundness error:  
\begin{theorem}[Informal version of Theorem~\ref{thm:w_sparse_partial_match_protocol}]\label{thm:informal_protocol}
Fix $d,w\ge 1, \epsilon>0$.
    Consider the communication problem where Alice has $x\in \{0,1\}^d$ and Bob  has $y\in \{0,1,\star\}^d$ with at most $w$ stars, and they need to compute the function $f_{PM}(x,y)$ which is 1 iff $x_i=y_i$ on all $i$ with $y_i\neq \star$. Then, for any product distribution over $(x,y)$, there exists an $\epsilon$-error UAM protocol with communication $\tilde O(\sqrt{w\log (1/\epsilon)})$. 
\end{theorem}

Note that using this framework for the sparse Partial Match problem, together with the theorem from above with $\epsilon\approx n^{-\Omega(1/c)}$, we obtain Theorem~\ref{thm_main_informal} (modulo the hidden $\log c$ factors).

\paragraph{Our result 3: Lower Bounds.} 
In an attempt to understand the disparity between the online and offline complexities of OV, we also explore the limitations of partial match algorithms in our framework.
We note that time-space lower bounds for PM have been studied since at least~\cite{BORl}.
For example, \Pat \cite{p08, Pat11} proved that any \emph{decision tree} data structure for $\PM_{n,d}$ 
must either have size (space) $2^{\Omega(d)}$, or depth (query time) $\Omega(n^{1-\epsilon}/d)$. 
Unfortunately, for general (word-RAM or cell-probe) data structures, 
only much weaker lower bounds are known --  the highest unconditional lower bound is merely $t_q > \tilde{\Omega}(d)$ \cite{p08, BR02}, even against linear space data structures (in fact, this is the highest unconditional lower bound on any static problem to date). 
 To circumvent the barriers of unconditional lower bounds, 
 a substantial line of work has focused on proving geometric lower bounds in more restricted (``indivisible'') models of data structures, most notably, the List-of-Points model \cite{alrw16,ak20}, which captures essentially all data independent data structures (see Section~\ref{sec:lop_lower_bound} for formal definition).
In this model, the data structure is not allowed to encode the data but only navigate it through pointers (nevertheless, LoP is incomparable to the aforementioned decision-tree model, as the latter does not allow preprocessing at all). 

We prove a lower bound in the List-of-Points model, which nearly matches our upper bound up to polynomial dependence on $c$ (see Theorem~\ref{thm:main_lower_bound} in Section \ref{sec:lop_lower_bound} for formal statement). For this purpose, we adapt a lower bound of \cite{ak20}; while those lower bound were developed for \emph{approximate} set-containment, it turns out we can still leverage them to prove (meaningful) lower bounds for the exact problem when dimension is $O(\log n)$.

\begin{theorem}[List-of-Points Lower Bound for $\PM$] \label{thm_LoP_LB_informal}
Any list-of-points data structure for $\PM_{n,c\log n}$ with $n^{O(1)}$ space, must have query time $t_q \geq \Omega(n^{1-1/\sqrt{c}})$. 
\end{theorem}

The list-of-points model is a general model that captures all known data-independent structures for high-dimensional similarity search, including LSH \cite{IM98} and LSF \cite{CP17}.
However, LoP lower bounds do not apply to \emph{data-dependent} algorithms, and hence in principle Theorem \ref{thm_LoP_LB_informal} does not rule out further improvement to our (data-dependent) upper bound. Nevertheless, the lower bound from Theorem \ref{thm_LoP_LB_informal} is proved by considering \emph{random} (``planted'') instances, whereas our upper bound from Theorem \ref{thm_main_informal} holds for worst-case instances, possibly explaining the $1/\sqrt{c}$ ``gap'' in the query times of both theorems. In light of these remarks, it would be interesting to prove an \emph{adaptive} analogue of Theorem \ref{thm_LoP_LB_informal} against data-dependent algorithms, bearing that polynomial unconditional (cell-probe) lower bounds are beyond the reach of current techniques.\footnote{One approach to a more powerful such lower bound would be via the approach from~\cite{ar16lower}, who showed a lower bound in a (carefully formalized) data-dependent LSH model, for the approximate near-neighbor problem.}%

For further calibration, we note that a polynomial space and preprocessing data structure for $\PM_{n,c\log n}$ with \emph{truly} sublinear query time (independent of $c$), say $t_q = n^{0.99}$, would refute the $\OV$ conjecture (by a standard self-reducibility trick, see \cite{rub18}). As such, assuming SETH, it is unlikely that the dependence on $c$ in the exponent can be eliminated. 
Our lower bound (Theorem \ref{thm_LoP_LB_informal}) can be viewed as evidence that \emph{polynomial} dependence on $(1/c)$ in the query time (exponent) may in fact be the best one could hope for in the online setting, and hence provides a certain separation from the logarithmic dependence  $(1/\log c)$ in the offline case of OV (Eq.~\eqref{eq_offlineOV}).%

\subsection{Technical Overview}

We now present the high-level ideas used to develop our main theorem. As mentioned above, our algorithm follows from the combination of two components: (1) a new framework for developing data structures from communication protocols, (2) a new communication protocol for the sparse partial match problem in our specialized distributional model. We describe each of these components next.
\subsubsection{Data-Dependent Data Structures from  Communication Protocols}
The classical connection between data structures and asymmetric (``lopsided'') communication complexity \cite{mnsw98, p08} has long been the primary tool for proving data structure \emph{lower bounds}: A data structure with space $s$ and query time $t$ for a search problem $\cP_f(X_n,y)$ (a-la \eqref{eq_pattern_matching}) implies a 
$t$-round communication protocol $\pi$ for the corresponding communication game $\cG_f(X,y)$,
in which Alice (holding the dataset $X_n = (x_1,\ldots,x_n) \subset \{0,1\}^d$) sends $a \sim tw$ bits where $w$ is the word size, and Bob (the ``querier'', holding $y \in \{0,1\}^d$) sends $b = O(t\log s)$ bits. 
Despite the success of this simulation argument in yielding (unconditional) state-of-the-art  static and dynamic lower bounds on many important data structure problems \cite{p08, SVY09, PT11, WY16}, it suffers from an inherent limitation: In  communication world, Alice (the ``data-structure'') is an \emph{all-powerful} player, whereas in the data structure world, Alice is merely a \emph{memoryless} table of cells (and therefore cannot perform any computation on Bob's messages/probes). This discrepancy between the models has a daunting consequence, namely, for any $d$-dimensional problem $\cP_f(X_n,y)$, there is always a trivial $O(d)$-bit protocol where Bob simply sends his input (query) to Alice, who then declares the answer. Clearly, this trivial protocol does  \emph{not} yield any valid data structure, and this suggests that in general, communication protocols are \bf  much  stronger than data structures\rm.

The most enduring message of our work is a framework which shows that for a certain variant of \emph{distributional} unambiguous Artur-Merlin (UAM) communication models, the opposite direction is also true: 
For a data structure problem $\cP_f(X_n, y)$ of the form \eqref{eq_pattern_matching} (where $X_n = (x_1,\ldots,x_n)\subset \{0,1\}^d$, $y\in \{0,1\}^d$), consider the induced ``marginal'' two-party communication problem $$f(x_i,y) $$ between Alice and Bob (note that, unlike the aforementioned lopsided reduction, this is a symmetric problem with respect to players' inputs). 
Oversimplifying our actual result, we show that a computationally-efficient protocol 
$\pi$ in which Alice sends $c_a$ bits and Bob sends $c_b$ bits,
and which solves $f(x,y)$ under any \emph{product distribution} %
with $\epsilon$ \emph{one-sided} error, %
can be translated into a data structure for $\cP_f(X_n,y)$ with 
\[
\text{space}~s = 2^{O(c_a+c_b)} \cdot n, \;\;\;\; ~~ \text{query time}~t_q = 2^{c_a} + \epsilon n.
\]

More precisely, our communication model requires that the protocol $\pi$ works for {\bf distributional $x$ and worst-case $y$}: Alice's input $x \sim \lambda$ is generated from an arbitrary distribution $\lambda$, whereas Bob's input $y$ is a worst-case input.\footnote{While this is important for obtaining our final result, one can still use standard product distribution model to obtain data structures in the {\em cell-probe model}. This can be done by running MWU-like worst-case to average-case reduction as in \cite{annrw18-spectralGaps}. Thus the restriction to worst-case $y$'s is not fundamental.} (Crucially, $\lambda$ is not allowed to depend on $y$, hence this model generalizes the class of product distributions $\mu = \mu_x\times\mu_y$). %
We also allow a know-all Merlin who is a prover for the answer 1 (yes); but we also
require that for any $x$ and $y$, there is at most one accepted message from Merlin, whose length we denote by $c_m$. We characterize a protocol by its error and soundness.
\begin{itemize}
\item We say the protocol has $\delta$ {\rm soundness} if given any incorrect Merlin's message, the protocol rejects (i.e., outputs $0$) with probability at least $1-\delta$.\footnote{Note that this unambiguous soundness definition is different from that of \cite{goos2015zero}.} 
\item We say the protocol has $\epsilon$ {\em error} if given the correct Merlin's message, the protocol's output differs from $f(x,y)$ with probability at most $\epsilon$, where the error is measured over both the randomness of the protocol and the randomness of the input $x$. 
\end{itemize}
In the simulation, we shall take $\lambda$ to be the uniform distribution over the $n$ data points $x_i$, and in this way $\epsilon n$ naturally represents the expected number of data points that the query mistakes on. To enable our ``reverse simulation'', we require that the communication protocol has one-sided error, more precisely, only {\bf false positive} error --- that whenever on an input pair $(x,y)$, where $x$ and $y$ are matching, the protocol must output $1$ under the correct Merlin's message. We need this requirement so that in the simulation the data structure can safely ignore all $(x_i,y)$ pairs on which the protocol outputs $0$. %

\begin{theorem}[Data Structures from Product-Distribution Protocols, Informal version of Theorem~\ref{thm:reduction_CC_DS}] \label{thm_meta_informal}
If there is a protocol $\pi$ for the 2-party problem $f(x,y) : \{0,1\}^d \times \{0,1\}^d \to \{0,1\}$ with soundness $\delta$ and false positive error $\epsilon$ under distributional $x \sim \lambda$ and worst-case $y$, and the number of bits that Alice, Bob, and Merlin send are bounded by $c_a$, $c_b$, and $c_m$, then there is a data structure for $\cP_f(X_n,y)$ (c.f \eqref{eq_pattern_matching}) with space $s = 2^{O(c_b)} \cdot n$ %
and query time $t_q = 2^{c_a+c_m} + 2^{c_m} \cdot \delta n + \epsilon n$.
\end{theorem}

Bearing the technical nuances of the model, the key point of Theorem \ref{thm_meta_informal} is that product distributions often allow for \emph{much smaller} (sublinear) communication complexity compared to worst-case (randomized) CC, as in the notable special case of $f = \mathsf{Disj}$ \cite{bfs86} corresponding to $\PM$.

\paragraph{Proof Idea of Theorem \ref{thm_meta_informal}}
Let $\pi$ by a (randomized) protocol for $f$ in the ``strong'' distributional UAM model above with $\delta$ soundness and $\epsilon$ false-positive error. Since the protocol works for any product distribution, it must in particular work under the distribution $x \sim \lambda = \mathcal{U}(x_1,x_2,\ldots, x_n)$ which is uniform over the dataset $X_n$. 
Our simulation works as follows. 
\begin{itemize}
    \item {\bf Preprocessing:} At preprocessing time, the data structure builds a tree that alternates between Alice nodes, Bob nodes, and Merlin nodes. Each node enumerates all possible messages of the protocol. In this way a root-to-leaf path corresponds to one possible transcript of the protocol, and there are in total $2^{O(c_a+c_b+c_m)}$ leaves. Each leaf stores all data points $x$ that could follow the transcript of this leaf with output $1$. Since each node stores at most $n$ data points, the total space is bounded by $s = 2^{O(c_a+c_b)} \cdot n$.
    \item {\bf Query:} Given a query $y$, the data structure simulates Bob's messages in the protocol tree: It walks down the tree, generates messages of Bob according to $y$, while enumerating \emph{all} possible messages from Alice and Merlin. In this way the query reaches $2^{c_a+c_m}$ number of leaves. After reaching a leaf, the query checks if there exists a matching data point stored in that leaf.
    \begin{itemize}
        \item Since the protocol has soundness guarantee $\delta$ for the $2^{c_m}$ incorrect Merlin's messages, the query could encounter $2^{c_m} \cdot \delta \cdot n$ extra non-matching data points.
        \item Since the protocol has an error probability of $\epsilon$ under the unique correct Merlin's message, the query also encounters $\epsilon n$ extra non-matching data points.
    \end{itemize}
    So the total query time is $t_q = 2^{c_a+c_m} + 2^{c_m} \cdot \delta n + \epsilon n$.
\end{itemize}

Note that the data structure we get from this simulation is inherently data-dependent. Indeed, the communication protocol depends on $\lambda = \mathcal{U}(X_n)$, and this means the tree structure that we build is different for different datasets.

\subsubsection{A New Protocol for Product Sparse Partial Match}\label{sec:tech_overview_protocols}
Next we show how to build a sublinear communication protocol for the $w$-sparse partial match problem with \emph{exponentially small false-positive error} ($\epsilon \sim \exp(-w/c)$) in our distributional UAM model. We first reduce the $w$-sparse partial match problem $\PM_{d,w}$ to the $w$-sparse set disjointness problem $\mathsf{Disj}_{d,w}$ where $y$ has at most $w$ number of $0$'s. Although the communication complexity of the disjointness problem $\mathsf{Disj}_{d}$ over worst-case inputs has a lower bound of $\Omega(d)$, the classic result of \cite{bfs86} showed an $O(\sqrt{d} \log d)$ protocol with constant error for inputs from product distributions. We generalize this protocol to a protocol for $\mathsf{Disj}_{d,w}$ with arbitrary false positive error $\epsilon$, soundness $\delta$, and communication 
\[
c_a, c_b, c_m \leq \sqrt{w \log \frac{1}{\epsilon}},
\]
as shown in Theorem~\ref{thm:informal_protocol}, and the formal statement can be found in Theorem~\ref{thm:w_sparse_partial_match_protocol}.

When $w = c \log n$, we can set $\epsilon = n^{-1/(100c)}$ and $\delta = n^{-0.2}$, and thus obtain $c_a, c_b, c_m \leq 0.1 \log n$. Simulation of this communication protocol gives our main result Theorem \ref{thm_main_informal}: a data structure for partial match with space $s = \poly(n)$ and query time %
\begin{align*}
t_q = &~ 2^{c_a+c_m} + 2^{c_m} \cdot \delta n + \epsilon n \\
= &~ n^{0.2} + n^{0.1} \cdot n^{-0.2} \cdot n + n^{1-1/(100c)} = O(n^{1-1/\Theta(c)}).
\end{align*}

We build this communication protocol in three phases:
\paragraph{Phase 1 (Reduction from sparse partial match to sparse set disjointness):} We first reduce the $w$-sparse partial match problem $\PM_{d,w}$ (where $y$ has at most $w$ number of $\star$'s) to the $w$-sparse set disjointness problem $\mathsf{Disj}_{d,w}$ (where $y$ has at most $w$ number of $0$'s). We cannot use the standard reduction from partial match to set containment / set disjointness since it doesn't preserve sparsity. 

We propose a new reduction by a communication protocol. We sample from the distribution $\lambda$ to find any $X$ such that $X$ matches $y$, and if we can't find such a string $X$ then it means the error probability is very small. We then compute $x' = x \oplus X$. Note that $x$ matches $y$ if and only $x$ agrees with $y$ and hence $X$ on all non-$\star$ coordinates. So we have that $x$ matches $y$ if and only if the ones of $x'$ is a subset of the $\star$'s of $y$. This sparse set containment problem is equivalent to the sparse set disjointness problem.

\paragraph{Phase 2 (Sparse set disjointness protocol with communication $\sqrt{w \log \frac{1}{\epsilon}} \cdot \log d$):}
We present the protocol as a set containment protocol since it's easier to describe. We follow a similar idea as the $O(\sqrt{d} \log d)$-protocol of \cite{bfs86}. Repeat for at most $\sqrt{\frac{w}{\log(1/\epsilon)}}$ iterations:
\begin{itemize}
    \item Alice checks if her set $|x| < \sqrt{w \log \frac{1}{\epsilon}}$, if so, she sends it to Bob using $\sqrt{w \log \frac{1}{\epsilon}} \cdot \log d$ bits.
    \item Otherwise, Alice and Bob together use public randomness to sample $t=\frac{1}{\epsilon}$ sets $X_1, \cdots, X_t$ from the distribution $\lambda$ conditioned on $|X| \geq \sqrt{w \log \frac{1}{\epsilon}}$.
    \begin{itemize}
    \item If none of the $X_i$'s satisfy $X \subseteq y$, then the protocol outputs $0$.
    \item Otherwise, Bob send the index $i^*$ such that $X_{i^*} \subseteq y$ using $\log(t)$ bits. Alice update $x \leftarrow x \backslash X_{i^*}$ and Bob update $y \leftarrow y \backslash X_{i^*}$.
    \end{itemize}
\end{itemize}
There are at most $\sqrt{\frac{w}{\log(1/\epsilon)}}$ iterations since initially the size of $y$ is $w$, and we reduce the size of $y$ by $\sqrt{w \log \frac{1}{\epsilon}}$ in each iteration. The total communication of Alice is $\sqrt{w \log \frac{1}{\epsilon}} \cdot \log d$, and the total communication of Bob is $\sqrt{\frac{w}{\log(1/\epsilon)}} \cdot \log t \leq \sqrt{w \log \frac{1}{\epsilon}}$.

Finally, we note that this protocol has false negative error $\epsilon$ since it outputs $0$ when no set $X_i$ is a subset of $y$. However, our framework requires the protocol to have false positive error instead. To address this, our modified protocol attempts to find a {\em certificate} for a negative answer (an intersecting coordinate). Specifically, in the protocol we check whether there exists a set $X_{i^*}$ such that $|X_{i^*}\backslash y| \leq \log(\frac{1}{\epsilon})$, i.e., $X_{i}$ is almost a subset of $y$. 
\begin{itemize}
    \item If there exists such a set, then Bob sends the set $|X_{i^*}\backslash y|$ together with the index $i^*$ using $\log(\frac{1}{\epsilon}) \log d$ bits, so this is still within his budget.
    \item If there doesn't exist such a set, then with probability $1-\epsilon$ the input $x$ satisfies that $|x\backslash y| > \log(\frac{1}{\epsilon})$. The protocol then randomly samples a subset $S \subset [d]$ of size $d/2$ and recurse to $x_S$ and $y_S$. With probability $1-\epsilon$, one of the coordinate in $x\backslash y$ is sampled into $S$.
\end{itemize}

\paragraph{Phase 3 (Non-determinism to remove $\log d$):}
Our protocol in Phase 2 has communication cost $\sqrt{w \log(\frac{1}{\epsilon})} \log d$ when sending a small set of size $s = \sqrt{w \log(\frac{1}{\epsilon})}$. We introduce a know-all player Merlin to send over this small set more efficiently. Merlin sends a \emph{relative encoding} of the smaller set as a \emph{subset} of the larger set $y$ using $\log(\binom{w}{s}) = O(\sqrt{w \log(\frac{1}{\epsilon})}\log c)$ bits. Alice and Bob then implement the standard equality protocol using public randomness to check if Merlin's message is correct, and this takes extra $\log(\frac{1}{\delta}) \leq 0.2 \log n$ bits of communication.

\subsection{Further Discussion and Related Work}

\paragraph{Algorithm design via communication complexity.}
While non-standard, our work is not the first to use communication complexity for algorithm design. One of the earliest examples is implicitly in \cite{KOR00}, who showed that a {\em sketching} (or one-way communication) protocol for problem $f$, of size $s$, would imply a data structure with space $n^{O(s)}$. Our work could be seen as extending that result to more powerful communication models, and, more crucially, over distributional protocols, which may have dramatically lower communication complexity as is the case for set disjointness. Also, very recently, independent work \cite{bao2024averagedistortionsketching} designs sketches under product distributions for the problem of $\ell_p$ distance estimation, which can be used to design new data structures for approximate NNS under
$\ell_p$.

In \cite{PW10}, \Pat and Williams 
showed that a (computationally efficient) protocol with sublinear $o(n)$ communication for 3-party set-disjointness in the \emph{number-on-forehead} (NOF) model would refute SETH. This result can be viewed, in a counter-positive manner, as an $\Omega(n)$ conditional lower
bound on the randomized communication complexity of 3NOF disjointness.
In a similar spirit, the core of the  breakthrough line of work on ``distributed PCPs'' \cite{ARW17, chen18} is designing a fine-grained algorithm for (offline) OV through efficient \emph{Merlin-Arthur} protocol for set-disjointness.   
However, besides the disparity of communication models used, these 
 result differ from ours in that they are both used in the \emph{offline} setting, as opposed to yielding a new online algorithm. 

Finally, \cite{chen2019classical} show that one can use AM (and other) protocols to obtain {\em offline} algorithms via the polynomial method. 
A natural question is whether one can similarly build data structures from AM protocols. We find this to be unlikely as, together with our list-of-points data structure lower bound, it would imply AM lower bounds, a central open question in communication complexity. Note that UAM is precisely the frontier where we have tools to prove lower bounds \cite{goos2015zero}.
 
\paragraph{Data-dependent algorithms, and approximate similarity search.}
The framework fundamentally yields \emph{data-dependent} data structures from communication protocols for a product distribution.
Until recently, the ``gold-standard''  indexing approach for high-dimensional online pattern matching problems was to use oblivious randomized space partitions, i.e., a generic (randomized) mapping $h: \{0,1\}^d \to S$ independent of the specific dataset, where the query algorithm (decision tree) goes through all data points that maps to the same element in $S$ as the query. The classic examples of such data-independent algorithm include locality-sensitive hashing (LSH) and filtering (LSF) for near-neighbor search (NNS, \cite{IM98}), and the ``chosen-path'' algorithm  \cite{CP17} and supermajority-based partitions \cite{ak20} for partial match. 
Data independent algorithms often have better guarantees on \emph{average-case} inputs --- where the data points are generated i.i.d.~from the same distribution --- even if their performance can deteriorate on worst-case datasets. (This is often the case for the offline problems as well: e.g., see \cite{Valiant12} for Hamming nearest neighbor search, and \cite{alman2025average} for OV.)

This observation led to the development of  \emph{data dependent} data structures, 
which use an \emph{adaptive} mapping $h$ at preprocessing time, where $h$ itself depends on the dataset \cite{AINR-subLSH,AR-optimal,annrw18-spectralGaps,charikar2020kernel}. Data-dependent data structures are usually composed of two components. One component is a \emph{worst-case to average-case} reduction, where
the high-level idea is to remove some ``structured'' parts of the dataset. The second component is to apply the aforementioned data-independent data structures that work better in the average-case. Implicitly, data-dependent algorithms previously have been used in \cite{Ind98b} as well as in the
aforementioned  data structure of \cite{cip02, chan19kd} for PM.

Our framework naturally yields data-dependent data structures, as the underlying communication protocols depend on the distribution of the points in the dataset. 
The closest (in spirit) work to ours is the work of \cite{annrw18-spectralGaps} who devised a framework for building data-dependent data structures for arbitrary approximate NNS problems. They introduce a key geometric notion of cutting modulus: defined from a graph $G$ over all possible points, with edges only on ``matching'' pairs (where $f(x,y)=1$), and a distribution $\mu(x,y)$ over edges. This framework can also be seen as requiring (small) distributional complexity, and indeed uses multiplicative weights update (MWU) to reduce worst-case to the distributional case (see also \cite{kush2021near,andoni2021average} who use a related notion, of average embedding, without explicit use of MWU). That said, their framework is quite different from ours: cutting modulus/average embedding is inherently a (geo)metric notion, for approximation problems, and eventually gives an LSH-like data structure (though see follow-up in~\cite{andoni2021approximate}). Our framework is for arbitrary predicates $f$ and the resulting data structure is not LSH-like.

We also think that our framework is conceptually simpler, as it avoids any  explicit reduction from worst-case to average-case (e.g.,  ``removal'' of structured parts of the dataset). Instead, the framework leverages efficient communication protocols under product distributions; thus any special structure of the dataset is directly exploited by the communication protocols. 

We also remark that \cite{ak20} studied the related, approximate version of the set containment, where one is to distinguish whether a query set $y\subset x_i$, for $x_i\in X_n$, or $y\cap x_i<\alpha \cdot |y|$ for some $\alpha < 1$. Their algorithm can be used to solve the exact version of the PM problem by setting $\alpha=1-1/d$, but that results into query time of the sort of $n^{1-o(1)}$ even for, say, $d=2\log n$.\footnote{The precise dependence requires solving a certain optimization, which is non-immediate.}%

\paragraph{Organization of the paper}
In Section~\ref{sec:preliminary} we present formal definitions of the data structure problems and our communication model. 

Our main results are shown in Sections~\ref{sec:w_sparse_protocols}, \ref{sec:w_sparse_reduction_cc_to_ds}, and \ref{sec:partial_match_ds}. In Section~\ref{sec:w_sparse_protocols} we present our protocols in a bottom-up fashion. In Section~\ref{sec:w_sparse_reduction_cc_to_ds} we present the framework to reduce from communication protocols to data structures.
In Section~\ref{sec:partial_match_ds} we use the results from the previous two sections to prove our main data structure result for $\PM_{n,c\log n}$ with query time $n^{1-\tilde{\Theta}(1/c \log^2 c)}$ and space $n^{1.1}$. %

In Section~\ref{sec:lop_lower_bound} we show a lower bound in the list-of-points model, asserting a query time of $\Omega( n^{1-1/\sqrt{c}})$ for any polynomial space data structure. 

Finally, in Appendix~\ref{sec:tight_communication_protocol} we show that the communication complexity $\sqrt{d \log(\frac{1}{\epsilon})}$ is tight for $\mathsf{Disj}_{d}$ with $\epsilon$ error under product distributions, a result of independent interest in communication complexity.

\section{Preliminary}\label{sec:preliminary}
\paragraph{Notations.}
For any integers $n \geq n' > 0$, we define $[n]=\{1,2,\cdots,n\}$, and we define $[n':n] = \{n',n'+1, \cdots, n\}$.

For any vector $x \in D^d$ where $D$ is any domain, we use $x_i$ to denote the $i$-th entry of $x$. For any subset $S \subseteq [d]$, we use $x|_S \in D^{|S|}$ to denote the substring of $x$ with coordinates in $S$.

For any vector $x \in \{0,1\}^d$, in this paper we interchangeably view $x$ as a boolean string of length $d$ and a subset of $[d]$. We use $|x|$ to denote the size of $x$, or equivalently, the number of $1$'s in $x$.

For any two vectors $x, y \in \{0,1\}^d$, we use $\oplus$ to denote the coordinate-wise XOR function: $(x \oplus y)_i = x_i \oplus y_i$.

We use $1_{\text{clause}}$ to denote the function that is $1$ if clause is true, and $0$ otherwise.

For any distribution $\lambda$ and any boolean function $f(\cdot)$, we use the notation $\lambda|_{f}$ to denote the conditional distribution of a random variable $X \sim \lambda$ given that $f(X) = 1$.

\subsection{Data structure problems}
We consider the following two data structure problems. %
\paragraph{$w$-sparse partial match.}
Preprocess a dataset of $n$ points $x_1,\ldots, x_n \in \{0,1\}^d$ into a data structure of size $s$, such that given a query $y \in \{0,1,\star\}^d$ with at most $w = c \log n$ number of $\star$'s, it can quickly report whether there exists a data point $x_i$ that matches $y$ on all non-$\star$ coordinates. 

\paragraph{$w$-sparse subset query.}
Preprocess a dataset of $n$ points $x_1,\ldots, x_n \in \{0,1\}^d$ into a data structure of size $s$, such that given a query $y \in \{0,1\}^d$ with at most $w = c \log n$ number of 1's, it can quickly report whether there exists a data point $x_i$ such that $x_i \subseteq y$. 
\\\\
We characterize the data structure by the following three quantities:
\begin{itemize}
    \item {\bf Preprocessing time:} time required to build the data structure $\textsc{DS}$ in the preprocessing phase. 
    \item {\bf Space:} space required to store the data structure $\textsc{DS}$. 
    \item {\bf Query time:} time to answer one query into DS.
\end{itemize}

\subsection{Functions to be solved by communication protocols}
We define the two functions that correspond to the partial match and subset query\footnote{This is also equivalently called set containment.} problems.

For any $x \in \{0,1\}^d$ and any $y \in \{0,1\}^d$, we define the subset query function
\begin{align*}
f_{\mathrm{sq}}(x,y) = 1_{x \subseteq y}.
\end{align*}

For any $x \in \{0,1\}^d$ and any $y \in \{0,1,\star\}^d$, we define the partial match function as
\begin{align*}
f_{\mathrm{pm}}(x,y) = \land_{i=1}^d (x_i = y_i \lor y_i = \star).
\end{align*}

\subsection{Communication model}\label{subsec_CC_model} 
In this section we define a special communication model that will be used throughout the following sections. The special model has stronger guarantees than the standard communication model for product distributions.

Before defining our special communication model, let us first recall the standard model of distributional communication complexity under product distributions:
\begin{definition}[Product-distributional communication model]\label{def:standard_communication_model}
Let $f: \X \times \Y \to \{0,1\}$ be the target function. We say a protocol $\pi$ solves $f$ over any product distribution if $\pi$ works as follows:
\begin{itemize}
    \item Let $\lambda, \rho$ be two arbitrary distributions over $\X$ and $\Y$. Alice and Bob both know $\lambda$ and $\rho$. Alice is given a private input $x \sim \lambda$, and Bob is given an independently sampled private input  $y \sim \rho$.
    \item Alice and Bob are allowed to use public randomness, denoted as $R$. 
\end{itemize}
We use $\pi(\lambda, \rho, x, y, R)$ to denote the output of the protocol when given distributions $\lambda$ and $\rho$ and inputs $x$ and $y$, and using public randomness $R$.\footnote{A standard averaging argument implies that w.l.o.g.~$\pi$'s randomness can be fixed under $\lambda\times \rho$, but this will no longer be true in our variant model below, hence we keep this feature.} 

We say the protocol has $\epsilon$ error if $\forall \lambda, \rho$, $\Pr_{x \sim \lambda, y \sim \rho, R}[\pi(\lambda, \rho, x, y, R) \neq f(x,y)] \leq \epsilon$. Further, we say the protocol has \underline{type I} error (false positive) if the protocol could be incorrect only when $f(x,y) = 0$, and we say the protocol has \underline{type II} error (false negative) if the protocol could be incorrect only when $f(x,y) = 1$.
\end{definition}

Our special communication model differs from the standard product model in four aspects. We list the differences and provide a short explanation of why we need these stronger definitions.
\begin{enumerate}
    \item Alice's input $x$ is generated from $\lambda$, but Bob's input $y$ is a worst-case input. (Crucially, $x$ does not depend on $y$).
    \begin{itemize}
        \item[$-$] In our reduction, the query algorithm simulates Bob's actions (messages), and the query itself is viewed as the input $y$ to Bob. Hence, it is important that the protocol works for \emph{any} input (query) $y$.
    \end{itemize}
    \item Bob does not know $\lambda$, i.e., Bob cannot use $\lambda$ during the protocol.
    \begin{itemize}
        \item[$-$] In our reduction, we define $\lambda$ to be the uniform distribution over the dataset $\mathcal{D} = \{x^{(1)}, \cdots, x^{(n)}\}$. Bob (i.e., the query) has no knowledge about the dataset and hence $\lambda$.
    \end{itemize}
    \item We introduce a third party Carol to model the public randomness. Carol writes to a public channel that is visible to both Alice and Bob.
    \begin{itemize}
        \item[$-$] We introduce Carol for simplicity of presentation when proving the reduction. The length of the messages of Carol naturally represents the description complexity of the protocol.
        
        One may argue that in the standard communication model, we can always fix the public randomness in advance and get a deterministic protocol. However, in our special model, Bob has a worst-case input $y$, and a fixed public randomness could only work for one $y$. Thus we still need to keep the randomness in our model. 
    \end{itemize}
    \item We introduce an ``oracle prover'' party Merlin that is similar to the Merlin in the standard Merlin-Arthur protocols.
    \begin{itemize}
        \item[$-$] Introducing Merlin makes our model more powerful as now we can design protocols with the advice from Merlin. We will show that we can simulate Merlin's messages efficiently in our data structure compiler.
    \end{itemize}
\end{enumerate}

We define this special type of communication protocols formally in the following definition.

\begin{definition}[Specialized product-distributional communication model]\label{def:special_communication_model}
Let $f: \X \times \Y \to \{0,1\}$ be the target function. There are four parties: Alice, Bob, Carol, and Merlin. We say a protocol $\pi$ solves $f$ in the special communication model if $\pi$ works as follows:
\begin{itemize}
\item Let $\lambda$ be an arbitrary distribution over $\X$, and let $y \in \Y$ be an arbitrary input. Alice knows $\lambda$ and her private input $x \sim \lambda$, Bob only knows his private input $y$, Carol only knows the distribution $\lambda$, and Merlin knows $\lambda$, and both $x$ and $y$. 
\item Carol is the only player that is allowed to use randomness, and we denote her randomness as $R_{\pub}$ (public randomness) and $R_{\pri}$ (private randomness), where $R_{\pub}$ is visible to Merlin, and $R_{\pri}$ is hidden from Merlin. Carol and Merlin know all the messages sent between Alice and Bob. 
\end{itemize}

We use $\pi(\lambda, x, y, m, R_{\pub}, R_{\pri})$ to denote the output of the protocol when given a distribution $\lambda$, inputs $x$ and $y$, Merlin's advice $m$, and randomness $R_{\pub}$ and $R_{\pri}$.

For any $\lambda, x, y, R_{\pub}$, there is a unique special advice $m^*(\lambda, x, y, R_{\pub})$ from Merlin, and we say the protocol has $\epsilon$ error and $\delta$ soundness if $\forall \lambda$ and $\forall y \in \{0,1\}^d$ the protocol has the following guarantees:
\begin{itemize}
\item {\bf Error under unique special advice.} The protocol guarantees that 
\[
\Pr_{x \sim \lambda, R_{\pub}, R_{\pri}}\big[\pi\big(\lambda, x, y, m^*(\lambda, x, y, R_{\pub}), R_{\pub}, R_{\pri} \big) \neq f(x,y)\big] \leq \epsilon.
\]
Further, we say the protocol has \underline{type I} error (false positive) if the protocol could be incorrect under the special advice only when $f(x,y) = 0$, and we say the protocol has \underline{type II} error (false negative) if the protocol could be incorrect under the special advice only when $f(x,y) = 1$.
\item {\bf Soundness.} For any $x \sim \lambda$ and any $R_{\pub}$, for any $m \neq m^*(\lambda, x, y, R_{\pub})$, the protocol guarantees that $\Pr_{R_{\pri}}\big[\pi\big(\lambda, x, y, m, R_{\pub}, R_{\pri} \big) \neq 0\big] \leq \delta$.
\end{itemize}

We say a protocol $\pi$ has $(c_a, c_b, c_c, c_m)$-cost if during the protocol, in total Alice sends Bob $c_a$ bits, Bob sends Alice $c_b$ bits, Carol sends Alice and Bob $c_c$ bits, and Merlin sends Alice and Bob $c_m$ bits. Further, we say the protocol $\pi$ has $(t_a, t_b, t_c)$-time if Alice, Bob, and Carol each use $t_a$, $t_b$ and $t_c$ time to construct all of their messages during the protocol.
\end{definition}

We also define a specialized worst-case communication model where the only difference with the above definition is Alice also has a worst-case input $x$.

\section{Protocols}\label{sec:w_sparse_protocols}
In this section we design our communication protocol for solving the bounded Partial Match problem. 
Our protocol consists of three phases, as outlined in Section~\ref{sec:tech_overview_protocols}. We present the phases in a bottom-up fashion, culminating in the main protocol result of Theorem~\ref{thm:w_sparse_partial_match_protocol}.

\subsection{Phase 3: Base case protocol with advice from Merlin}

In this section we present a protocol $\pi_{\mathrm{base}}$ for solving the sparse partial match problem and the sparse subset query problem. This protocol crucially uses the advice from Merlin. In later protocols we will use this protocol as a subroutine.

The protocol $\pi_{\mathrm{base}, f, d, z, w, \delta}(x, y, m, R_{\pri})$ is in our specialized worst-case communication model, and it has the following parameters:
\begin{itemize}
    \item $f \in \{f_{\mathrm{pm}}, f_{\mathrm{sq}}\}$: The function that we want to solve, which is either the partial match function or the subset query function.
    \item $d$: The dimension of the inputs $x$ and $y$.
    \item $z, w \in [d]$: The sparsity parameters for the inputs.
    \item $\delta \in (0,0.1)$: The error parameter.
\end{itemize}
And the inputs are as follows: 
\begin{itemize}
    \item Alice's input $x$ which is any string in $\{0,1\}^d$ with at most $z$ ones.
    \item Bob's input $y$ which is any string in $\{0,1,\star\}^d$ with at most $w$ stars when $f = f_{\mathrm{pm}}$, and is any string in $\{0,1\}^d$ with at most $w$ ones when $f = f_{\mathrm{sq}}$.
    \item Merlin's advice $m$.
    \item The private randomness $R_{\pri}$ which is hidden from Merlin. 
\end{itemize}

\paragraph{Protocol.} Let $t = \log(\frac{1}{\delta})$. 
\begin{enumerate}
\item Alice and Bob output $0$ and abort the protocol if $z > w$.
\item Merlin sends a message $m$ to Alice and Bob, claiming that $m$ equals to the unique special advice $m^*(x,y)$ which is defined as follows:
\begin{itemize}
\item If $f = f_{\mathrm{pm}}$, then $m^*(x,y) := x_I$ where $I:= \{i \mid y_i = \star\}$.
\item If $f = f_{\mathrm{sq}}$, then define $m^*(x,y)$ to be the index of $x\cap y$ in the lexicographically ordered list of all subsets of 
$y$ of size at most $z$.
\end{itemize}
\item\label{step:ov_y} Bob constructs $\ov{y} \in \{0,1\}^d$ as follows:
\begin{itemize}
\item If $f = f_{\mathrm{pm}}$, then Bob replaces each $\star$ entry in $y$ with the corresponding value from the advice string $m$ to obtain $\ov{y}$.
\item If $f = f_{\mathrm{sq}}$, then Bob views $m$ as an index of the lexicographically ordered list of all subsets of 
$y$ of size at most $z$, and sets $\ov{y}$ as the $m$-th subset in that list.
\end{itemize}
\item Carol uses $R_{\pri}$ (randomness that is hidden from Merlin) to generate $t$ independent random strings $r^{(1)}, \dots, r^{(t)} \in \{0,1\}^d $, with entries chosen independently and uniformly from $\{0,1\}$. Carol sends these strings to Alice and Bob.
\item Alice computes a vector $a \in \{0,1\}^t$ such that for each $i \in [t]$, $a_i = \sum_{j=1}^d x_j \cdot r^{(i)}_j \pmod{2}$. Similarly Bob computes a vector $b \in \{0,1\}^t$ such that for each $i \in [t]$, $b_i = \sum_{j=1}^d \ov{y}_j \cdot r^{(i)}_j \pmod{2}$. Alice sends $a$ to Bob, and Bob sends $b$ to Alice.
\item Alice and Bob checks if $a = b$. If $a = b$, then they output $1$, and if $a \neq b$, then they output $0$.
\end{enumerate}
\begin{lemma}[Base case protocol]\label{lem:base_case_protocol}
For any function $f \in \{f_{\mathrm{pm}}, f_{\mathrm{sq}}\}$, any dimension $d$, any sparsity parameters $z, w \in [d]$, %
and any error parameter $\delta \in (0,0.1)$, the protocol $\pi_{\text{base}, f, d, z, w, \delta}$ solves the function $f$ with the following guarantees:
\begin{enumerate}
\item When Merlin's message is the unique special message, the protocol has $\delta$ false positive error.
\item The protocol has $\delta$ soundness. 
\item The communication cost of Merlin is bounded by $c_m \leq w$ if $f = f_{\mathrm{pm}}$, and $c_m \leq \log z + z \log(\frac{ew}{z})$ if $f = f_{\mathrm{sq}}$.
\item The communication cost of Alice and Bob are both bounded by $c_a, c_b \leq \log(\frac{1}{\delta})$.
The time complexity of both players are bounded by $t_a, t_b \leq O(d \cdot \log(\frac{1}{\delta}))$.
\item The communication cost of Carol is bounded by $c_c \leq d \cdot \log(\frac{1}{\delta})$.
The time complexity of Carol is bounded by $t_c \leq O(d \cdot \log(\frac{1}{\delta}))$.
\end{enumerate}
\end{lemma}
\begin{proof}
{\bf Error under unique special advice.}
Assume that Merlin's message is $m = m^*(x,y)$. We first prove that $\ov{y} = x$ if and only if $f(x,y) = 1$. We prove this by considering the two cases of $f$:
\begin{itemize}
\item If $f = f_{\mathrm{pm}}$, then $m^*(x,y) = x_I$, where $I$ is the set of coordinates where $y$ is $\star$. By the definition of $\ov{y}$ in Step~\ref{step:ov_y}, we have $\ov{y}_I = m^*(x,y) = x_I$, and so $\ov{y} = x$ if and only if $x$ matches $y$.
\item If $f = f_{\mathrm{sq}}$, then $m^*(x,y)$ is the index of the set $x \cap y$ as a subset of $y$, and by the definition of $\ov{y}$ in Step~\ref{step:ov_y}, we have $\ov{y} = x \cap y$, and so $\ov{y} = x$ if and only if $x$ is a subset of $y$.
\end{itemize}
Next we compute the false negative and false positive errors when $m = m^*(x,y)$.
\begin{itemize}
\item If $f(x, y) = 1$, then we have $\ov{y} = x$, and so each $a_i = \sum_{j=1}^d x_j \cdot r^{(i)}_j \pmod{2} = \sum_{j=1}^d \ov{y}_j \cdot r^{(i)}_j \pmod{2} = b_i$, so the protocol outputs $1$, and there is no false negative error.
\item If $f(x,y) = 0$, then we have $\ov{y} \neq x$, and so for each $i \in [t]$ we have $\Pr[a_i = b_i] \leq 1/2$, so with probability $\geq 1-2^{-t} = 1-\delta$ there exists a coordinate $i$ where $a_i \neq b_i$ and the protocol outputs $0$. So the false positive error is at most $\delta$.
\end{itemize}

{\bf Soundness.} We first prove that when Merlin's message $m \neq m^*(x,y)$, then we must have $\ov{y} \neq x$. We again consider the two cases of $f$:
\begin{itemize}
\item If $f = f_{\mathrm{pm}}$, then when $m \neq m^*(x,y) = x_I$, there must exist a coordinate $i^*\in I$ where $\ov{y}_{i^*} = m_{i^*} \neq x_{i^*}$. 
\item If $f = f_{\mathrm{sq}}$, then when $m \neq m^*(x,y)$, $m$ doesn't encode $x \cap y$, so the decoded $\ov{y} \neq x \cap y$, so there must exist a coordinate $i^*\in y$ where $\ov{y}_{i^*} \neq x_{i^*}$. 
\end{itemize}
Then using the same argument as the $\ov{y} \neq x$ case above, $\Pr[a=b] \leq 2^{-t} = \delta$. So the soundness is at most $\delta$.

{\bf Cost.} Finally we bound the communication and time cost of each player.
\begin{itemize}
\item {\bf Merlin.} If $f = f_{\mathrm{pm}}$, Merlin sends a message of at most $|I| \leq w$ bits.

If $f = f_{\mathrm{sq}}$, since $|x| \leq z \leq w$ and $|y| \leq w$, we have that there are in total at most $\sum_{z'=1}^z \binom{w}{z'} \leq \sum_{z'=1}^z (\frac{ew}{z'})^{z'} \leq z \cdot (\frac{ew}{z})^z$ number of possible subsets of $y$ with size at most $z$. So the index of such subsets of $y$ can be encoded using $\log(z \cdot (\frac{ew}{z})^z) \leq \log z + z \log(\frac{ew}{z})$ number of bits, and this is the bound of the size of Merlin's message. 
\item {\bf Alice and Bob.} Alice's message $a$ and Bob's message $b$ both have length at most $t$.

Alice and Bob each spend time at most $O(d \cdot \log(\frac{1}{\delta}))$ to compute $a$ and $b.$

\item {\bf Carol.} Carol sends $r^{(1)}, \cdots, r^{(t)} \in \{0,1\}^d$ using at most $d \cdot \log(\frac{1}{\delta})$ bits.

Sampling these random strings takes time $O(d \cdot \log(\frac{1}{\delta}))$. \qedhere
\end{itemize}
\end{proof}

\subsection{Phase 2: Protocol for \texorpdfstring{$w$}{w}-sparse subset query}
In this section we present a protocol $\pi_{\mathrm{sq}, d, w, \epsilon,\delta}(\lambda, x, y, m, R_{\pub}, R_{\pri})$ in the specialized product-distributional communication model for the subset query function $f_{\mathrm{sq}}(x,y)$, where the inputs are the distribution $\lambda$, Alice's input $x$ which is sampled from $\lambda$, Bob's input $y$ which is any string in $\{0,1\}^d$ with at most $w$ ones, Merlin's advice $m$, and the randomness $R_{\pub}$ and $R_{\pri}$, where $R_{\pri}$ is hidden from Merlin.

\paragraph{Protocol.} Let $\ell = \sqrt{\frac{w}{\log(w/ \epsilon)}}$. Let $\epsilon' = \frac{\epsilon}{20 \ell}$, $\delta' = \frac{\delta}{10}$, and let $t = \frac{10}{\epsilon'} \log(\frac{1}{\epsilon'})$. Let $h = \log(\frac{1}{\epsilon'})$.

\begin{enumerate}
\item If $w \leq 100 \log(\frac{w}{\epsilon})$:
\begin{enumerate}
\item If $|x| > w$, then Alice outputs $0$ and informs Bob using $O(1)$ bits.
\item If $|x| \leq w$, then all players run the base case protocol $\pi_{\text{base}, f_{\mathrm{sq}}, d, w, w, \delta'}(x, y, m, R_{\pri})$, and Alice and Bob output the returned value of this subroutine.
\end{enumerate}
\item Let $U = [d]$ and let $w' = w$. Repeat for at most $2 \ell$ iterations:
\begin{enumerate}
\item If $|x| > w'$, then Alice outputs $0$ and informs Bob using $O(1)$ bits.
\item If $|x| \leq w/\ell$, then Alice uses $O(1)$ bits to inform Bob of this. All players run the base case protocol $\pi_{\text{base}, f_{\mathrm{sq}}, |U|, w/\ell, w, \delta'}(x, y, m, R_{\pri})$, and Alice and Bob output the returned value of this subroutine.
\item If $w/\ell < |x| \leq w'$, then Alice uses $O(1)$ bits to inform Bob of this. Carol uses randomness $R_{\pub}$ to sample $t$ independent sets $X_1, X_2, \cdots, X_t \sim \lambda|_{w/\ell < |X| \leq w'}$. Bob checks if there exists a set $X_i$ such that $|X_i \backslash y| \leq h$.
\begin{enumerate}
    \item If there is no such set, i.e., $|X_i \backslash y| > h$ for all $i \in [t]$. Then Carol uses randomness $R_{\pub}$ to sample $\log(\frac{10\ell}{\delta'})$ independent sets $S_1, S_2, \cdots, S_{\log(\frac{10\ell}{\delta'})} \subseteq U$ where each coordinate is sampled i.i.d.~with probability $1/2$. 
    \begin{enumerate}
        \item If there is no set $S_j$ such that $|y \cap S_j| \leq 2 w'/3$, then the protocol outputs $1$.
        \item Otherwise, Bob finds a set $S_{j^*}$ such that $|y \cap S_{j^*}| \leq 2w'/3$, and sends the index $j^*$ to Alice. Alice sets $x' \leftarrow x \cap S_{j^*}$, Carol sets $\lambda'$ to be the restriction of $\lambda$ to the domain $S_{j^*}$, and Bob sets $y' \leftarrow y \cap S_{j^*}$. The players run the protocol $\pi_{\text{sq}, |S_{j^*}|, 2w'/3, \epsilon/2, \delta'}(\lambda', x', y', m, R_{\pub}, R_{\pri})$ recursively and output its returned value.
    \end{enumerate}
    \item If there exist such sets, then Bob sets $i^*$ as the minimum index such that $|X_{i^*} \backslash y| \leq h$, and sends $i^*$ together with the index of $(X_{i^*} \backslash y)$ in the lexicographically ordered list of all subsets of $X_{i^*}$ to Alice. %
    \begin{enumerate}
    \item Alice computes $x \cap (X_{i^*} \backslash y)$. If this set is not empty, then Alice outputs $0$ and informs Bob.
    \item If $x \cap (X_{i^*} \backslash y) = \emptyset$. Then Alice updates her set $x \leftarrow x \backslash X_{i^*}$, Bob updates his set $y \leftarrow y \backslash X_{i^*}$, Carol update the distribution $\lambda$ to be the restriction of $\lambda$ to the domain $U \backslash X_{i^*}$, and all players update $U \leftarrow U \backslash X_{i^*}$, and $w' \leftarrow w' - |X_{i^*} \cap y|$.
    \end{enumerate}
\end{enumerate}
\end{enumerate}
\end{enumerate}

\begin{lemma}[$w$-sparse subset query protocol]\label{lem:w_sparse_subset_query_protocol}
For any $\delta \leq \epsilon \in (0,0.1)$, any dimension $d$, and any sparsity parameter $w \leq d$, the protocol $\pi_{\mathrm{sq}, d, w, \epsilon, \delta}$ solves the $w$-sparse subset query problem $f_{\mathrm{sq}}$ with the following guarantees:
\begin{enumerate}
\item The protocol has $\delta$ soundness.
\item When Merlin's message is the unique special message, the protocol has false positive error $\epsilon+\delta$.
\item The communication cost of Merlin is bounded by
\[
c_m \leq O\left(\sqrt{w \log(w/\epsilon)}\Big(1 + \log(\frac{w}{\log(w/\epsilon)}) \Big)\right).
\]
\item The communication cost of Alice is bounded by
\[
c_a \leq O\left(\sqrt{\frac{w}{\log(w / \epsilon)}} + \log(\frac{1}{\delta})\right),
\]
and the communication cost of Bob is bounded by
\[
c_b \leq O\left(\sqrt{w \log(w/\epsilon)} \cdot \log(\frac{w}{\log(w/\epsilon)})+ \log \log(\frac{w}{\delta}) \cdot \log w + \log(\frac{1}{\delta}) \right).
\]
The time complexity of both players are bounded by 
\[
t_a, t_b \leq O\left(\frac{d w}{\epsilon} (\log d)^3 \log(\frac{1}{\epsilon})^2 + d \log(\frac{1}{\delta})\right).
\]
\item The communication cost of Carol is bounded by 
\[
c_c \leq O\left(\frac{dw}{\epsilon} \log(\frac{w}{\delta}) (\log d)^3 \log(\frac{1}{\epsilon})^2 \right).
\]

The time complexity of Carol is bounded by 
\[
t_c \leq O\left((w / \epsilon) (\log d)^2 \log (1/\epsilon) \mathcal{T}_{\mathrm{Sample}} + d \mathcal{T}_{\mathrm{Update}} + \log(\frac{w}{\delta}) d \log d\right),
\]
where $\mathcal{T}_{\mathrm{Sample}}$ denotes the time to sample a set from distribution $\lambda$ conditioned on that the size of the sampled set is within a given range, and $\mathcal{T}_{\mathrm{Update}}$ denotes the time to update $\lambda$ to its restriction on a subset of $[d]$.
\end{enumerate}
\end{lemma}
\begin{proof}
We prove this lemma by induction on $w$. When $w \leq 10$, the protocol must enter Step~1, and by Lemma~\ref{lem:base_case_protocol} all claims in the lemma statement are satisfied. Next we consider any fixed $w > 10$, and prove the lemma for $w$ while assuming it is true for sparsity at most $w-1$.

If $w \leq 100 \log(\frac{w}{\epsilon})$, then we enter Step 1, and by Lemma~\ref{lem:base_case_protocol}, the protocol has $\delta$ soundness and $\delta$ false positive error, and the communication cost of Alice and Bob are bounded by $O(\log(\frac{1}{\delta}))$, the communication cost of Merlin is bounded by $O(\log w + w \log(e)) = O(w) \leq O(\sqrt{w \log(\frac{w}{\epsilon})})$. We can also easily check all other claims in the lemma statement are satisfied.

Now we can assume $w > 100 \log(\frac{w}{\epsilon})$ and we always enter Step 2. Next we prove some basic properties of the protocol. Observe that the protocol has $2 \ell$ iterations, and in each iteration it enters exactly one of the following six clauses: Step 2(a), 2(b), 2(c)(i)(A), 2(c)(i)(B), 2(c)(ii)(A), or 2(c)(ii)(B). The protocol will break from the loop and output an answer unless it enters Step 2(c)(ii)(B). In the proof, for any $i \in [0 : 2 \ell]$, we define notations $U^{(i)}$, $w'^{(i)}$, $x^{(i)}$, $y^{(i)}$, and $\lambda^{(i)}$ to be the corresponding $U$, $w'$, $x$, $y$, and $\lambda$ after $i$ iterations. Note that initially we have $U^{(0)} = [d]$, $w'^{(0)} = w$, $x^{(i)} = x$, $y^{(i)} = y$, $\lambda^{(i)} = \lambda$. For any $i \in [2\ell]$, we also define $X_{i^*}^{(i)}$ to be the $X_{i^*}$ chosen by Bob in Step 2(c)(ii) in the $i$-th iteration.
\begin{itemize}
\item {\bf Bounds of parameters.} Since $w > 100 \log(\frac{w}{\epsilon})$, we have $\ell = \sqrt{\frac{w}{\log(w/ \epsilon)}} > 10$, and that
\begin{align}\label{eq:h_is_small}
\frac{w}{\ell} = \sqrt{w \log(w/\epsilon)} 
> 10 \log(\frac{w}{\epsilon}) 
\geq 10 \log(\frac{20 \ell}{\epsilon}) = 10 \log(\frac{1}{\epsilon'}) = 10h,
\end{align}
where the third step follows from $w \geq 20 \ell$ because we have $w/\ell =\sqrt{w \log(w/\epsilon)} \geq 10 \log(w/\epsilon) \geq 20$ since $w>10$ and $\epsilon<0.1$.
\item {\bf The protocol always terminates with an output.} Consider the $i$-th iteration for any $i \in [2\ell]$. We either terminates with an output by entering Steps 2(a), 2(b), 2(c)(i)(A), 2(c)(i)(B), 2(c)(ii)(A), or we enter Step 2(c)(ii)(B). By the way that we generate $X_{i^*}^{(i)}$ in Step 2(c), we have $|X_{i^*}^{(i)}| > w/\ell$. Since we also have $|X_{i^*}^{(i)} \backslash y^{(i-1)}| \leq h$ in Step 2(c)(ii)(B), we have 
\[
|X_{i^*}^{(i)}\cap y^{(i-1)}| > w/\ell - h \geq 0.9 w/\ell,
\]
where the last step follows from Eq.~\eqref{eq:h_is_small}. So we have $w'^{(i)} = w'^{(i-1)} - |X_{i^*}^{(i)}\cap y^{(i-1)}| \leq w'^{(i-1)} - 0.9 w / \ell$. Thus after $2\ell$ iterations we must have $w'^{(2\ell)} \leq w - 2 \ell \cdot (0.9 w / \ell)< 0$. This means we must have stopped the protocol by entering either Step 2(a) or Step 2(b) before the $2 \ell$ iteration.

\item {\bf In each iteration $i$, $x^{(i)} \subseteq y^{(i)}$ if and only if $x \subseteq y$.} We prove this by induction. The base case for $i = 0$ is trivially true. Consider any iteration $i \geq 1$. We have $x^{(i)} = x^{(i-1)} \backslash X_{i^*}^{(i)}$ and $y^{(i)} = y^{(i-1)} \backslash X_{i^*}^{(i)}$, and we also have $x^{(i-1)} \cap (X_{i^*}^{(i)} \backslash y^{(i-1)}) = \emptyset$ since otherwise we would have stopped the protocol in Step 2(c)(ii)(A). So we have $x^{(i-1)} \subseteq y^{(i-1)}$ if and only if $x^{(i)} \subseteq y^{(i)}$.
\item {\bf In each iteration $i$, $y^{(i)}$ has at most $w'^{(i)}$ number of $1$'s.} This directly follows from $y^{(i)} = y^{(i-1)} \backslash X_{i^*}^{(i)}$ and $w'^{(i)} = w'^{(i-1)} - |X_{i^*}^{(i)}\cap y^{(i-1)}|$.
\item {\bf The outputs in Step 2(a) or Step 2(c)(ii)(A) are correct.} If the protocol enters Step 2(a) in iteration $i$, then this means $|x^{(i-1)}| > w'^{(i-1)}$, but we also have $|y^{(i-1)}| \leq w'^{(i-1)}$, so $x^{(i-1)}$ cannot be a subset of $y^{(i-1)}$, and this implies $x$ cannot be a subset of $y$, so it's correct to output $0$. 

If the protocol enters Step 2(c)(ii)(A) in iteration $i$, then this means there is an coordinate in $x^{(i-1)} \cap (X_{i^*}^{(i)} \backslash y^{(i-1)})$, so $x^{(i-1)}$ is not a subset of $y^{(i-1)}$, and this implies $x$ cannot be a subset of $y$, so it's correct to output $0$. 
\end{itemize}

{\bf Soundness.} We first define Merlin's unique special advice $m^*(\lambda, x, y, R_{\pub})$. Our protocol only uses Merlin's advice and randomness $R_{\pri}$ in the base case subset query subroutine $\pi_{\text{base}}$ in Steps 1(b) or 2(b) of the current protocol, or in the recursive calls in Step 2(c)(i)(B). Also note that during the protocol and its recursive calls, we run the base case subset query subroutine for at most once, because after running the subroutine we immediately output its returned value and terminate the protocol. So for any given $\lambda, x, y, R_{\pub}$, the transcript of the protocol is completely fixed except for the internal transcript of the base case subset query subroutine. This means given $\lambda, x, y, R_{\pub}$ we can identify the exact step and iteration in which the subroutine is called, along with the specific inputs for the subroutine. We define $m^*(\lambda, x, y, R_{\pub})$ to be the unique special Merlin's advice for this subroutine.

Next we consider when Merlin's message $m \neq m^*(\lambda, x, y, R_{\pub})$. There are two cases that the protocol could output $1$:
\begin{enumerate}
\item The protocol ends with a call to the base case subroutine $\pi_{\text{base}}$ in either Steps~1(b) and 2(b) or in the recursive calls. Then by Lemma~\ref{lem:base_case_protocol} we have that under any Merlin's advice $m$ that is not the unique special advice, the base case subroutine outputs $0$ with probability $1-\delta'$.
\item The protocol outputs $1$ in Step 2(c)(i)(A). This only happens when the $\log(\frac{10\ell}{\delta'})$ independent sets $S_j$'s all satisfy $|y \cap S_j| > 2w'/3$, and this happens with probability at most $2^{-\log(\frac{10\ell}{\delta'})} = \frac{\delta'}{10 \ell}$ in each iteration, so using Union bound over all $2\ell$ iterations, this happens with probability at most $\delta'/5$.
\end{enumerate}
Using union bound over these two cases, we have that the probability the protocol outputs $1$ under any Merlin's message $m \neq m^*(\lambda, x, y, R_{\pub})$ is at most $\delta' + \delta'/5 \leq \delta$.

{\bf Error under unique special advice.} Next we assume that Merlin's message is the unique special message $m^*(\lambda, x, y, R_{\pub})$ as defined in the above paragraph.

As shown in the beginning of the proof, the output of the protocol is always correct in Step 2(a) or 2(c)(ii)(A), so an error could only occur if the protocol outputs an answer in Step 2(b), 2(c)(i)(A), or 2(c)(i)(B). Next we bound the error of these three cases one by one.
\begin{enumerate}
\item {\bf Step 2(b).} In this step, we run the subroutine $\pi_{\text{base}, f_{\mathrm{sq}}, |U|, w/\ell, w, \delta'}(x, y, m, R_{\pri})$, and by Lemma~\ref{lem:base_case_protocol}, there is a false positive error of $\delta'$ when Merlin's message is the unique special message. Also note that we only enter this step once during all iterations, so the total error of this step is $\delta'$.
\item {\bf Step 2(c)(i)(A).} In this step we always output $1$, so it only incurs a false positive error. Same as the proof in the soundness section, in each iteration we could enter this step with probability at most $\frac{\delta'}{10 \ell}$, so over $2\ell$ iterations the probability that we ever enter this step is at most $\frac{\delta'}{5}$, and this is an upper bound of the error caused by this step.
\item {\bf Step 2(c)(i)(B).} In this step we recursively determine if $x \cap S_{j^*}$ is a subset of $y \cap S_{j^*}$. If $x \subseteq y$, then we must also have $x\cap S_{j^*}\subseteq y\cap S_{j^*}$, and since by the induction hypothesis the recursive call $\pi_{\text{sq}, |S_{j^*}|, 2w'/3, \epsilon/2, \delta'}(\lambda', x', y', m, R_{\pub}, R_{\pri})$ only has false positive error, this step is always correct when $x \subseteq y$. Next we bound the false positive error, i.e., we bound the probability that this step outputs $1$ when in fact $x$ is not a subset of $y$. We consider two cases:

If $x$ is not a subset of $y$, and $x\cap S_{j^*}$ is also not a subset of $y\cap S_{j^*}$. Then by the induction hypothesis the error of the recursive call $\pi_{\text{sq}, |S_{j^*}|, 2w'/3, \epsilon/2, \delta'}$ is at most $\epsilon/2+\delta'$. Also note that we only enter the recursive call once during all iterations, so there is no need to union bound over all $2\ell$ iterations.

If $x$ is not a subset of $y$, but $x\cap S_{j^*}$ is a subset of $y\cap S_{j^*}$. This could happen in the $i$-th iteration only if one of the following three bad events happens: (1) $|x^{(i)} \backslash y^{(i)}| > h$, but after sampling $S_{j^*}$ we fail to include any coordinate from $x^{(i)} \backslash y^{(i)}$, or (2) $|x^{(i)} \backslash y^{(i)}| \leq h$, $x\cap S_{j^*} \subseteq y\cap S_{j^*}$, and $\Pr_{X \sim \lambda^{(i)}|_{w/\ell < |X| \leq w'^{(i)}}}[|X \backslash y^{(i)}| \leq h] \leq \epsilon'$, or (3) $|x^{(i)} \backslash y^{(i)}| \leq h$, $x\cap S_{j^*} \subseteq y\cap S_{j^*}$, and $\Pr_{X \sim \lambda^{(i)}|_{w/\ell < |X| \leq w'^{(i)}}}[|X \backslash y^{(i)}| \leq h] \geq \epsilon'$. Next we consider these three cases one by one.
\begin{itemize}
    \item Case 1: If $|x^{(i)} \backslash y^{(i)}| > h$, but after sampling $S_{j^*}$ we have $x\cap S_{j^*} \subseteq y\cap S_{j^*}$.

    Since each coordinate is independently included in $S_{j^*}$ with probability $1/2$, the probability that no coordinate in $x^{(i)} \backslash y^{(i)}$ is included in $S_{j^*}$ is at most $(1/2)^{|x^{(i)} \backslash y^{(i)}|} \leq 2^{-h} = \epsilon'$. We also need to add the probability that $|y^{(i)} \cap S_{j^*}| > 2w'/3$, since then we would have entered Step 2(c)(i)(A) instead. Using Chernoff bound this probability is bounded by $e^{-\Theta(w')} \leq e^{-\Theta(w/\ell)} \leq 2^{-h} = \epsilon'$ since $w' \geq w/\ell > 10 h$. So by union bound the error in this case is at most $2 \epsilon'$.
    \item Case 2: If $\Pr_{X \sim \lambda^{(i)}|_{w/\ell < |X| \leq w'^{(i)}}}[|X \backslash y^{(i)}| \leq h] \leq \epsilon'$, yet $|x^{(i)} \backslash y^{(i)}| \leq h$.

    Since $x^{(i)}$ follows the distribution $\lambda^{(i)}|_{w/\ell < |X| \leq w'^{(i)}}$ if we enter Step 2(c), we have that $|x^{(i)} \backslash y^{(i)}| \leq h$ could only happen for $\epsilon'$ fraction of $x$ that enters Step 2(c).
    \item Case 3: If $\Pr_{X \sim \lambda^{(i)}|_{w/\ell < |X| \leq w'^{(i)}}}[|X \backslash y^{(i)}| \leq h] \geq \epsilon'$.

    We prove that the probability we enter Step 2(c)(i) in this case is small. Note that we only enter Step 2(c)(i) when $|X_j \backslash y^{(i)}| > h$ for all $j \in [t]$. Define $t$ Bernoulli random variables $Y_1, Y_2, \cdots, Y_t$ such that $Y_j = 1$ iff $|X_j \backslash y^{(i)}| \leq h$. $Y_1, \cdots, Y_t$ are independent because $X_1, \cdots, X_t$ are sampled independently. Let $Y = \sum_{j=1}^t Y_j$. We enter Step 2(c)(i) only when $Y = 0$. By the condition $\Pr_{X \sim \lambda^{(i)}|_{w/\ell < |X| \leq w'^{(i)}}}[|X \backslash y^{(i)}| \leq h] \geq \epsilon'$ we have $\Pr[Y_j = 1] \geq \epsilon'$. Thus $\E[Y] \geq t \epsilon' = 10 \log(1/\epsilon')$. Using Chernoff bound we have
    \[
    \Pr[Y = 0] < \Pr[Y < \E[Y]/2] < e^{-\E[y]/8} \leq \epsilon'.
    \]
\end{itemize}
Combining these three cases over all $2\ell$ iterations using the union bound, we have the error of this case is at most $2\ell \cdot (2 \epsilon' + \epsilon' + \epsilon') \leq \epsilon/2$.

Combining the two cases of whether or not $x \cap S_{j^*}$ is a subset of $y \cap S_{j^*}$, we have that the total error of Step 2(c)(i)(B) is at most $\epsilon + \delta'$.
\end{enumerate}
Note that all the errors above are false positive errors. Adding them up using the union bound, we have that the total error under the unique special advice is at most $\epsilon + \delta$.

{\bf Cost.} Finally we bound the communication cost and time complexity of each players. We only consider Step 2 since we already bounded the cost of Step 1.
\begin{itemize}
\item {\bf Merlin.} Merlin sends messages either in the subroutine in Step 2(b) or in the recursive call in Step 2(c)(i)(B). By Lemma~\ref{lem:base_case_protocol} and since we run the protocol with $z = w/\ell$, Merlin's message in Step 2(b) is bounded by 
\begin{align*}
c_m \leq \log z + z \log(\frac{ew}{z}) 
= &~ \log(\frac{w}{\ell}) + \frac{w}{\ell} \log(e \ell) \\
\leq &~ \sqrt{w \log(w/\epsilon)} \cdot \big(3 + \log(\frac{w}{\log(w/\epsilon)})\big).
\end{align*}
By the induction hypothesis, Merlin's message in Step 2(c)(i)(B) is bounded by 
\[
\sqrt{(2w/3) \log((4w/3)/\epsilon)} \cdot \big(3 + \log(\frac{2w/3}{\log((4w/3)/\epsilon)})\big) \leq \sqrt{w \log(w/\epsilon)} \cdot \big(3 + \log(\frac{w}{\log(w/\epsilon)})\big).
\]

Merlin's cost is bounded by the maximum of these two costs.

\item {\bf Alice:} In each iteration Alice sends $O(1)$ bits in the current iteration and continue to the next iteration. Apart from these costs, during all iterations, Alice sends one of the following two messages for once: (1) a message of at most $\log(\frac{1}{\delta'})$ bits (from Lemma~\ref{lem:base_case_protocol}) in the subroutine of Step 2(b), or (2) a message of at most $O\Big(\sqrt{\frac{2w/3}{\log\big((2w/3) / (\epsilon/2)\big)}} + \log(\frac{1}{\delta})\Big)$ bits (by the induction hypothesis) in the recursive call of Step 2(c)(i)(B). Thus the total number of bits that Alice sends is 
\begin{align*}
c_a = &~ O(\ell) + \max\Big\{\log(\frac{1}{\delta'}), O\Big(\sqrt{\frac{2w/3}{\log\big((2w/3) / (\epsilon/2)\big)}} + \log(\frac{1}{\delta})\Big) \Big\} \\
= &~ O(\sqrt{\frac{w}{\log(w/\epsilon)}}) + \max\Big\{\log(\frac{1}{\delta'}), O\Big(\sqrt{\frac{2w/3}{\log\big((2w/3) / (\epsilon/2)\big)}} + \log(\frac{1}{\delta})\Big) \Big\} \\
\leq &~ O\Big(\sqrt{\frac{w}{\log(w /\epsilon)}} + \log(\frac{1}{\delta})\Big).
\end{align*}

In each iteration, Alice needs to go over the $t$ sets $X_i$'s and update her set $x$, and this takes at most $O(t d \log d)$ time. She also run the subroutine in Step 2(b) or the recursive call in Step 2(c)(i)(B) at most once, and this takes an extra $O(d \log(\frac{1}{\delta}))$ time by Lemma~\ref{lem:base_case_protocol}. So the total time of Alice is $t_a \leq O(\ell t d \log d + d \log(\frac{1}{\delta})) = O((d \ell^2 / \epsilon) \log (\ell/\epsilon) \log d + d \log(\frac{1}{\delta})) \leq O((d w / \epsilon) (\log d)^3 \log(1/\epsilon)^2 + d \log(\frac{1}{\delta}))$.

\item {\bf Bob:} In each iteration, in Step 2(c)(ii) Bob sends the index $i^*$ using $\log t \leq O(\log(\ell/\epsilon)) \leq O(\log(w/\epsilon))$ bits together with the index of the subset $(X_{i^*} \backslash y)$ using $\log(\binom{w}{h}) \leq O(h \log(\frac{w}{h})) \leq O(\log(\frac{w}{\epsilon}) \log(\frac{w}{\log(w/\epsilon)}))$ bits, and then he continues to the next iteration. Bob also needs to send $\log \log(\frac{\ell}{\delta})$ bits in Step 2(c)(i)(B) when sending $j^*$ for every level of recursion, and this gives a total cost of $\log \log(\frac{\ell}{\delta}) \cdot \log w$ bits. Apart from these costs, during all iterations, Bob also sends one of the following two messages for once: (1) a message of at most $\log(\frac{1}{\delta'})$ bits (from Lemma~\ref{lem:base_case_protocol}) in the subroutine of Step 2(b), or (2) a message of at most $O\Big(\sqrt{(2w/3)\log\big(\frac{2w/3}{\epsilon/2}\big)}\cdot \log(\frac{2w/3}{\log(\frac{2w/3}{\epsilon/2})}) + \log(\frac{1}{\delta})\Big)$ bits (by the induction hypothesis) in the recursive call of Step 2(c)(i)(B). Thus the total number of bits that Bob sends is 
\begin{align*}
c_b = &~ \ell \cdot O(\log(w/\epsilon)) \log(\frac{w}{\log(w/\epsilon)}) + \log \log(\frac{\ell}{\delta}) \cdot \log w \\
&~ + \max\left\{\log(\frac{1}{\delta'}), O\Big(\sqrt{(2w/3)\log\big(\frac{2w/3}{\epsilon/2}\big)}\cdot \log(\frac{2w/3}{\log(\frac{2w/3}{\epsilon/2})}) + \log(\frac{1}{\delta})\Big) \right\} \\
\leq &~ O\left(\sqrt{w \log(w/\epsilon)} \cdot \log(\frac{w}{\log(w/\epsilon)})+ \log \log(\frac{w}{\delta}) \cdot \log w + \log(\frac{1}{\delta}) \right) 
\end{align*}
where the second step follows from $\ell = \frac{w}{\log(w/\epsilon)}$, and we don't lose an extra $\log w$ term by the recursion because the coefficient of the term $\sqrt{(2w/3)\log\big(\frac{2w/3}{\epsilon/2}\big)}\cdot \log(\frac{2w/3}{\log(\frac{2w/3}{\epsilon/2})})$ is less than 1, so the geometric series of all levels of recursions converges to a constant.

The total time of Bob is the same as Alice: $t_b \leq O((d w / \epsilon) (\log d)^3 \log(1/\epsilon)^2 + d \log(\frac{1}{\delta}))$.

\item {\bf Carol:} In each iteration Carol writes $t d \log d$ bits to the public channel in Step 2(c) and $\log(\frac{10\ell}{\delta'}) d \log d$ bits in Step 2(c)(i). Thus over $\ell$ iterations the total number of bit that Carol writes is
\[
c_c = O(\ell (t+\log(\frac{w}{\delta})) d \log d) \leq O(\frac{dw}{\epsilon} \log(\frac{w}{\delta}) (\log d)^3 \log(\frac{1}{\epsilon})^2 ).
\]

The total time of Carol is 
\begin{align*}
t_c \leq &~ O(\ell t \mathcal{T}_{\mathrm{Sample}} + \ell \mathcal{T}_{\mathrm{Update}} + \log(\frac{10\ell}{\delta'}) d \log d) \\
= &~ O((\ell^2 / \epsilon) \log (\ell/\epsilon) \mathcal{T}_{\mathrm{Sample}} + \ell \mathcal{T}_{\mathrm{Update}} + \log(\frac{w}{\delta}) d \log d) \\
\leq &~ O((w / \epsilon) (\log d)^2 \log (1/\epsilon) \mathcal{T}_{\mathrm{Sample}} + d \mathcal{T}_{\mathrm{Update}} + \log(\frac{w}{\delta}) d \log d), 
\end{align*}
where $\mathcal{T}_{\mathrm{Sample}}$ is the time to sample a set from distribution $\lambda|_{w/\ell<|X| \leq w'}$, $\mathcal{T}_{\mathrm{Update}}$ is the time to update $\lambda$ to its restriction on a subset of $[d]$. \qedhere
\end{itemize}
\end{proof}

\subsection{Phase 1 and Overall Protocol: Protocol for \texorpdfstring{$w$}{w}-sparse partial match}
In this section we present a protocol $\pi_{\mathrm{pm}, d, w, \epsilon,\delta}(\lambda, x, y, m, R_{\pub}, R_{\pri})$ in the specialized product-distributional communication model for the partial match function $f_{\mathrm{pm}}(x,y)$, where the inputs are the distribution $\lambda$, Alice's input $x \in \{0,1\}^d$ which is sampled from $\lambda$, Bob's input $y$ which is any string in $\{0,1,\star\}^d$ with at most $w$ stars, Merlin's advice $m$, and the randomness $R_{\pub}$ and $R_{\pri}$, where $R_{\pri}$ is hidden from Merlin. 

At a high-level, this protocol reduces the partial match problem to the subset query problem while maintaining the sparsity guarantee, and it uses the subset query protocols of Lemma~\ref{lem:base_case_protocol} and Lemma~\ref{lem:w_sparse_subset_query_protocol} in the previous sections as subroutines.

\paragraph{Protocol.} Let $t = \frac{100}{\epsilon} \log(\frac{10}{\epsilon})$, and let $h = \log(\frac{10}{\epsilon})$. %

\begin{enumerate}
\item If $w \leq 100 \log(\frac{w}{\epsilon})$, then all players run the base case protocol $\pi_{\text{base}, f_{\mathrm{pm}},d, w,w, \delta}(x, y, m, R_{\pri})$, and Alice and Bob output the returned value of this subroutine.
\item Carol uses randomness $R_{\pub}$ to sample $t$ sets $X_1, X_2, \cdots, X_t \sim \lambda$. Bob checks if there exists $X_i$ such that the number of unmatched coordinates of $X_i$ and $y$ is $\leq h$.
\begin{enumerate}
\item If there is no such set, then Carol uses randomness $R_{\pub}$ to sample $\log(\frac{10}{\delta})$ number of independent sets $S_1, S_2, \cdots, S_{\log(\frac{10}{\delta})} \subseteq [d]$ where each coordinate is sampled i.i.d.~with probability $1/2$. 
\begin{enumerate}
        \item If there is no set $S_j$ such that the number of $\star$'s in $y|_{S_j}$ is $\leq 2 w/3$, then the protocol outputs $1$.
        \item Otherwise, Bob finds a set $S_{j^*}$ such that the number of $\star$'s in $y|_{S_j}$ is $\leq 2 w/3$, and Bob sends the index $j^*$ to Alice. Alice sets $x' \leftarrow x|_{S_{j^*}}$, Carol sets $\lambda'$ to be the restriction of $\lambda$ to the domain $S_{j^*}$, and Bob sets $y' \leftarrow y|_{S_{j^*}}$. The players run the protocol $\pi_{\text{pm}, S_{j^*}, 2w/3, \epsilon/2, \delta/10}(\lambda', x', y', m, R_{\pub}, R_{\pri})$ recursively and output its returned value.
    \end{enumerate}
\item If there exist such sets, then Bob sets $i^*$ to be the minimum index such that the unmatched coordinates of $X_{i^*}$ and $y$ is $\leq h$, and sends $i^*$ to Alice using $\log t$ bits.
\begin{enumerate}
    \item Alice computes a string $x' = x \oplus X_{i^*} \in \{0,1\}^d$, i.e., for all $j \in [d]$, $x'_j = 1$ if $x_j \neq (X_{i^*})_j$ and $x'_j = 0$ if $x_j = (X_{i^*})_j$.
    \item Bob computes a string $y' \in \{0,1,\star\}^d$ such that for all $j \in [d]$, (1) $y'_j = \star$ if $y_j = \star$, (2) $y'_j = 1$ if $y_j \neq (X_{i^*})_j$ and $y_j \neq \star$, (3) $y'_j = 0$ if $y_j = (X_{i^*})_j$.
    \item If $|x'| > w + h$, then Alice outputs $0$ and informs Bob using $O(1)$ bits.
    \item If $|x'| \leq w + h$, then Alice and Bob run two subroutines of subset query: 
    \begin{enumerate}
        \item Check if the ones of $x'$ is a subset of the ones and $\star$'s of $y'$: run the protocol $\pi_{\mathrm{sq}, d, w+h, \epsilon/10, \delta/10}(\lambda, x', y'_{\star} \cup y'_{\text{one}}, m, R_{\pub}, R_{\pri})$, where $y'_{\star} := \{i \in [d] \mid y'_i = \star\}$ and $y'_{\text{one}} := \{i \in [d] \mid y'_i = 1\}$.
        \item Check if the ones of $y'$ is a subset of the ones of $x'$: run the base case protocol $\pi_{\text{base}, f_{\mathrm{sq}}, d, h, w+h, \delta/10}(y'_{\text{one}}, x', m, R_{\pri})$, where $y'_{\text{one}} := \{i \in [d] \mid y'_i = 1\}$.
        \item The protocol outputs $1$ if both subroutines return $1$, otherwise it outputs $0$.
    \end{enumerate}
\end{enumerate}
\end{enumerate}
\end{enumerate}

\begin{theorem}[$w$-sparse partial match protocol]\label{thm:w_sparse_partial_match_protocol}
For any $\delta \leq \epsilon \in (0,0.1)$, any dimension $d$, and any sparsity $w \leq d$, the protocol $\pi_{\mathrm{pm}, d, w, \epsilon, \delta}$ solves the $w$-sparse partial match problem $f_{\mathrm{pm}}$ with the following guarantees:
\begin{enumerate}
\item The protocol has $\delta$ soundness.
\item When Merlin's message is the unique special message, the protocol has false positive error $\epsilon+\delta$.
\item The communication cost of Merlin is bounded by
\[
c_m \leq O\left( \Big(\sqrt{w \log(\frac{w}{\epsilon})} + \log(\frac{w}{\epsilon}) \Big) \cdot \log(\frac{w}{\log(1/\epsilon)}) \right).
\]
\item The communication cost of Alice is bounded by
\[
c_a \leq O\left(\log(\frac{1}{\delta}) + \sqrt{\frac{w}{\log(1/\epsilon)}}\right),
\]
and the communication cost of Bob is bounded by
\[
c_b \leq O\left(\Big(\sqrt{w \log(\frac{w}{\epsilon})} + \log(\frac{w}{\epsilon})\Big) \cdot \log(\frac{w}{\log(1/\epsilon)}) + \log \log(\frac{w}{\delta}) \cdot \log\Big(w + \log(\frac{w}{\epsilon}) \Big) + \log(\frac{1}{\delta}) \right).
\]
The time complexity of both players are bounded by 
\[
t_a, t_b \leq O\left(\frac{d w^3}{\epsilon} \log(\frac{1}{\delta}) (\log d)^3 \log(\frac{w}{\epsilon})^3 \right).
\]
\item The communication cost of Carol is bounded by 
\[
c_c \leq O\left(\frac{d w^3}{\epsilon} \log(\frac{1}{\delta}) (\log d)^3 \log(\frac{w}{\epsilon})^3 \right).
\]

The time complexity of Carol is bounded by 
\[
t_c \leq O\left(\frac{d w^3}{\epsilon} \log(\frac{1}{\delta}) (\log d)^3 \log(\frac{w}{\epsilon})^3 \cdot \mathcal{T}_{\mathrm{Sample}} + d \cdot \mathcal{T}_{\mathrm{Update}} \right),
\]
where $\mathcal{T}_{\mathrm{Sample}}$ denotes the time to sample a set from distribution $\lambda$ conditioned on that the size of the sampled set is within a given range, and $\mathcal{T}_{\mathrm{Update}}$ denotes the time to update $\lambda$ to its restriction on a subset of $[d]$.
\end{enumerate}
\end{theorem}
\begin{proof}
We first prove some basic properties of the protocol before proving its error and cost guarantees. Note that our protocol follows one of the following three cases: it either enters Step 1 (the base case) and terminates, or enters Step 2(b) (a reduction to subset queries) and terminates, or proceeds to Step 2(a) where it recurses with updated parameters $2w/3$, $\epsilon/2$, $\delta/10$. Let $\ell$ denote the total recursion depth of the protocol. After $i$ levels of recursions, denote the parameters as $w^{(i)}, \epsilon^{(i)}, \delta^{(i)}$, and they become $w^{(i)} = (\frac{2}{3})^i \cdot w$, $\epsilon^{(i)} = (\frac{1}{2})^i \cdot \epsilon$, and $\delta^{(i)} = (\frac{1}{10})^i \cdot \delta$. If $i \geq 2 \log(\frac{w}{\log(1/\epsilon)})$, we have
\begin{align*}
w^{(i)} = (\frac{2}{3})^i \cdot w \leq (\frac{1}{2})^{\log(\frac{w}{\log(1/\epsilon)})} \cdot w = \log(1/\epsilon) \leq \log(\frac{2^i}{\epsilon}) = \log(\frac{1}{\epsilon^{(i)}}).
\end{align*}
So we must have entered the base case of Step 1 before reaching $i \geq 2 \log(\frac{w}{\log(1/\epsilon)})$ levels of recursions, and this means the total recursion depth is bounded by
\begin{align}\label{eq:ell_bound}
\ell \leq 2 \log(\frac{w}{\log(1/\epsilon)}).
\end{align}
Thus we can decompose the execution of the protocol into the following three parts:
\begin{enumerate}
    \item At most $\ell \leq 2 \log(\frac{w}{\log(1/\epsilon)})$ executions of the sampling processes of Step 2 and 2(a).
    
    Each sampling process of Step 2 samples at most $t' \leq \frac{100 \cdot 2^{\ell}}{\epsilon} \log(\frac{10 \cdot 2^{\ell}}{\epsilon})$ number of $X_i$'s, and each sampling process of Step 2(a) samples at most $\log(\frac{10^{\ell+1}}{\delta})$ number of $S_i$'s.
    \item At most one execution of the two subset query subroutines of Step 2(b).
    
    The two subset query subroutines are $\pi_{\mathrm{sq}, d, w'+h', \epsilon', \delta'}$ and $\pi_{\text{base}, f_{\mathrm{sq}}, d, h', w'+h', \delta'}$, where the parameters are $w' \leq w$, $h' \leq \log(\frac{10\cdot 2^{\ell}}{\epsilon})$, $\epsilon' \geq 2^{-\ell} \epsilon$, and $\delta' \geq 10^{-\ell} \delta$.
    \item At most one execution of the base case subroutine of Step 1.
    
    Where the base case subroutine is $\pi_{\text{base}, f_{\mathrm{pm}}, d, w', w', \delta'}$ where the parameters are $w' \leq 100 (\log(\frac{w}{\epsilon}) + \ell)$, and $\delta' \geq 10^{-\ell} \cdot \delta$.
\end{enumerate}
Also note that the protocol executes exactly one of Step 2(b) and Step 3.

{\bf Correctness of reduction to subset query.}
Next we prove the reduction from partial match to subset query of Step 2(b) is correct. We make the following claim:
\begin{align}\label{eq:correctness_pm_to_sq}
x \text{ matches }y \iff \text{(ones of $x'$) $\subseteq$ (ones and $\star$'s of $y'$), and (ones of $y'$) $\subseteq$ (ones of $x'$)},
\end{align}
where $x'$ and $y'$ are as defined in Step 2(b)(i) and 2(b)(ii). Let $H \subseteq [d]$ denote the unmatched coordinates of $X_{i^*}$ and $y$, and we have $|H| \leq h$. Let $y_{\star} \subseteq [d]$ denote the set of $\star$'s of $y$. We prove the following properties of $x'$ and $y'$:
\begin{itemize}
\item $y'$ is $\star$ on the set $y_{\star}$, $1$ on the set $H$, and $0$ on $[d]\backslash (H \cup y_{\star})$. 

Proof: This follows from the definition of $y'$: $y'$ maintains the $\star$'s of $y$, and for the rest of the coordinates, $y'_j = 1$ if and only if $y_j \neq (X_{i^*})_j$.
\item If $x$ matches $y$, then $x'$ is $1$ on the set $H$, and it is $0$ on the set $[d]\backslash (H \cup y_{\star})$.

Proof: This is because if $x$ matches $y$, then (1) for any $j \in H$, $x_j = y_j \neq (X_{i^*})_j$, and so $x'_j = x_j \otimes (X_{i^*})_j = 1$, and (2) for any $j \in [d]\backslash (H \cup y_{\star})$, $x_j = y_j = (X_{i^*})_j$, and so $x'_j = x_j \otimes (X_{i^*})_j = 0$.
\item If $x$ doesn't match $y$, then either there exists $j \in H$ such that $x'_j = 0$, or there exists $j \in [d]\backslash (H \cup y_{\star})$ such that $x'_j = 1$.

Proof: If $x$ doesn't match $y$, then there must exist a coordinate $j \in [d] \backslash y_{\star}$ such that $x_j \neq y_j$. (1) If $j \in H$, then since we also have $y_j \neq (X_{i^*})_j$ for $j \in H$, we must have $x_j = (X_{i^*})_j$, and so $x'_j = x_j \otimes (X_{i^*})_j = 0$. (2) If $j \notin H$, i.e., $j \in [d]\backslash (H \cup y_{\star})$, then since we have $y_j = (X_{i^*})_j$ for $j \in [d]\backslash (H \cup y_{\star})$, we must have $x_j \neq (X_{i^*})_j$, and so $x'_j = x_j \otimes (X_{i^*})_j = 1$. 
\end{itemize}
Using these properties, we can prove the claim by the following arguments:
\begin{itemize}
\item If $x$ matches $y$, then (1) the ones of $x'$ is a subset of $y_{\star} \cup H$, which is the set of ones and $\star$'s of $y'$, and (2) the ones of $y'$ is $H$, which is a subset of the ones of $x'$.
\item If $x$ doesn't match $y$, then either (1) there exists $j \in H$ such that $x'_j = 0$, and hence the ones of $y'$ is not a subset of the ones of $x'$, or (2) there exists $j \in [d]\backslash (H \cup y_{\star})$ such that $x'_j = 1$, and hence the ones of $x'$ is not a subset of the ones and $\star$'s of $y'$.
\end{itemize}

{\bf Soundness.} Recall that we can decompose the execution of the protocol into three parts, and the protocol only uses Merlin's message in the subroutines of Part 2 or Part 3. Given $x$, $y$, $\lambda$ and the randomness $R_{\pub}$, the execution of the protocol is completely fixed, and the subroutines that are called and their respective inputs are also fixed. There is a unique special advice of Merlin for every subroutine that is called. We define $m^*(\lambda,x,y,R_{\pub})$ to be the concatenation of these messages, and we claim this is the unique special advice for this protocol.

Consider any $m \neq m^*(\lambda,x,y,R_{\pub})$, next we bound the probability that the protocol outputs $1$ with input $m$. Recall that we can decompose the execution of the protocol into three parts, next we bound the probability of each part:
\begin{enumerate}
\item Consider any $i \in [\ell]$. In the $i$-th level of recursion, the protocol might output $1$ in Step 2(a)(i) when all sets $S_1, S_2, \cdots, S_{\log(\frac{10}{\delta^{(i)}})} \subseteq [d]$ satisfy that the number of $\star$'s in $y|_{S_j}$ is $> 2 w/3$. This happens with probability at most $\frac{\delta^{(i)}}{10}$. Since $\delta^{(i)} = \frac{\delta}{10^i}$, over all $\ell$ iterations, the total probability that the protocol outputs $1$ in Step 2(a)(i) is bounded by $\frac{\delta}{10}$.
\item If the protocol enters Step 2(b), then since $m \neq m^*(\lambda,x,y,R_{\pub})$, at least one of the two subroutines $\pi_{\text{base},f_{\mathrm{sq}}}$ and $\pi_{\text{sq}}$ has a Merlin's message that is not the unique special message, and by Lemma~\ref{lem:base_case_protocol} and Lemma~\ref{lem:w_sparse_subset_query_protocol} this subroutine will return $0$ with probability at least $1-\delta/10$, and so our protocol also returns $0$ with probability at least $1-\delta/10$.
\item If the protocol enters the base case of Step 1, note that then the protocol never enters Step 2(b), so we use Merlin's message only for the subroutine $\pi_{\text{base},f_{\mathrm{pm}}}$. Since this message is not the unique special message, by Lemma~\ref{lem:base_case_protocol} the subroutine will return $0$ with probability at least $1-\delta/10$, and so our protocol also returns $0$ with probability at least $1-\delta/10$.
\end{enumerate}
Combining these three parts using the union bound, we have that the overall soundness is $\leq \delta$.

{\bf Error under unique special advice.} Next we assume that Merlin's advice is the unique special advice $m^*(\lambda,x,y,R_{\pub})$. Recall that we can decompose the execution of the protocol into three parts, next we analyze the error probability of each part:
\begin{enumerate}
\item {\bf Step 2(a).} In this step we recursively determine if $x|_{S_{j^*}}$ matches $y|_{S_{j^*}}$, and an error could occur if $x$ doesn't match $y$, but $x|_{S_{j^*}}$ matches $y|_{S_{j^*}}$. The analysis for this error is similar as that of Step 2(c)(i)(B) in Lemma~\ref{lem:w_sparse_subset_query_protocol}. There are three cases for the $i$-th recursion:
\begin{itemize}
\item Case 1: If the number of unmatched coordinates of $x^{(i)}$ and $y^{(i)}$ is $> h$, but after sampling $S_{j^*}$ we have $x^{(i)}|_{S_{j^*}}$ matches $y^{(i)}|_{S_{j^*}}$, or we enter Step 2(a)(i).

The probability that $x^{(i)}|_{S_{j^*}}$ matches $y^{(i)}|_{S_{j^*}}$ is at most $2^{-h} \leq \epsilon^{(i)}/10$. And by Chernoff bound, the probability that we enter Step 2(a)(i) is at most $e^{-\Theta(w^{(i)})} \leq e^{-\Theta(\log(w^{(i)}/\epsilon^{(i)}))} \leq \epsilon^{(i)}/10$. So the total error in the $i$-th recursion is at most $\epsilon^{(i)}/5$.
\item Case 2: If $\Pr_{X \sim \lambda^{(i)}}[|\text{unmatched coordinates of } X \text{ and } y^{(i)}| \leq h] \leq \epsilon^{(i)}/10$, yet the number of unmatched coordinates of $x^{(i)}$ and $y^{(i)}$ is also $\leq h$.

Then we directly have that this happens with probability at most $\epsilon^{(i)}/10$.
\item Case 3: If $\Pr_{X \sim \lambda^{(i)}}[|\text{unmatched coordinates of } X \text{ and } y^{(i)}| \leq h] \geq \epsilon^{(i)}/10$.

Then by Chernoff bound, $\Pr\big[|\text{unmatched coordinates of } X_j \text{ and } y^{(i)}| \leq h, \forall j \in [t]\big] \leq e^{-\Theta(t \epsilon^{(i)}/10)} \leq \epsilon^{(i)}/10$.
\end{itemize}
Combining these three cases using the union bound, we have that the error of Step 2(a)(ii) in the $i$-th recursion is at most $\epsilon^{(i)}/2$. Adding up the error over all $\ell$ iterations and since $\epsilon^{(i)} = 2^{-i} \cdot \epsilon$, we have that the total error is bounded by $\epsilon/2$.
\item {\bf Step 2(b).} If the protocol enters Step 2(b) and runs the two subroutines $\pi_{\text{base},f_{\mathrm{sq}}}$ and $\pi_{\text{sq}}$, by Lemma~\ref{lem:base_case_protocol} and Lemma~\ref{lem:w_sparse_subset_query_protocol}, each of these two subroutines has a false positive error of at most $\epsilon/10 + \delta/10$, so the total error is bounded by $\epsilon/5 + \delta/5$.
\item {\bf Step 1.} If the protocol enters Step 1 and run the subroutine $\pi_{\text{base},f_{\mathrm{fm}}}$, by Lemma~\ref{lem:base_case_protocol}, this subroutine has a false positive error of at most $\delta$.
\end{enumerate}
Finally note that we only enter one of Step 2(b) and Step 1. So the total error is bounded by $\epsilon/2 + \max\{\epsilon/5 + \delta/5, \delta\} \leq \epsilon + \delta$. Also note that this error is false positive.

{\bf Cost.} Finally we bound the communication cost and time complexity of each players.

{\bf Merlin.} Merlin's communication cost has the following two parts: 
\begin{itemize}
    \item The cost of $\pi_{\text{base},f_{\mathrm{pm}}, d, w', w', \delta'}$ subroutine of Step 1, where the parameters are $w' \leq 100 (\log(\frac{w}{\epsilon}) + \ell)$, and $\delta' \geq 10^{-\ell} \cdot \delta$. Using Lemma~\ref{lem:base_case_protocol}, in this subroutine Merlin sends at most $w' \leq O(\log(\frac{w}{\epsilon}))$ bits.
    \item The cost of the $\pi_{\text{base},f_{\mathrm{sq}}, d, h', w'+h', \delta'}$ and $\pi_{\mathrm{sq}, d, w'+h', \epsilon', \delta'}$ subroutines of Step 2(b), where the parameters are $w' \leq w$, $h' \leq \log(\frac{10\cdot 2^{\ell}}{\epsilon})$, $\epsilon' \geq 2^{-\ell} \epsilon$, and $\delta' \geq 10^{-\ell} \delta$. Using Lemma~\ref{lem:base_case_protocol} and Lemma~\ref{lem:w_sparse_subset_query_protocol}, in these two subroutines the communication cost of Merlin is most 
    \begin{align*}
    &~ \Big(\log h' + h' \log(\frac{e (w'+h')}{h'})\Big) + O\Big(\sqrt{(w'+h') \log((w'+h')/\epsilon')} \cdot (1 + \log(\frac{w'+h'}{\log((w'+h')/\epsilon')}))\Big) \\
    \leq &~ O\left(\log(\frac{w}{\epsilon}) \cdot \log(\frac{w}{\log(1/\epsilon)}) + \sqrt{\big(w + \log(\frac{w}{\epsilon})\big) \cdot \log(\frac{w}{\epsilon})} \cdot (1 + \log(\frac{w}{\log(1/\epsilon)}) \right) \\
    \leq &~ O\left( \Big(\sqrt{w \log(\frac{w}{\epsilon})} + \log(\frac{w}{\epsilon}) \Big) \cdot \log(\frac{w}{\log(1/\epsilon)}) \right)
    \end{align*}
    where the second step follows from Eq.~\eqref{eq:ell_bound} that $\ell \leq 2 \log(\frac{w}{\log(1/\epsilon)})$ and so $h' \leq O(\log(\frac{w}{\epsilon}))$.
\end{itemize}
Thus the total communication cost of Merlin is $c_m \leq O\left( \Big(\sqrt{w \log(\frac{w}{\epsilon})} + \log(\frac{w}{\epsilon}) \Big) \cdot \log(\frac{w}{\log(1/\epsilon)}) \right)$ bit.

{\bf Alice:} Alice's communication and time costs have the following three parts:
\begin{itemize}
\item For each of the $\ell \leq 2 \log(\frac{w}{\log(1/\epsilon)})$ levels of recursions, in the sampling processes of Step 2 and 2(a), Alice doesn't send any message, but she needs to read $t' \leq \frac{100 \cdot 2^{\ell}}{\epsilon} \log(\frac{10 \cdot 2^{\ell}}{\epsilon})$ number of $X_i$'s, and $\log(\frac{10^{\ell+1}}{\delta})$ number of $S_j$'s, and in total this takes time
\begin{align*}
\ell \cdot d \cdot (t' + \log(\frac{10^{\ell+1}}{\delta})) \leq O\left(\frac{d w^2}{\epsilon} \log(\frac{1}{\delta}) \log(\frac{w}{\epsilon})^2\right).
\end{align*}
where we used that $\ell \leq 2 \log(\frac{w}{\log(1/\epsilon)})$.
\item The cost of $\pi_{\text{base},f_{\mathrm{pm}}, d, w', w', \delta'}$ subroutine of Step 1, where the parameters are $w' \leq 100 (\log(\frac{w}{\epsilon}) + \ell)$, and $\delta' \geq 10^{-\ell} \cdot \delta$. Using Lemma~\ref{lem:base_case_protocol}, in this subroutine Alice has communication cost $\log(\frac{1}{\delta'}) = O(\log(\frac{1}{\delta}) + \log(\frac{w}{\log(1/\epsilon)}))$, and time cost $O(d \log(\frac{1}{\delta'})) = O(d \log(\frac{1}{\delta}) + d \log(\frac{w}{\log(1/\epsilon)}))$.
\item The cost of the $\pi_{\text{base},f_{\mathrm{sq}}, d, h', w'+h', \delta'}$ and $\pi_{\mathrm{sq}, d, w'+h', \epsilon', \delta'}$ subroutines of Step 2(b), where the parameters are $w' \leq w$, $h' \leq \log(\frac{10\cdot 2^{\ell}}{\epsilon})$, $\epsilon' \geq 2^{-\ell} \epsilon$, and $\delta' \geq 10^{-\ell} \delta$. Using Lemma~\ref{lem:base_case_protocol} and Lemma~\ref{lem:w_sparse_subset_query_protocol}, and since $\ell \leq 2 \log(\frac{w}{\log(1/\epsilon)})$ and so $h' \leq O(\log(\frac{w}{\epsilon}))$, in these two subroutines Alice has communication cost
\begin{align*}
\log(\frac{1}{\delta'}) + O\left(\sqrt{\frac{w'+h'}{\log((w'+h')/\epsilon')}} + \log(\frac{1}{\delta'})\right) 
\leq O\left(\log(\frac{1}{\delta}) + \sqrt{\frac{w}{\log(1/\epsilon)}}\right),
\end{align*}
and time cost
\begin{align*}
d \log(\frac{1}{\delta'}) + O\left(\frac{d (w'+h')}{\epsilon'} (\log d)^3 \log(\frac{1}{\epsilon'})^2 + d \log(\frac{1}{\delta'})\right) 
\leq &~ O\left(\frac{d w^3}{\epsilon} \log(\frac{1}{\delta}) (\log d)^3 \log(\frac{w}{\epsilon})^3 \right).
\end{align*}
\end{itemize}
Thus the total communication cost of Alice is $O\left(\log(\frac{1}{\delta}) + \sqrt{\frac{w}{\log(1/\epsilon)}}\right)$, and the total time cost of Alice is $O\left(\frac{d w^3}{\epsilon} \log(\frac{1}{\delta}) (\log d)^3 \log(\frac{w}{\epsilon})^3 \right)$.

{\bf Bob:} First note that the time complexity of Bob follows the same upper bound as that of Alice. Bob's communication cost has the following three parts:
\begin{itemize}
\item For $\ell \leq 2 \log(\frac{w}{\log(1/\epsilon)})$ levels of recursions, in the sampling processes of Step 2 and 2(a), Bob sends the index $i^*$ using $\log(t') \leq O(\ell + \log(\frac{1}{\epsilon}))$ bits, and the index $j^*$ using $O(\log (\ell) + \log \log(\frac{1}{\delta}))$ bits. So in total the communication cost of Bob in the sampling processes is
\begin{align*}
\ell \cdot O(\ell + \log(\frac{1}{\epsilon}) + \log \log(\frac{1}{\delta})) \leq O\left(\log(\frac{w}{\log(1/\epsilon)}) \cdot \Big(\log(\frac{w}{\log(1/\epsilon)}) + \log(\frac{1}{\epsilon}) + \log \log(\frac{1}{\delta})\Big)\right),
\end{align*}
where we used that $\ell \leq 2 \log(\frac{w}{\log(1/\epsilon)})$.
\item The cost of $\pi_{\text{base},f_{\mathrm{pm}}, d,  w', w', \delta'}$ subroutine of Step 1, where the parameters are $w' \leq 100 (\log(\frac{w}{\epsilon}) + \ell)$, and $\delta' \geq 10^{-\ell} \cdot \delta$. Using Lemma~\ref{lem:base_case_protocol}, in this subroutine Bob has communication cost $\log(\frac{1}{\delta'}) = O(\log(\frac{1}{\delta}) + \log(\frac{w}{\log(1/\epsilon)}))$.
\item The cost of the $\pi_{\text{base},f_{\mathrm{sq}}, d, h', w'+h', \delta'}$ and $\pi_{\mathrm{sq}, d, w'+h', \epsilon', \delta'}$ subroutines of Step 2(b), where the parameters are $w' \leq w$, $h' \leq \log(\frac{10\cdot 2^{\ell}}{\epsilon})$, $\epsilon' \geq 2^{-\ell} \epsilon$, and $\delta' \geq 10^{-\ell} \delta$. Using Lemma~\ref{lem:base_case_protocol} and Lemma~\ref{lem:w_sparse_subset_query_protocol}, and since $\ell \leq 2 \log(\frac{w}{\log(1/\epsilon)})$ and so $h' \leq O(\log(\frac{w}{\epsilon}))$, in these two subroutines Bob has communication cost
\begin{align*}
&~ O\left(\sqrt{(w'+h') \log(\frac{w'+h'}{\epsilon'})} \cdot \log(\frac{w'+h'}{\log((w'+h')/\epsilon')})+ \log \log(\frac{w'+h'}{\delta'}) \cdot \log (w'+h') + \log(\frac{1}{\delta'}) \right) \\
\leq &~ O\left(\Big(\sqrt{w \log(\frac{w}{\epsilon})} + \log(\frac{w}{\epsilon})\Big) \cdot \log(\frac{w}{\log(1/\epsilon)}) + \log \log(\frac{w}{\delta}) \cdot \log\Big(w + \log(\frac{w}{\epsilon}) \Big) + \log(\frac{1}{\delta}) \right).
\end{align*}
\end{itemize}
Thus the total communication cost of Bob is $ O\Big(\Big(\sqrt{w \log(\frac{w}{\epsilon})} + \log(\frac{w}{\epsilon})\Big) \cdot \log(\frac{w}{\log(1/\epsilon)}) + \log \log(\frac{w}{\delta}) \cdot \log\Big(w + \log(\frac{w}{\epsilon}) \Big) + \log(\frac{1}{\delta}) \Big)$.

{\bf Carol:} Let $\mathcal{T}_{\mathrm{Sample}}$ be the time to sample a set from distribution $\lambda$ conditioned on that the size of the sampled set is in some range, and let $\mathcal{T}_{\mathrm{Update}}$ be the time to update $\lambda$ to be the restriction on a subset of $[d]$. Carol's communication and time costs have the following three parts:
\begin{itemize}
\item For each of the $\ell \leq 2 \log(\frac{w}{\log(1/\epsilon)})$ levels of recursions, in the sampling processes of Step 2 and 2(a), Carol needs to generate $t' \leq \frac{100 \cdot 2^{\ell}}{\epsilon} \log(\frac{10 \cdot 2^{\ell}}{\epsilon})$ number of $X_i$'s from $\lambda$, and generate $\log(\frac{10^{\ell+1}}{\delta})$ number of $S_j$'s as random sets, since $\ell \leq 2 \log(\frac{w}{\log(1/\epsilon)})$, in total the communication cost of Carol is
\begin{align*}
\ell \cdot d \cdot (t' + \log(\frac{10^{\ell+1}}{\delta})) \leq O\left(\frac{d w^2}{\epsilon} \log(\frac{1}{\delta}) \log(\frac{w}{\epsilon})^2\right),
\end{align*}
and the time cost of Carol is
\begin{align*}
\ell \cdot d \cdot (t' \mathcal{T}_{\mathrm{Sample}} + \log(\frac{10^{\ell+1}}{\delta})) \leq O\left(\frac{d w^2}{\epsilon} \log(\frac{1}{\delta}) \log(\frac{w}{\epsilon})^2 \cdot \mathcal{T}_{\mathrm{Sample}}\right).
\end{align*}
\item The cost of $\pi_{\text{base},f_{\mathrm{pm}}, d, w, w', \delta'}$ subroutine of Step 1, where the parameters are $w' \leq 100 (\log(\frac{w}{\epsilon}) + \ell)$, and $\delta' \geq 10^{-\ell} \cdot \delta$. Using Lemma~\ref{lem:base_case_protocol}, in this subroutine the communication and time of Carol are both bounded by $O(d \log(\frac{1}{\delta'})) = O(d \log(\frac{1}{\delta}) + d \log(\frac{w}{\log(1/\epsilon)}))$.
\item The cost of the $\pi_{\text{base},f_{\mathrm{sq}}, d, h', w'+h', \delta'}$ and $\pi_{\mathrm{sq}, d, w'+h', \epsilon', \delta'}$ subroutines of Step 2(b), where the parameters are $w' \leq w$, $h' \leq \log(\frac{10\cdot 2^{\ell}}{\epsilon})$, $\epsilon' \geq 2^{-\ell} \epsilon$, and $\delta' \geq 10^{-\ell} \delta$. Using Lemma~\ref{lem:base_case_protocol} and Lemma~\ref{lem:w_sparse_subset_query_protocol}, and since $\ell \leq 2 \log(\frac{w}{\log(1/\epsilon)})$ and so $h' \leq O(\log(\frac{w}{\epsilon}))$, in these two subroutines Carol has communication cost
\begin{align*}
d \log(\frac{1}{\delta'}) + O\left(\frac{d (w'+h')}{\epsilon'} \log(\frac{w'+h'}{\delta'}) (\log d)^3 \log(\frac{1}{\epsilon'})^2\right) 
\leq &~ O\left(\frac{d w^3}{\epsilon} \log(\frac{1}{\delta}) (\log d)^3 \log(\frac{w}{\epsilon})^3 \right),
\end{align*}
and time cost
\begin{align*}
&~ d \log(\frac{1}{\delta'}) + O\left((\frac{w'+h'}{\epsilon'}) (\log d)^2 \log (\frac{1}{\epsilon'}) \mathcal{T}_{\mathrm{Sample}} + d \mathcal{T}_{\mathrm{Update}} + \log(\frac{w'+h'}{\delta'}) d \log d\right) \\
\leq &~ O\left(\frac{d w^3}{\epsilon} \log(\frac{1}{\delta}) (\log d)^3 \log(\frac{w}{\epsilon})^3 \cdot \mathcal{T}_{\mathrm{Sample}} + d \cdot \mathcal{T}_{\mathrm{Update}} \right).
\end{align*}
\end{itemize}
Thus the total communication cost of Carol is $O\left(\frac{d w^3}{\epsilon} \log(\frac{1}{\delta}) (\log d)^3 \log(\frac{w}{\epsilon})^3 \right)$, and the total time cost of Carol is $O\left(\frac{d w^3}{\epsilon} \log(\frac{1}{\delta}) (\log d)^3 \log(\frac{w}{\epsilon})^3 \cdot \mathcal{T}_{\mathrm{Sample}} + d \cdot \mathcal{T}_{\mathrm{Update}} \right)$.
\end{proof}

\section{Reduction from Communication Protocols to Data Structures}\label{sec:w_sparse_reduction_cc_to_ds}
We now show how to design a word-RAM data structure for a search problem as in Eq.~\eqref{eq_pattern_matching} using a communication protocol in the specialized product-distribution model (Definition~\ref{def:special_communication_model}), yielding Theorem \ref{thm_meta_informal}.

\begin{theorem}[Reduction from communication protocols to data structures]\label{thm:reduction_CC_DS}
Let $f: \X \times \Y \to \{0,1\}$ be the target function. If there exists a protocol $\pi$ for $f$ in the specialized communication model (Definition~\ref{def:special_communication_model}) with $\delta$ soundness, $\epsilon$ false positive error under the unique special message of Merlin, $(c_a,c_b,c_c,c_m)$-cost, and $(t_a, t_b, t_c)$-time, then there exists a data structure that supports the following operations:
\begin{itemize}
    \item {\bf Preprocess:} Given an input set $\mathcal{D} = \{x^{(1)}, \cdots, x^{(n)}\} \subset \X$, preprocess the set $\mathcal{D}$ using $s = 2^{O(c_a + c_b + c_m)} \cdot (n + c_c)$ space, and $t_p = 2^{O(c_a + c_b + c_m)} \cdot n \cdot (t_a + t_b + t_c)$ time.
    \item {\bf Query:} Given a query $y \in \Y$, output all $x^{(i)} \in \mathcal{D}$ such that $f(x^{(i)}, y) = 1$. The expected query time is $t_q = 2^{O(c_a+c_m)} \cdot (c_a + c_b + c_m + c_c + t_b) + O(\epsilon n) + 2^{O(c_m)}\cdot \delta n + O(n_y)$, where $n_y$ is the number of $x^{(i)}$'s that $f(x^{(i)}, y) = 1$.
\end{itemize}
\end{theorem}
\begin{proof}
Let $\pi$ be the protocol described in the theorem statement, and let $R_{\pub}$ and $R_{\pri}$ be the sequence of random bits used by Carol in the protocol, where $R_{\pri}$ is hidden from Merlin. {\bf Let $\lambda$ be the uniform distribution over $\mathcal{D}$.} We consider executing the protocol $\pi$ with distribution $\lambda$ and randomness $R_{\pub}$ and $R_{\pri}$. For simplicity, in this proof we will use $\pi(x, y, m)$ as a shorthand of $\pi(\lambda, x, y, m, R_{\pub}, R_{\pri})$.

We define a transcript of the protocol $\pi$ as the messages between Alice, Bob, and Merlin. Note that when the distribution $\lambda$ and the randomness $R_{\pub}$ and $R_{\pri}$ are fixed, the messages of Carol are completely determined by the transcript.

\noindent {\bf Preprocess:} The data structure builds a tree to simulate the protocol $\pi$ when given the distribution $\lambda$ and randomness $R_{\pub}$ and $R_{\pri}$. In each step of the protocol, there are five possible behaviors:
\begin{enumerate}
    \item Alice sends Bob a message.
    \item Bob sends Alice a message.
    \item Merlin sends a message to Alice and Bob.
    \item Carol generates a message using the randomness $R_{\pub}, R_{\pri}$ and the distribution $\lambda$, and writes this message to the public channel. %
    \item Alice and Bob both output a value and the protocol ends.
\end{enumerate}
When building the tree, each node corresponds to one step of the protocol. We start with the root node, and build each node depending on the behavior of that step:
\begin{enumerate}
    \item If Alice sends Bob a message, we call this node an Alice node. In this node we store an array of all possible messages of Alice in this step. Each entry of the array stores a pointer to a child node that corresponds to the next step of the protocol. We then recursively proceed to all the child nodes.
    \item If Bob sends Alice a message, we call this node a Bob node. In this node we store an array of all possible messages of Bob in this step. Each entry of the array stores a pointer to a child node that corresponds to the next step of the protocol. We then recursively proceed to all the child nodes.
    \item If Merlin sends a message to Alice and Bob, we call this node a Merlin node. In this node we store an array of all possible messages of Merlin in this step. Each entry of the array stores a pointer to a child node that corresponds to the next step of the protocol. We then recursively proceed to all the child nodes.
    \item If Carol generates a message using the randomness $R_{\pub},R_{\pri}$ and the distribution $\lambda$, we call this node a Carol node. Carol's message is fixed given the transcript so far and $R_{\pub},R_{\pri},\lambda$. We store this unique message of Carol in this node. This node only has one child node.%
    \item If Alice and Bob both output something and end the protocol, we call this node a leaf node. We call the node a ``$0$-leaf node'' if the output is $0$, and ``$1$-leaf node'' if the output is $1$. A root-to-leaf path corresponds to a transcript $\tau$ of the protocol. 
    
    For each $1$-leaf node $v$, we store a list of all $x^{(i)} \in \mathcal{D}$ that satisfies when Alice has input $x^{(i)}$ and Bob and Merlin send messages according to $\tau$, Alice's messages are consistent with $\tau$.
\end{enumerate}

\noindent {\bf Query:} Given a query $y \in \{0, 1\}^d$, for every Merlin's message $m$ that is possible, we follow the tree to simulate the protocol of Bob when his input is $y$ and Merlin's message is $m$. We start from the root node.
\begin{enumerate}
    \item If the node is an Alice node, we recurse on all child nodes of this node.
    \item If the node is a Bob node, we compute the message of Bob based on $y$ and the current transcript of the protocol. We only recurse to the child node that corresponds to this message.
    \item If the node is a Merlin node, we use $m$ as Merlin's message. We only recurse to the child node that corresponds to this message.
    \item If the node is a Carol node, we read the message stored in this node.
    \item If the node is a $1$-leaf node, we enumerate all $x^{(i)}$'s stored in this node.
\end{enumerate}
Finally, we output all $x^{(i)}$'s in the visited $1$-leaves that satisfy $f(x^{(i)}, y) = 1$.

In the analysis, we use $L(y, m)$ to denote the set of $1$-leaves that are visited by the query with Merlin's message $m$.

\noindent {\bf Properties:} Before proving the correctness and time complexity, we first list some properties of data structure that directly follow from the definitions.
\begin{enumerate}
    \item There is a one-to-one mapping between the root-to-leaf paths of the tree and all semantically possible transcripts.
    \item There is a one-to-one mapping between the root-to-leaf paths for the leaves in $L(y, m)$ and all semantically possible transcripts where Bob has input $y$, Merlin's message is $m$, and the output is $1$.
    \item $\forall v \in L(y)$, let $\tau$ denote the transcript that corresponds to $v$. The leaf $v$ stores all $x^{(i)}$ that satisfy (1) the transcript of $\pi$ on input $(x^{(i)}, y, m)$ is $\tau$, (2) the output of $\pi$ on $(x^{(i)}, y, m)$ is $1$.
\end{enumerate}

\noindent {\bf Correctness:} Consider any $x^{(i)} \in \mathcal{D}$ such that $f(x^{(i)}, y)=1$. The protocol guarantees that there exists a unique special message of Merlin $m^* = m^*(\lambda, x^{(i)}, y, R_{\pub})$. Since the protocol only has false positive error, we have that the protocol must output $1$ when given the input $(x^{(i)}, y, m^*)$.

The query algorithm described above enumerates all possible Merlin's messages, so it sets Merlin's message to be $m^*$ in one of its iterations. By properties 2 and 3 above, the query algorithm visits all the possible $1$-leaves, so it examines all $x^{(i')}$ where $\pi(x^{(i')}, y, m^*) = 1$. As a result, it will encounter $x^{(i)}$. This argument applies to every $x^{(i)}$ for which $f(x^{(i)}, y)=1$, so the query will encounter all such $x^{(i)}$'s.

\noindent {\bf Space:} Each root-to-leaf path of the tree corresponds to a unique transcript. Since Alice sends $c_a$ bits, Bob sends $c_b$ bits, and Merlin sends $c_m$ bits in the protocol, there are in total $2^{O(c_a+c_b+c_m)}$ number of leaves. The space of the data structure consists of three parts:
\begin{itemize}
    \item Space to store the tree structure: Since the tree has $2^{O(c_a+c_b+c_m)}$ number of leaves, it takes $2^{O(c_a+c_b+c_m)}$ bits to store the tree structure.
    \item Space to store the messages of Carol: On each root-to-leaf path, we need to store all the messages of Carol using $c_c$ bits. So in total this takes $2^{O(c_a+c_b+c_m)} \cdot c_c$ bits.
    \item Space to store $x^{(i)}$'s in the leaf nodes: In each leaf node we store at most $n$ number of $x^{(i)}$'s. So in total this takes $2^{O(c_a+c_b+c_m)} \cdot n$ bits.
\end{itemize}
Thus in total the data structure uses $s = 2^{O(c_a+c_b+c_m)} \cdot (n + c_c)$ bits of space.

\noindent {\bf Preprocessing time.} When building the data structure we need to do the following for each leaf:
\begin{itemize}
    \item For each data point $x^{(i)} \in \mathcal{D}$, in order to check whether we need to store $x^{(i)}$ in this leaf, we simulate the protocol where Alice has input $x^{(i)}$ and Bob and Merlin send the messages according to the transcript of this leaf. This takes $O(n (t_a + t_b + t_c))$ time in total.
\end{itemize}
Since there are $2^{O(c_a+c_b+c_m)}$ number of leaves, the total preprocessing time is $t_p = 2^{O(c_a+c_b+c_m)} \cdot n \cdot (t_a + t_b + t_c)$.

\noindent {\bf Query time.} 
Each leaf node reached by the query corresponds to a unique transcript where Bob has input $y$ (see property 2). Since the total messages of Alice has length $c_a$ bits, and there are $2^{O(c_m)}$ possible messages of Merlin, in total the query reaches $\sum_m |L(y,m)| \leq 2^{O(c_a+c_m)}$ number of leaf nodes. The query time of the data structure consists of three parts:
\begin{itemize}
    \item Time to reach the leaf nodes: The tree has height at most $O(c_a+c_b+c_m)$, and the query simulates Bob using $t_b$ time for each leaf. So in total this takes $2^{O(c_a+c_m)} \cdot (c_a+c_b + c_m + t_b)$ time.
    \item Time to read the messages of Carol: On the path to each leaf, the query needs to read $c_c$ bits of Carol. So in total this takes $O(2^{O(c_a+c_m)} \cdot c_c)$ time.
    \item Time to enumerate the $x^{(i)}$'s stored in the leaf nodes: %
    \begin{itemize}
    \item For all $x^{(i)} \in \mathcal{D}$ where $f(x^{(i)},y) = 1$: First note that there are $n_y$ number of such $x^{(i)}$'s.
    \begin{itemize}
        \item For the unique special message of Merlin $m^* = m^*(\lambda, x^{(i)}, y, R_{\pub})$, the protocol outputs $1$ given the input $(x^{(i)}, y, m^*)$, so the query algorithm examines all such $x^{(i)}$'s at least once using $O(n_y)$ time.
        \item For every other possible Merlin's message $m$, since the protocol has soundness $\delta$, the protocol outputs $1$ given the input $(x^{(i)}, y, m)$ with probability at most $\delta$. There are in total $2^{O(c_m)}$ number of Merlin's messages, so in expectation the query algorithm examines all such $x^{(i)}$'s for $2^{O(c_m)} \cdot \delta n_y$ times under these messages $m$'s.
    \end{itemize}
    Adding up the above two cases, in expectation the query algorithm spends $O(n_y) + 2^{O(c_m)} \cdot \delta n_y$ time to visit $x^{(i)}$'s where $f(x^{(i)},y) = 1$.
    \item For all $x^{(i)} \in \mathcal{D}$ where $f(x^{(i)},y) = 0$: First note that there are $n-n_y$ such $x^{(i)}$'s.
    \begin{itemize}
        \item For the unique special message of Merlin $m^* = m^*(\lambda, x^{(i)}, y, R_{\pub})$, since the protocol has false positive error $\epsilon$ under this unique special message, i.e., we have
        \[
        \Pr_{x \sim \lambda, R_{\pub}, R_{\pri}}\big[\pi\big(\lambda, x, y, m^*(\lambda, x, y, R_{\pub}), R_{\pub}, R_{\pri} \big) \neq f(x,y)\big] \leq \epsilon,
        \]
        and since $\lambda$ is the uniform distribution over $\mathcal{D}$, in expectation the query algorithm examines all such $x^{(i)}$ for $O(\epsilon n)$ times under the unique special messages $m^*$'s.
        \item For every other possible Merlin's message $m$, using the same argument as the ``$f(x^{(i)},y) = 1$'' case, in expectation the query algorithm examines all such $x^{(i)}$'s for $2^{O(c_m)} \cdot \delta (n-n_y)$ times under these messages $m$'s.
    \end{itemize}
    Adding up the above two cases, in expectation the query algorithm spends $O(\epsilon n + 2^{O(c_m)} \cdot \delta (n-n_y))$ time to visit $x^{(i)}$'s where $f(x^{(i)},y) = 0$.
    \end{itemize}
\end{itemize}
Thus in total the data structure has expected query time $t_q = O(2^{O(c_a+c_m)} \cdot (c_a+c_b+c_m+c_c+t_b) + n_y + \epsilon n + 2^{O(c_m)} \cdot \delta n)$.
\end{proof}

\section{Data structure for partial match with \texorpdfstring{$w$}{w} wildcards}\label{sec:partial_match_ds}
In this section we combine the communication protocol of Theorem~\ref{thm:w_sparse_partial_match_protocol} and the reduction of Theorem~\ref{thm:reduction_CC_DS} to construct a data structure for the partial match problem with up to $w$ wildcards.

\begin{theorem}[Data structure for partial match with up to $w$ wildcards]\label{thm:partial_match_ds}
For any large enough integer $n$, any $1 < c \leq \frac{\log n}{10 (\log \log n)^4}$, any $d \leq n^{0.1}$, there exists a data structure for the partial match problem with up to $w = c \log n$ wildcards with the following guarantees:
\begin{itemize}
    \item {\bf Preprocess:} Given an input set $\mathcal{D} = \{x^{(1)}, \cdots, x^{(n)}\} \subset \X$, preprocess the set $\mathcal{D}$ using $s = O(n^{1.1})$ space, and $t_p = O(n^{2.1} d)$ time.
    \item {\bf Query:} Given a query $y \in \Y$, output all $x^{(i)} \in \mathcal{D}$ such that $f(x^{(i)}, y) = 1$. The expected query time is $t_q = O(n^{1-\frac{1}{\Theta(c (\log c)^2)}} + n_y)$, where $n_y$ is the number of $x^{(i)}$'s that $f(x^{(i)}, y) = 1$.
\end{itemize}
\end{theorem}
\begin{proof}
Let $C_1$ be the maximum constant in the big $O$ notations of Theorem~\ref{thm:w_sparse_partial_match_protocol}, and let $C_2$ be the maximum constant in the big $O$ notations of Theorem~\ref{thm:reduction_CC_DS}. Let $\lambda$ be the uniform distribution over $\mathcal{D}$. 
Let $\epsilon = 2^{-\frac{\log n}{c (\log c)^2} \cdot \frac{1}{10^9 C_1^4 C_2^4 \log(C_1 C_2)^2}}$, $\delta = n^{-\frac{1}{100 C_1 C_2}}$. Note that we have $\log(\frac{1}{\epsilon}) = \frac{\log n}{c (\log c)^2} \cdot \frac{1}{10^9 C_1^4 C_2^4 \log(C_1 C_2)^2}$, and $\log(\frac{1}{\delta}) = \frac{\log n}{100 C_1 C_2}$. For large enough $n$ and since $c \leq \frac{\log n}{10 (\log \log n)^4}$, we have that
\begin{align}
\log(\frac{w}{\epsilon}) = &~ \log (c \log n) + \frac{\log n}{c (\log c)^2} \cdot \frac{1}{10^9 C_1^4 C_2^4 \log(C_1 C_2)^2} \notag \\
\leq &~ \frac{\log n}{c (\log c)^2} \cdot \frac{1}{9 \cdot 10^8 \cdot C_1^4 C_2^4 \log(C_1 C_2)^2}, \label{eq:ub_params_1} \\
\frac{w}{\log(1/\epsilon)} = &~ c^2 (\log c)^2 \cdot 10^9 C_1^4 C_2^4 \log(C_1 C_2)^2 \notag \\
\leq &~ \frac{(\log n)^2}{(\log \log n)^2} \cdot 10^9 C_1^4 C_2^4 \log(C_1 C_2)^2 \leq \frac{1}{10^4 C_1^2 C_2^2} (\log n)^2, \label{eq:ub_params_2} \\
\log(\frac{w}{\log(1/\epsilon)}) = &~ \log\Big(c^2 (\log c)^2 \cdot 10^9 C_1^4 C_2^4 \log(C_1 C_2)^2 \Big) \leq 4 \log c + 6 \log(C_1 C_2) + 30. \label{eq:ub_params_3}
\end{align}

By Theorem~\ref{thm:w_sparse_partial_match_protocol}, there is a communication protocol for $w$-sparse partial match with the following guarantees:
\begin{enumerate}
\item The protocol has $\delta$ soundness.
\item When Merlin's messages are the unique special messages, the protocol has false positive error $\epsilon+\delta$.
\item The communication cost of Merlin is bounded by
\begin{align*}
c_m \leq &~ C_1 \cdot \Big(\sqrt{w \log(\frac{w}{\epsilon})} + \log(\frac{w}{\epsilon}) \Big) \cdot \log(\frac{w}{\log(1/\epsilon)}) \\
\leq &~ C_1 \cdot \Big(\sqrt{c \log n \cdot \frac{\log n}{c (\log c)^2} \cdot \frac{1}{9 \cdot 10^8 \cdot C_1^4 C_2^4 \log(C_1 C_2)^2}} \Big) \cdot \Big( 4 \log c + 6 \log(C_1 C_2) + 30 \Big) \\
\leq &~ \frac{\log n}{1000 C_1 C_2^2},
\end{align*}
where the second step follows from Eq.~\eqref{eq:ub_params_1} and \eqref{eq:ub_params_3}.
\item The communication cost of Alice is bounded by
\begin{align*}
c_a \leq &~ C_1 \cdot \left(\log(\frac{1}{\delta}) + \sqrt{\frac{w}{\log(1/\epsilon)}}\right) \\
\leq &~ C_1 \cdot \left( \frac{\log n}{100 C_1 C_2} + \frac{\log n}{100 C_1 C_2} \right) 
\leq \frac{\log n}{50 C_2},
\end{align*}
where the second step follows from Eq.~\eqref{eq:ub_params_2}.

The communication cost of Bob is bounded by
\begin{align*}
c_b \leq &~ C_1 \cdot \left(\Big(\sqrt{w \log(\frac{w}{\epsilon})} + \log(\frac{w}{\epsilon})\Big) \cdot \log(\frac{w}{\log(1/\epsilon)}) + \log \log(\frac{w}{\delta}) \cdot \log\Big(w + \log(\frac{w}{\epsilon}) \Big) + \log(\frac{1}{\delta}) \right) \\
\leq &~ \frac{\log n}{1000 C_1 C_2^2} + O((\log \log n)^2) + \frac{\log n}{100 C_2} 
\leq \frac{\log n}{50 C_2}
\end{align*}
where in the second step we bound the first term by the same proof as the upper bound of $c_m$, bound the second term by noting that $\log \log (\frac{w}{\delta}) = O(\log \log n)$ and $\log\Big(w + \log(\frac{w}{\epsilon}) = O(\log \log n)$, and bound the third term by $\log(\frac{1}{\delta}) = \frac{\log n}{100 C_1 C_2}$, and the third step follows from $O((\log \log n)^2)$ is much smaller than the $\log n$ terms when $n$ is large enough.

The time complexity of both players are bounded by 
\begin{align*}
t_a, t_b \leq &~ O\left(\frac{d w^3}{\epsilon} \log(\frac{1}{\delta}) (\log d)^3 \log(\frac{w}{\epsilon})^3 \right) 
\leq O\left(n^{\frac{1}{10^9 c (\log c)^2}} \cdot d \right) \cdot \poly(\log n),
\end{align*}
where the second step follows from $\epsilon = 2^{-\frac{\log n}{c (\log c)^2} \cdot \frac{1}{10^9 C_1^4 C_2^4 \log(C_1 C_2)^2}}$, and that the terms $w$, $\log(\frac{1}{\delta})$, $\log d$, $\log(\frac{w}{\epsilon})$ are all bounded by $O(\log n)$.
\item The communication cost of Carol is bounded by 
\begin{align*}
c_c \leq O\left(\frac{d w^3}{\epsilon} \log(\frac{1}{\delta}) (\log d)^3 \log(\frac{w}{\epsilon})^3 \right) \leq O\left(n^{\frac{1}{10^9 c (\log c)^2}} \cdot d \right) \cdot \poly(\log n),
\end{align*}
where we use the same argument for the upper bound of $t_a$ and $t_b$.

The time complexity of Carol is bounded by
\begin{align*}
t_c \leq &~ O\left(\frac{d w^3}{\epsilon} \log(\frac{1}{\delta}) (\log d)^3 \log(\frac{w}{\epsilon})^3 \cdot \mathcal{T}_{\mathrm{Sample}} + d \cdot \mathcal{T}_{\mathrm{Update}} \right) \\
\leq &~O\left(n^{1+\frac{1}{10^9 c (\log c)^2}} \cdot d \cdot \poly(\log n) + d n\right)
\end{align*}
where the second step follows from the same argument for the upper bound of $t_a$ and $t_b$, and that $\mathcal{T}_{\mathrm{Sample}} \leq O(n)$ and $\mathcal{T}_{\mathrm{Update}} \leq O(n)$ when $\lambda$ is the uniform distribution over $\mathcal{D}$.
\end{enumerate}

Then using Theorem~\ref{thm:reduction_CC_DS} and this protocol, we can build a data structure for the partial match problem with up to $w = c \log n$ wildcards with the following guarantees:
\begin{itemize}
\item {\bf Query time:}
\begin{align*}
t_q = &~ 2^{C_2 \cdot (c_a+c_m)} \cdot (c_a + c_b + c_m + c_c + t_b) + O(\epsilon n) + 2^{C_2 \cdot c_m}\cdot \delta n + O(n_y) \\
\leq &~ O(n^{0.1} d \cdot \poly(\log n)) + O(n^{1-\frac{1}{\Theta(c(\log c)^2)}}) + O(n^{1-\frac{1}{1000 C_1 C_2}}) \\
\leq &~ O(n^{1-\frac{1}{\Theta(c(\log c)^2)}}),
\end{align*}
where in the second step we bound the first term using $2^{C_2 \cdot (c_a+c_m)} \leq O(n^{0.05})$ and $c_a + c_b + c_m + c_c + t_b \leq n^{0.05} d \poly(\log n)$, we bound the second term using $\epsilon = 2^{-\frac{\log n}{\Theta(c(\log c)^2)}}$, we bound the third step using that $2^{C_2 \cdot c_m} \leq n^{\frac{1}{1000 C_1 C_2}}$ and $\delta = n^{-\frac{1}{100 C_1 C_2}}$.
Note that the dominating term of the query time comes from $O(\epsilon n)$.
\item {\bf Space:}
\begin{align*}
s = &~ 2^{C_2 \cdot (c_a + c_b + c_m)} \cdot (n + c_c) \\
\leq &~ n^{0.05} \cdot (n + n^{10^{-9}}\cdot  d \cdot \poly(\log n)) \leq O(n^{1.1}),
\end{align*}
where the second step follows from $c_a + c_b + c_m \leq \frac{\log n}{20 C_2}$.
\item {\bf Preprocessing time:}
\begin{align*}
t_p = &~ 2^{C_2 \cdot (c_a + c_b + c_m)} \cdot n \cdot (t_a + t_b + t_c) \\
\leq &~ n^{0.05} \cdot n \cdot n^{1+10^{-9}} \cdot d \cdot \poly(\log n) 
\leq O(n^{2.1} d). \qedhere
\end{align*}
\end{itemize}
\end{proof}

\section{List-of-Points Lower Bound}\label{sec:lop_lower_bound}
In this section we show that the list-of-points lower bound of \cite{ak20} implies an $n^{1-\Theta(1/\sqrt{c})}$ list-of-points lower bound for the query time of the partial match problem $\PM_{n,d=c\log n}$. Note that this problem is equivalent to the online OV problem with dimension $d=2c\log n$, and in fact we prove our lower bound for the online OV problem. The list-of-points model is a general model that captures data-independent algorithms, and it was introduced in \cite{alrw16}. See Definition~1.5 of \cite{alrw16} and Definition 2 of \cite{ak20}.
\begin{definition}[List-of-points data structures]
Let $U$ and $Q$ be spaces of data points and queries. Let $S: U \times Q \to \{0, 1\}$ be a similarity measure.
A list-of-points data structure is defined as follows:
\begin{itemize}
    \item We fix (possibly random) sets $A_i \subseteq U$, for $i \in [m]$. Also, for each possible query point $y \in Q$, we associate a (random) set of indices $I(y) \subseteq [m]$.
    \item For a given dataset $\mathcal{D} \subseteq U$, the data structure maintains $m$ lists of points $L_1, L_2, \cdots, L_m$, where $L_i = \mathcal{D} \cap A_i$.
    \item On query $y \in Q$, we scan through each list $L_i$ for all $i \in I(y)$ and check whether there exists some $x \in L_i$ with $S(x, y) = 1$. If it exists, return $x$.
\end{itemize}
The total space is $s = m + \sum_{i=1}^m |L_i|$, and the query time is $t = |I(y)| + \sum_{i \in I(y)} |L_i|$.
\end{definition}

Our main lower bound theorem states that any list-of-points data structure for the online OV problem with dimension $d=c\log n$ must use $n^{1-1/\sqrt{c}}$ query time if it uses polynomial space. 
\begin{theorem}[Main list-of-points lower bound]\label{thm:main_lower_bound}
Let $1600 \leq c \leq n$. Consider any list-of-points data structure for the online OV problem over $n$ points of dimension $d = c \log n$ that succeeds with probability $\geq 0.995$. For any $\gamma \in [1, \sqrt{c}]$, if the data structure uses expected space $\leq n^{1+\gamma}$, then its expected query time is at least
\[
n^{1 - \Theta(\gamma/\sqrt{c})}.
\]
\end{theorem}
We use the results of \cite{ak20} to prove this theorem. In Section~\ref{sec:background_ak20}, we provide the necessary background of \cite{ak20}. In Section~\ref{sec:lower_bound_KL}, we combine the lower bound of \cite{ak20} and a ``random instance to worst case'' reduction tailored for our problem to prove a lower bound that has KL divergence terms. In Section~\ref{sec:simplify_KL} we simplify the KL divergence terms. Finally in Section~\ref{sec:proof_lower_bound} we prove Theorem~\ref{thm:main_lower_bound}.

\subsection{Background of \texorpdfstring{\cite{ak20}}{[AK20]}}\label{sec:background_ak20}
We first introduce some definitions from \cite{ak20} that we will use in our proofs.
\begin{definition}[KL-divergence]
The KL-divergence between any two discrete distributions over $\Omega$ is defined as
\[
\D{P}{Q} = \sum_{\omega \in \Omega} P(\omega) \log \frac{P(\omega)}{Q(\omega)}.
\]
In particular, for Bernoulli distributions with probabilities $p$ and $q$, we define
\[
\dd{p}{q} = p \log \frac{p}{q} + (1 - p) \log \frac{1-p}{1-q}.
\]
\end{definition}

For any probability $p \in [0,1]$, we use $\text{Bernoulli}(p)$ to denote the Bernoulli distribution that generates $1$ with probability $p$ and $0$ with probability $1-p$. With a slight abuse of notation, we also use $\text{Bernoulli}(P)$ to denote two-dimensional Bernoulli distributions, which is defined as follows:
\begin{definition}[Two dimensional Bernoulli distributions]
Let $P = \left[\begin{smallmatrix}
P_{0,0} & P_{0,1} \\
P_{1,0} & P_{1,1}
\end{smallmatrix}\right]
$ where for all $i,j \in \{0,1\}$, $P_{i,j} \in [0,1]$, and $\sum_{i=0}^1\sum_{j=0}^1 P_{i,j} = 1$. 
We say a random variable $X \in \{0,1\}^2$ follows the distribution $\text{Bernoulli}(P)$ if $\Pr[X = [\begin{smallmatrix}
1-i \\
1-j
\end{smallmatrix}]] = P_{i,j}$ for all $i,j \in \{0,1\}$.
\end{definition}
For any two probability distributions $T$ and $P$, we also use $T \ll P$ to denote that $T$ is absolutely continuous with respect to $P$.

\cite{ak20} considers the gap set similarity search problem (GapSS) with four parameters $(w_q, w_u, w_1, w_2)$ such that $0 < w_2 < w_1 \leq \min\{w_q, w_u\} < 1$ and $w_q d$ and $w_u d$ are integers.
\begin{definition}[The $(w_q,w_u,w_1,w_2)$-GapSS problem]
Given a dataset $\mathcal{D}$ of $n$ sets where each set $x \subseteq [d]$ has size $w_u d$, build a data structure so that for any query set $y \subseteq [d]$ of size $w_q d$: 
\begin{itemize}
    \item either returns $x'\in \mathcal{D}$ with $|x'\cap y| > w_2 d$,
    \item or determines that there is no $x\in \mathcal{D}$ with $|x\cap y| \ge w_1 d$.
\end{itemize}
\end{definition}

\paragraph{Relation of online OV to GapSS.}
When $w_1 = w_u \leq w_q$, the GapSS problem is equivalent to the (approximate) subset query problem: if there exists $x \in \mathcal{D}$ such that $x \subseteq y$, we need to output an approximate solution. 
Subset query is equivalent to online OV. We also view the sets $x, y \subseteq [d]$ as vectors in $\{0,1\}^d$. $x \subseteq y$ iff $x$ and $\overline{y}$ are orthogonal, where $\overline{y}_i = 1 - y_i$.

\subsection{A List-of-Points Lower Bound with KL Divergence Terms}\label{sec:lower_bound_KL}
In this section we use the results from \cite{ak20} to prove a list-of-points lower bound with KL divergence terms. We will simplify the KL divergence terms in the next section.

We first define a random instance for subset query, which follows from the discussion in the beginning of Section 3 of \cite{ak20}. Theorem~3 of \cite{ak20} sets $w_2 = w_u w_q$ but allows arbitrary $w_1$. Since we consider the subset query problem, we restrict $w_1 = w_u$. %
\begin{definition}[Random instance for subset query]\label{def:random_instance}
Let $0< w_u \leq w_q < 1$ be two parameters. We define a distribution of dataset-query pairs $(\mathcal{D}, y)$ where $\mathcal{D} \subseteq \{0,1\}^d$ and $|\mathcal{D}|=n$, and $y \in \{0,1\}^d$. A $(w_u, w_q)$-random dataset-query pair for subset query is drawn from the following distribution:
\begin{enumerate}
\item The query $y \in \{0,1\}^d$ is generated by independently sampling each of its $d$ bits from a $\text{Bernoulli}(w_q)$ distribution.
\item A dataset $\mathcal{D} \subseteq \{0,1\}^d$ is first constructed by sampling $n-1$ vectors, each with independently generated bits from $\text{Bernoulli}(w_u)$.
\item A special point $x' \in \{0,1\}^d$ is then created by sampling $x'_i \sim \text{Bernoulli}(w_u/w_q)$ if $y_i=1$, and setting $x'_i = 0$ otherwise. This point is also added to $\mathcal{D}$.
\end{enumerate}
The goal of the data structure is to preprocess $\mathcal{D}$ such that it recovers $x'$ when given the query $y$.
\end{definition}
Note that the special point $x'$ is always a subset of the query $y$, and we can equivalently sample each $(y_i, x'_i) \sim \text{Bernoulli}(P)$ where $P = \left[\begin{smallmatrix}
w_u & w_q - w_u \\
0 & 1-w_q
\end{smallmatrix}\right]$.

\paragraph{Lower bound of \cite{ak20}.} The proof of the lower bound (Theorem~3) of \cite{ak20} has two steps:
\begin{enumerate}
    \item {\bf Lower bound for random instances.} They first prove a lower bound for any list-of-points data structure that successfully solves an random instance similar to the above with probability $\geq 0.99$.
    \item {\bf Reduction of random instances to worst case.} Then they show that any data structure for the GapSS problem also directly solves the random instance, so the lower bound in Step 1 also applies to the GapSS problem. 
    
    To prove this, they need the dimension to be large: $w_q w_u d = \omega(\log n)$, so that in the random instance the sizes of the data points, the query, and their intersections all concentrate to within $1+o(1)$ factors of their expectations.
\end{enumerate}

In the proof of our lower bound (Theorem~\ref{thm:lower_bound_KL}), we will use the same lower bound for random instances (Step 1) as Theorem 3 of \cite{ak20}, while proving a different ``random instances to worst case'' reduction (Step 2) tailored for the exact subset query problem. 

We no longer need the $w_q w_u d = \omega(\log n)$ constraint since we allow data points and queries to have arbitrary number of ones. Instead, we need a different constraint $w_u (1 - w_q) > \frac{2}{c}$ so that in the random instance with high constant probability the special point $x'$ is the unique data point that satisfy $x' \subseteq y$.

\begin{theorem}[List-of-points lower bound with KL divergence terms]\label{thm:lower_bound_KL}
Assume $n \geq 200$. Consider any list-of-points data structure for the online OV problem over $n$ points of dimension $d = c \log n$, which has expected space $n^{1 + \rho_u}$ and expected query time $n^{\rho_q}$, and succeeds with probability $\geq 0.995$. Then for any two parameters $(w_q, w_u)$ that satisfy $0 < w_u \leq w_q < 1$ and $w_u (1 - w_q) > \frac{2}{c}$, and any $\alpha \in [0,1]$, we have that
\[
\alpha \rho_q + (1 - \alpha) \rho_u \geq \inf_{\substack{t_q, t_u \in [0,1] \\ t_u \neq w_u}} \left( \alpha \frac{\D{T}{P} - \dd{t_q}{w_q}}{\dd{t_u}{w_u}} + (1 - \alpha) \frac{\D{T}{P} - \dd{t_u}{w_u}}{\dd{t_u}{w_u}} \right),
\]
where $P = \left[\begin{smallmatrix}
w_u & w_q - w_u \\
0 & 1-w_q
\end{smallmatrix}\right]
$ and $T = \underset{T \ll P, \underset{X \sim T}{\E}[X] = [\substack{t_q \\ t_u}]}{\arg \inf} \D{T}{P}$. 
\end{theorem}
\begin{proof}
Since online OV is equivalent to subset query, in this proof we consider subset query instead. Let $A$ be a data structure as described in the theorem statement. Consider a $(w_u, w_q)$-random dataset-query pair $(\mathcal{D}, y)$ defined in Definition~\ref{def:random_instance}. It suffices to prove that $A$ solves this random instance with success probability $\geq 0.99$, and the rest of the proof follows the same way as that of Theorem~3 of \cite{ak20}. 

The data structure $A$ only fails to output the special point $x'$ in two scenarios: (1) There exist another $x \in \mathcal{D}$ such that $x \subseteq y$, and $A$ outputs $x$ instead. (2) $x'$ is the unique data point that is a subset of $y$, but the data structure fails on this input $(\mathcal{D}, y)$. The second event happens with probability $< 0.005$. It remains to prove that the probability that there exists another data point $x \subseteq y$ is at most $0.005$. Then, using the union bound, we have that $A$ fails to solve the random instance with probability $< 0.01$.

Consider any data point $x \neq x'$. Since the bits of $x$ are sampled from $\text{Bernoulli}(w_u)$ and the bits of $y$ are sampled from $\text{Bernoulli}(w_q)$, we have
\begin{align*}
    \Pr[x \subseteq y] = \big(1 - w_u (1 - w_q) \big)^d \leq e^{- w_u (1 - w_q) d} \leq \frac{1}{n^2},
\end{align*}
where the second step follows from $1-x \leq e^{-x}$, the third step follows from $w_u (1 - w_q) > \frac{2}{c}$ and $d = c \log n$.

Thus we have that the expected number of $x$ that is a subset of $y$ is $n \cdot \Pr[x \subseteq y] \leq 1/n$. Using Markov's inequality, the probability that there exists an $x \neq x'$ that is a subset of $y$ is at most $1/n$. When $n \geq 200$, this probability is $\leq 0.005$.
\end{proof}

\subsection{Simplifying KL Divergence Terms}\label{sec:simplify_KL}
In this section, we prove a lower bound for the term $\alpha \frac{\D{T}{P} - \dd{t_q}{w_q}}{\dd{t_u}{w_u}} + (1 - \alpha) \frac{\D{T}{P} - \dd{t_u}{w_u}}{\dd{t_u}{w_u}}$ of Theorem~\ref{thm:lower_bound_KL}. In the proof we will use the following two technical lemmas, and their proofs can be found in Appendix~\ref{sec:lower_bound_appendix}.

\begin{lemma}[KL divergence lower bounds]\label{lem:KL_divergence_lower_bound}
For any $t, w\in (0, 1)$, the KL divergence $\dd{t}{w}$ satisfies the following lower bounds:
\begin{itemize}
\item If $t \geq 1.1 w$, then $\dd{t}{w} \geq \Omega(t)$.
\item If $t \leq 0.9 w$, then $\dd{t}{w} \geq \Omega(w)$.
\item If $0.9 w \leq t \leq 1.1 w$ and $w \leq 0.1$, then $\dd{t}{w} \geq \Omega(\frac{(t - w)^2}{w})$.
\end{itemize}
\end{lemma}

\begin{lemma}[Range of $t_q$]\label{lem:range_tq}
Consider any $c \geq 500$. Let $w_q = 1 - \frac{10}{\sqrt{c}}$, $w_u = \frac{1}{\sqrt{c}}$, and $\alpha = 1 - \frac{20}{\sqrt{c}}$. Let $t_u \leq t_q \in [0,1]$ where $t_u \neq w_u$ satisfies that $\ln(\frac{(t_q-t_u)}{(w_q-w_u)}) = (1-\alpha) \ln(\frac{(1-t_q)}{(1-w_q)}) + \alpha \ln\frac{t_q}{w_q}$. Then $t_q$ must fall into the following ranges:
\begin{itemize}
\item If $t_u \geq w_u$, then $t_q \in [w_q, w_q + t_u - w_u]$.
\item If $t_u < w_u$, then $t_q \in [w_q - (w_u - t_u), w_q)$.
\end{itemize}

\end{lemma}

Using these two lemmas, we are ready to prove the lower bound for the KL divergence term of Theorem~\ref{thm:lower_bound_KL}.
\begin{theorem}[Simplification of KL Divergence Terms]\label{thm:simplify_KL}
Consider any $c \geq 500$. For $w_q = 1 - \frac{10}{\sqrt{c}}$, $w_u = \frac{1}{\sqrt{c}}$, $\alpha = 1 - \frac{20}{\sqrt{c}}$, we have
\[
\inf_{\substack{t_q, t_u \in [0,1] \\ t_u \neq w_u}} \left( \alpha \frac{\D{T}{P} - \dd{t_q}{w_q}}{\dd{t_u}{w_u}} + (1 - \alpha) \frac{\D{T}{P} - \dd{t_u}{w_u}}{\dd{t_u}{w_u}} \right) \geq \alpha \left(1 - \frac{\Theta(1)}{\sqrt{c}}\right),
\]
where $P = \left[\begin{smallmatrix}
w_u & w_q - w_u \\
0 & 1-w_q
\end{smallmatrix}\right]
$ and $T = \underset{T \ll P, \underset{X \sim T}{\E}[X] = [\substack{t_q \\ t_u}]}{\arg \inf} \D{T}{P}$. 
\end{theorem}
\begin{proof}
First note that in order to make $\D{T}{P}$ finite while satisfying $\underset{X \sim T}{\E}[X] = [\substack{t_q \\ t_u}]$, we must have
\[
T = \begin{bmatrix}
t_u & t_q - t_u \\
0 & 1 - t_q
\end{bmatrix}.
\]
The two terms in the theorem statement become
\begin{align*}
\frac{\D{T}{P} - \dd{t_q}{w_q}}{\dd{t_u}{w_u}} = f_q(t_q, t_u),~~~~~~
\frac{\D{T}{P} - \dd{t_u}{w_u}}{\dd{t_u}{w_u}} = f_u(t_q, t_u),
\end{align*}
where
\begin{align*}
f_q(t_q, t_u) = &~ \frac{t_u \ln \frac{t_u}{w_u} + (t_q - t_u) \ln \frac{t_q - t_u}{w_q - w_u} - t_q \ln \frac{t_q}{w_q}}{t_u \ln \frac{t_u}{w_u} + (1 - t_u)  \ln\frac{1-t_u}{1-w_u}}, \\
f_u(t_q, t_u) = &~ \frac{(1-t_q) \ln \frac{1-t_q}{1-w_q} + (t_q - t_u) \ln \frac{t_q - t_u}{w_q - w_u} - (1-t_u) \ln \frac{1-t_u}{1-w_u}}{t_u \ln \frac{t_u}{w_u} + (1 - t_u)  \ln\frac{1-t_u}{1-w_u}}.
\end{align*}
Note that the functions remain the same after changing all $\log_2$ to $\ln$. In this proof, with an abuse of notation, we denote $\dd{t}{w} = t \ln \frac{t}{w} + (1-t) \ln \frac{1-t}{1-w}$.

Define the function $f(t_q, t_u) = \alpha \cdot f_q(t_q, t_u) + (1-\alpha) \cdot f_u(t_q, t_u)$. Our goal is to prove that
\[
\inf_{\substack{t_q, t_u \in [0,1] \\ t_u \neq w_u}} f(t_q, t_u) \geq \alpha 
\left(1 - \frac{\Theta(1)}{\sqrt{c}}\right).
\]

The gradient of $f(t_q, t_u)$ wrt $t_q$ is
\begin{align}\label{eq:gradient_f_tq_tu}
\frac{\partial}{\partial t_q} f(t_q, t_u) 
= &~ \frac{1}{\dd{t_u}{w_u}} \cdot \Big( \ln(\frac{(t_q-t_u)}{(w_q-w_u)}) - (1-\alpha) \ln(\frac{(1-t_q)}{(1-w_q)}) - \alpha \ln\frac{t_q}{w_q} \Big).
\end{align}
Note that $\frac{\partial}{\partial t_q} f(t_q, t_u)$ approaches $-\infty$ as $t_q \to t_u$, and it approaches $+\infty$ as $t_q \to 1$. We also have that $\frac{\partial}{\partial t_q} f(t_q, t_u)$ is monotonically increasing wrt $t_q$ because 
\begin{align*}
\frac{\partial^2 }{\partial t_q^2} f(t_q, t_u) = &~ \frac{1}{\dd{t_u}{w_u}} \cdot \frac{\partial}{\partial t_q} \Big( \ln(\frac{(t_q-t_u)}{(w_q-w_u)}) - (1-\alpha) \ln(\frac{(1-t_q)}{(1-w_q)}) - \alpha \ln\frac{t_q}{w_q} \Big) \\
= &~ \frac{1}{\dd{t_u}{w_u}} \cdot (\frac{1}{t_q-t_u} + \frac{1-\alpha}{1-t_q} - \frac{\alpha}{t_q}) > 0,
\end{align*}
where the last step follows from $\alpha \in (0,1)$, $t_q \in (t_u, 1)$, and so $\frac{1}{t_q-t_u} > \frac{\alpha}{t_q}$. So we have that the function $f(t_q, t_u)$ decreases and then increases, and its minimum for a fixed $t_u$ is reached when $\frac{\partial}{\partial t_q} f(t_q, t_u) = 0$, i.e., when
\begin{align}\label{eq:optimality_condition}
\ln(\frac{(t_q-t_u)}{(w_q-w_u)}) = (1-\alpha) \ln(\frac{(1-t_q)}{(1-w_q)}) + \alpha \ln\frac{t_q}{w_q}.
\end{align}

{\bf Case 1 ($t_u \geq 0.1$).}
Note that we have
\begin{align*}
&~ (1-t_q) \log \frac{1-t_q}{1-w_q} + (t_q - t_u) \log \frac{t_q - t_u}{w_q - w_u} - (1-t_u) \log \frac{1-t_u}{1-w_u} \\
= &~ (1-t_u) \cdot \Big( \frac{1-t_q}{1-t_u} \log(\frac{(1-t_q)(1-w_u)}{(1-t_u)(1-w_q)}) + (\frac{t_q-t_u}{1-t_u} \log(\frac{t_q-t_u}{1-t_u} \cdot \frac{1-w_u}{w_q-w_u})) \Big) \\
= &~ (1-t_u) \cdot \dd{\frac{1-t_q}{1-t_u}}{\frac{1-w_q}{1-w_u}} \geq 0,
\end{align*}
so we have $f(t_q, t_u) = \alpha \cdot f_q(t_q, t_u) + (1-\alpha) \cdot f_u(t_q, t_u)$ satisfies
\begin{align*}
&~ f(t_q, t_u) \\
= &~ \frac{(t_q - t_u) \ln \frac{t_q - t_u}{w_q - w_u} + \alpha t_u \ln \frac{t_u}{w_u} - \alpha t_q \ln \frac{t_q}{w_q} + (1 - \alpha) (1-t_q) \ln \frac{1-t_q}{1-w_q} - (1 - \alpha) (1-t_u) \ln \frac{1-t_u}{1-w_u}}{\dd{t_u}{w_u}} \\
= &~ \alpha \cdot \Big(1 - \frac{\dd{t_q}{w_q}}{\dd{t_u}{w_u}}\Big) + \frac{(1-t_u) \cdot \dd{\frac{1-t_q}{1-t_u}}{\frac{1-w_q}{1-w_u}}}{\dd{t_u}{w_u}}.
\end{align*}
And it suffices to prove $\frac{\dd{t_q}{w_q}}{\dd{t_u}{w_u}} \leq \Theta(\frac{1}{\sqrt{c}})$.

Using Lemma~\ref{lem:range_tq}, in this case we have $t_q \geq w_q$, and so
\begin{align*}
\dd{t_q}{w_q} = &~ t_q \ln\frac{t_q}{w_q} + (1-t_q) \ln\frac{1-t_q}{1-w_q} \\
\leq &~ \ln(\frac{1}{w_q})
\leq \frac{1}{1-\frac{10}{\sqrt{c}}} - 1 \leq \Theta(1/\sqrt{c}),
\end{align*}
where in the third step we used that $\ln(x) \leq x - 1$ for all $x > 0$ and $w_q = 1 - \frac{10}{\sqrt{c}}$.

Using Lemma~\ref{lem:KL_divergence_lower_bound}, when $t_u \geq 0.1 > 1.1 w_u$, we have that $\dd{t_u}{w_u} \geq \Omega(t_u) \geq \Omega(1)$. Combining these two bounds we have that $\frac{\dd{t_q}{w_q}}{\dd{t_u}{w_u}} \leq \Theta(\frac{1}{\sqrt{c}})$.

{\bf Case 2 ($1.1 w_u \leq t_u < 0.1$ or $t_u \leq 0.9 w_u$).}
We first prove two useful bounds in this case. Using Lemma~\ref{lem:KL_divergence_lower_bound}, (1) when $t_u \leq 0.9 w_u$, we have $\dd{t_u}{w_u} \geq \Omega(w_u)$, and (2) when $1.1 w_u \leq t_u < 0.1$, we have $\dd{t_u}{w_u} \geq \Omega(t_u)$. This means in either case we always have $\dd{t_u}{w_u} \geq \Omega(|t_u-w_u|)$. Using Lemma~\ref{lem:range_tq}, we have that $|t_q - w_q| \leq |t_u - w_u|$, and since $t_q \leq 1$ and $t_u \geq 0$, we also always have $|t_q - w_q| \leq \min\{|t_u - w_u|, \Theta(\frac{1}{\sqrt{c}})\}$. 

Next we bound $f(t_q, t_u)$. Under the optimality condition Eq.~\eqref{eq:optimality_condition}, we have that the $(t_q,t_u)$ that reaches the minimum of $f(t_q, t_u) = \alpha \cdot f_q(t_q, t_u) + (1-\alpha) \cdot f_u(t_q, t_u)$ must satisfy
\begin{align*}
&~ f(t_q, t_u) \\
= &~ \frac{(t_q - t_u) \ln \frac{t_q - t_u}{w_q - w_u} + \alpha t_u \ln \frac{t_u}{w_u} - \alpha t_q \ln \frac{t_q}{w_q} + (1 - \alpha) (1-t_q) \ln \frac{1-t_q}{1-w_q} - (1 - \alpha) (1-t_u) \ln \frac{1-t_u}{1-w_u}}{\dd{t_u}{w_u}} \\
= &~ \frac{(1 - t_u) \ln \frac{t_q - t_u}{w_q - w_u} + \alpha t_u \ln \frac{t_u}{w_u} - \alpha \ln \frac{t_q}{w_q} - (1 - \alpha) (1-t_u) \ln \frac{1-t_u}{1-w_u}}{\dd{t_u}{w_u}} \\
= &~ \alpha + \frac{- \alpha \ln \frac{t_q}{w_q} + (1-t_u) \ln(\frac{t_q-t_u}{w_q-w_u}) - (1-t_u) \ln \frac{1-t_u}{1-w_u}}{\dd{t_u}{w_u}},
\end{align*}
where in the third step we plugged in $(1-\alpha) \log(\frac{1-t_q}{1-w_q}) = \log (\frac{t_q - t_u}{w_q - w_u}) - \alpha \log\frac{t_q}{w_q}$ by Eq.~\eqref{eq:optimality_condition}.

Next we have
\begin{align}\label{eq:f_tq_tu_case_2}
f(t_q, t_u) = &~ \alpha + \frac{- \alpha \ln \frac{t_q}{w_q} + (1-t_u) \ln(\frac{t_q-t_u}{w_q-w_u}) - (1-t_u) \ln \frac{1-t_u}{1-w_u}}{\dd{t_u}{w_u}} \notag \\
= &~ \alpha + \frac{- \alpha \ln \frac{t_q}{w_q} + (1-t_u) \ln(\frac{(t_q-t_u)(1-w_u)}{(w_q-w_u)(1-t_u)})}{\dd{t_u}{w_u}} \notag \\
\geq &~ \alpha + \frac{-\alpha (\frac{t_q - w_q}{w_q}) + \frac{(t_q-w_q)(1-w_u) - (t_u-w_u)(1-w_q)}{(w_q - w_u)} - 0.6 \frac{\big((t_q-w_q)(1-w_u) - (t_u-w_u)(1-w_q)\big)^2}{(w_q - w_u)^2(1-t_u)}}{\dd{t_u}{w_u}} \notag \\
= &~ \alpha + \frac{(t_q - w_q) \cdot (\frac{1-w_u}{w_q - w_u} - \frac{\alpha}{w_q}) - \frac{(t_u-w_u)(1-w_q)}{w_q - w_u} - 0.6 \frac{\big((t_q-w_q)(1-w_u) - (t_u-w_u)(1-w_q)\big)^2}{(w_q - w_u)^2(1-t_u)}}{\dd{t_u}{w_u}}
\end{align}
where the second step follows from $\ln(x) \leq x - 1$ for all $x > 0$, and $\ln(x) \geq (x-1) - 0.6 (x-1)^2$ for all $x$ such that $|x-1| \leq 0.2$, and note that $|\frac{t_q-w_q}{w_q}| \leq \frac{|t_u-w_u|}{w_q} \leq 0.2$ and $|\frac{(t_q-t_u)(1-w_u)}{(w_q-w_u)(1-t_u)} - 1| = |\frac{(t_q-w_q)(1-w_u) - (t_u-w_u)(1-w_q)}{(w_q - w_u)(1-t_u)}| \leq \frac{|t_u-w_u|}{(w_q - w_u)(1-t_u)} \leq 0.2$. 

Next we bound the three terms in the above equation using the parameter values $w_u = \frac{1}{\sqrt{c}}$, $w_q = 1 - \frac{10}{\sqrt{c}}$, $\alpha = 1 - \frac{20}{\sqrt{c}}$, and $c \geq 500$. For the first term, we have
\begin{align*}
\Big|(t_q - w_q) \cdot (\frac{1-w_u}{w_q - w_u} - \frac{\alpha}{w_q})\Big| = &~ |t_q - w_q| \cdot (\frac{1-\frac{1}{\sqrt{c}}}{1 - \frac{11}{\sqrt{c}}} - \frac{1-\frac{20}{\sqrt{c}}}{1-\frac{10}{\sqrt{c}}}) \leq \Theta(\frac{|t_u-w_u|}{\sqrt{c}}).
\end{align*}
For the second term, we have
\begin{align*}
\Big|\frac{(t_u-w_u)(1-w_q)}{w_q-w_u}\Big|  = \frac{|t_u-w_u| \frac{10}{\sqrt{c}}}{1-\frac{11}{\sqrt{c}}} \leq &~ \Theta(\frac{|t_u-w_u|}{\sqrt{c}}).
\end{align*}
For the third term, we have
\begin{align*}
\frac{\big((t_q-w_q)(1-w_u) - (t_u-w_u)(1-w_q)\big)^2}{(w_q - w_u)^2(1-t_u)} \leq &~ \Theta\Big((t_q-w_q)^2 + \frac{(t_u-w_u)^2}{c} \Big).
\end{align*}
Plugging these three bounds into Eq.~\eqref{eq:f_tq_tu_case_2}, we have
\begin{align*}
f(t_q, t_u) \geq &~ \alpha - \frac{\Theta\Big(\frac{|t_u-w_u|}{\sqrt{c}} + (t_q-w_q)^2 + \frac{(t_u-w_u)^2}{c} \Big)}{\dd{t_u}{w_u}} \\
\geq &~ \alpha - \frac{\Theta(1)}{\sqrt{c}},
\end{align*}
where we used that $\dd{t_u}{w_u} \geq \Omega(|t_u-w_u|)$, and $|t_q - w_q| \leq \min\{|t_u - w_u|, \Theta(\frac{1}{\sqrt{c}})\}$. 

Note that since $\alpha = 1 - \frac{20}{\sqrt{c}}$, this implies that $f(t_q,t_u) \geq \alpha (1 - \frac{\Theta(1)}{\sqrt{c}})$.

{\bf Case 3 ($0.9 w_u < t_u < 1.1 w_u$).}
Using Lemma~\ref{lem:KL_divergence_lower_bound}, in this case we have $\dd{t_u}{w_u} \geq \Omega(\frac{(t_u-w_u)^2}{w_u})$.

Using the same proof as Case 2, we still have Eq.~\eqref{eq:f_tq_tu_case_2} that
\begin{align*}
f(t_q, t_u) \geq &~ \alpha + \frac{(t_q - w_q) \cdot (\frac{1-w_u}{w_q - w_u} - \frac{\alpha}{w_q}) - \frac{(t_u-w_u)(1-w_q)}{w_q - w_u} - 0.6 \frac{\big((t_q-w_q)(1-w_u) - (t_u-w_u)(1-w_q)\big)^2}{(w_q - w_u)^2(1-t_u)}}{\dd{t_u}{w_u}}.
\end{align*}
However, in this case the denominator $\dd{t_u}{w_u}$ could be much smaller than $\Theta(\frac{|t_u-w_u|}{\sqrt{c}})$, and we need to consider cancellations of the first two terms. 

Using the optimality condition of Eq.~\eqref{eq:optimality_condition}, and the fact that $\ln(x) \leq x - 1$ for all $x > 0$, and $\ln(x) \geq (x-1) - 0.6 (x-1)^2$ for all $x$, we have the following two equations:
\begin{align*}
&~ \frac{(t_q-w_q)-(t_u-w_u)}{(w_q-w_u)} - 0.6 \cdot (\frac{(t_q-w_q)-(t_u-w_u)}{(w_q-w_u)})^2 \leq -(1-\alpha) \frac{t_q-w_q}{1-w_q} + \alpha \frac{t_q-w_q}{w_q}, \\
&~ \frac{(t_q-w_q)-(t_u-w_u)}{(w_q-w_u)} \geq -(1-\alpha) \frac{t_q-w_q}{1-w_q} + \alpha \frac{t_q-w_q}{w_q} - 0.6 (1-\alpha) (\frac{t_q-w_q}{1-w_q})^2 - 0.6 \alpha (\frac{t_q-w_q}{w_q})^2.
\end{align*}
These two equations imply that
\begin{align*}
\frac{t_u - w_u}{w_q - w_u} - (t_q - w_q) \cdot \Big(\frac{1}{w_q-w_u} + \frac{1-\alpha}{1-w_q} - \frac{\alpha}{w_q}\Big) \geq &~ - 0.6 \cdot (\frac{(t_q-w_q)-(t_u-w_u)}{(w_q-w_u)})^2, \\
\text{and } \frac{t_u - w_u}{w_q - w_u} - (t_q - w_q) \cdot \Big(\frac{1}{w_q-w_u} + \frac{1-\alpha}{1-w_q} - \frac{\alpha}{w_q}\Big) \leq &~ 0.6 (t_q - w_q)^2 \cdot \Big( \frac{1-\alpha}{(1-w_q)^2} + \frac{\alpha}{w_q^2} \Big).
\end{align*}
Since $\frac{1}{(w_q-w_u)^2} = \Theta(1)$, $\frac{1-\alpha}{(1-w_q)^2} = \frac{20/\sqrt{c}}{100/c} = \Theta(\sqrt{c})$, $\frac{\alpha}{w_q^2} = \Theta(1)$, and $|t_q-w_q| \leq |t_u-w_u|$, we have
\begin{align*}
\Big| \frac{t_u - w_u}{w_q - w_u} - (t_q - w_q) \cdot \Big(\frac{1}{w_q-w_u} + \frac{1-\alpha}{1-w_q} - \frac{\alpha}{w_q}\Big) \Big| \leq \Theta(\sqrt{c} \cdot (t_q-w_q)^2 ).
\end{align*}
Using this, we can bound the first two terms of Eq.~\eqref{eq:f_tq_tu_case_2} as
\begin{align*}
&~ \Big| (t_q - w_q) \cdot (\frac{1-w_u}{w_q - w_u} - \frac{\alpha}{w_q}) - \frac{(t_u-w_u)(1-w_q)}{w_q - w_u} \Big| \\ 
\leq &~ \Big|(t_q - w_q) \cdot \Big( \frac{1-w_u}{w_q-w_u}-\frac{\alpha}{w_q} - (1-w_q) \Big(\frac{1}{w_q-w_u} + \frac{1-\alpha}{1-w_q} - \frac{\alpha}{w_q}\Big) \Big) \Big| + \Theta(\sqrt{c} \cdot (t_q-w_q)^2 (1-w_q)) \\
= &~ \Theta(\sqrt{c} \cdot (t_q-w_q)^2 (1-w_q)) = \Theta((t_q-w_q)^2),
\end{align*}
where in the third step we used that $w_q = 1 - \frac{10}{\sqrt{c}}$.

Plugging the equation above into Eq.~\eqref{eq:f_tq_tu_case_2}, and bounding the third term of Eq.~\eqref{eq:f_tq_tu_case_2} in the same way as Case 2 by $\frac{\big((t_q-w_q)(1-w_u) - (t_u-w_u)(1-w_q)\big)^2}{(w_q - w_u)^2(1-t_u)} \leq \Theta\Big((t_q-w_q)^2 + \frac{(t_u-w_u)^2}{c} \Big)$, we have
\begin{align*}
f(t_q, t_u) \geq &~ \alpha - \frac{\Theta\Big((t_q-w_q)^2 + \frac{(t_u-w_u)^2}{c} \Big)}{\dd{t_u}{w_u}} \\
\geq &~ \alpha - \frac{\Theta\Big((t_q-w_q)^2 + \frac{(t_u-w_u)^2}{c} \Big)}{\frac{(t_u-w_u)^2}{w_u}} \\
\geq &~ \alpha - \Theta(w_u) \\
\geq &~ \alpha (1 - \Theta(1/\sqrt{c})).
\end{align*}
\end{proof}

\subsection{Proof of the Main Lower Bound Theorem}\label{sec:proof_lower_bound}
Combining Theorem~\ref{thm:lower_bound_KL} and Theorem~\ref{thm:simplify_KL}, we have the following corollary.
\begin{corollary}\label{cor:1/sqrt_c_lower_bound}
Assume $n \geq c \geq 500$. Consider any list-of-points data structure for the online OV problem over $n$ points of dimension $d = c \log n$, which has expected space $n^{1 + \rho_u}$ and expected query time $n^{\rho_q}$, and succeeds with probability $\geq 0.995$. Then for we have that either $\rho_u \geq 1$, or
\begin{align*}
\rho_q \geq 1-\frac{\Theta(1)}{\sqrt{c}}.
\end{align*}
\end{corollary}
\begin{proof}
Combining Theorem~\ref{thm:lower_bound_KL} and Theorem~\ref{thm:simplify_KL}, we have that 
\begin{align*}
\alpha \rho_q + (1 - \alpha) \rho_u \geq \alpha \left(1 - \frac{C_0}{\sqrt{c}}\right),
\end{align*}
where $\alpha = 1 - \frac{20}{\sqrt{c}}$, and $C_0>0$ is a constant. So we have that either $(1-\alpha) \rho_u \geq \frac{21 \alpha}{\sqrt{c}}$, or
\begin{align*}
\alpha \rho_q \geq \alpha \left(1 - \frac{C_0 + 21}{\sqrt{c}}\right).
\end{align*}
And this immediately implies the claim that either $\rho_u \geq 1$ or $\rho_q \geq 1 - \frac{\Theta(1)}{\sqrt{c}}$.
\end{proof}

Now we are ready to prove the main lower bound Theorem~\ref{thm:main_lower_bound}.
\begin{proof}[Proof of Theorem~\ref{thm:main_lower_bound}]
Let $n^{\rho_q}$ and $n^{\rho_u + 1}$ denote the query time and space of the data structure described in the theorem statement. Let $\gamma$ be any number in range $[1, \sqrt{c}]$. We split the the dataset in groups of size $n^{1/\gamma}$ and build a separate data structure for each group. This gives another data structure with space $n^{1+\rho'_u}$ and query time $n^{\rho'_q}$, where
\begin{align*}
\rho'_u = \rho_u / \gamma, 
\rho'_q = 1 -(1-\rho_q) / \gamma.
\end{align*}
Using Corollary~\ref{cor:1/sqrt_c_lower_bound}, we have that either $\rho'_u \geq 1$, or $\rho'_q \geq 1 - \frac{\Theta(1)}{\sqrt{c}}$. And means we have that either $\rho_u \geq \gamma$, or $\rho_q \geq \frac{\Theta(\gamma)}{\sqrt{c}}$.
\end{proof}

\bibliographystyle{alpha}
\bibliography{ref,andoni}

\newpage
\appendix
\section{Tightness of Distributional Communication Protocol for Set Disjointness}\label{sec:tight_communication_protocol}
In this section we show that the communication protocol of Lemma~\ref{lem:w_sparse_subset_query_protocol} is essentially tight for $w = d$. We will convert this protocol to the standard communication model as defined in Definition~\ref{def:standard_communication_model}, and prove a matching lower bound.

\subsection{Upper bound}
In this section we show the upper bound: we show that the  communication protocol of Lemma~\ref{lem:w_sparse_subset_query_protocol} can be converted to the standard communication model as defined in Definition~\ref{def:standard_communication_model}.

\begin{theorem}[Protocol of Lemma~\ref{lem:w_sparse_subset_query_protocol} in standard communication protocol]\label{thm:communication_protocol}
Fixing $\lambda$ and $\rho$ to be two distributions over $\{0,1\}^{d}$. For any $\epsilon \in (0, 1/2)$, and any integers $d, \ell \in \mathbb{N}_+$, there exists a deterministic communication protocol for set disjointness over $\lambda \times \rho$ with $\epsilon$ type II error, and during the protocol
\begin{itemize}
    \item Alice sends $a \leq \ell + \lceil d/\ell \rceil \cdot \log(3e \ell)$ bits,
    \item Bob sends $b \leq 2 \ell \cdot \log(16 \ell/\epsilon))$ bits.
\end{itemize}
\end{theorem}
\begin{corollary}
By choosing $\ell = \sqrt{d / \log(1 / \epsilon)} \geq 1$, there is a communication protocol for set disjointness over product distribution which has $\epsilon$ type II error and has communication complexity $O(\sqrt{d \log(1/\epsilon)} \cdot \log d)$.
\end{corollary}
\begin{proof}[Proof of Theorem~\ref{thm:communication_protocol}]
In Lemma~\ref{lem:w_sparse_subset_query_protocol} we showed a protocol in the special communication model (Definition~\ref{def:special_communication_model}) with the same communication cost as in this theorem. Since the special communication model (Definition~\ref{def:special_communication_model}) is a more restricted model than the standard communication model (Definition~\ref{def:standard_communication_model}), we directly have a protocol in the standard communication model that works for any distribution $\lambda \times \rho$. This resulting protocol uses public randomness, and the only remaining issue is to make it deterministic.

\paragraph{Fixing public randomness.}
Fixing distributions $\lambda$ and $\rho$. We have shown that there exists a protocol $\pi$ which uses public randomness $R$ and satisfies that
\[
\Pr_{x \sim \lambda, y \sim \rho, R}[\pi_R(x, y) \neq f(x, y)] \leq \epsilon.
\]
Thus there must exist a fixed public randomness $R_0$ such that 
\[
\Pr_{x \sim \lambda, y \sim \rho}[\pi_{R_0}(x, y) \neq f(x, y)] \leq \epsilon.
\]
Fixing $R = R_0$, we get the desired deterministic protocol.
\end{proof}

\subsection{Lower bound}
In this section we show a matching lower bound in Theorem~\ref{thm:lower_bound}. We prove it in a similar way as Theorem~11 of \cite{mnsw98} using a richness lemma (Lemma~5 of \cite{mnsw98}).

Before presenting our main theorem, we first prove a tool.
\begin{fact}[Disjoint probability]\label{fac:disjoint_prob}
Let $\epsilon \in (0, 1/2)$. Let $k \leq l < d/3$ be two integers that satisfy $k \cdot l < d \ln(1/\epsilon) / 3$. Let $\lambda$ be the uniform distribution over $x \subseteq [d]$ of size $k$, and let $\rho$ be the uniform distribution over $y \subseteq [d]$ of size $l$. We have
\[
\Pr_{x \sim \lambda, y \sim \rho}[x \cap y = \emptyset] \geq \epsilon.
\]
\end{fact}
\begin{proof}
\begin{align*}
    \Pr_{x \sim \lambda, y \sim \rho}[x \cap y = \emptyset] = \frac{\binom{d-k}{l}}{\binom{d}{l}} \geq (1 - \frac{k}{d - l})^l \geq e^{-kl/(d-k-l)} \geq e^{-(d/3) \ln(1/\epsilon) / (d/3)} = \epsilon,
\end{align*}
where the second step follows from $\frac{d-k}{d} \geq \frac{d-k-1}{d-1} \geq \cdots \geq \frac{d-k-l}{d-l}$, the third step follows from $(1-\frac{k}{d-l})^{((d-l)/k) - 1} \geq e^{-1}$, the fourth step follows from $k,l < d/3$ and $kl < d \ln(1/\epsilon) / 3$.
\end{proof}

We include the richness lemma (Lemma~5 of \cite{mnsw98}) here for completeness.

For any communication problem $f: X \times Y \to \{0,1\}$ where $X, Y \subseteq \{0,1\}^d$, we identify $f$ with the matrix $M$ where $M_{x,y} = f(x,y)$. We say $f$ is a $(u, v)$-rich problem if $M$ has at least $v$ columns each containing at least $u$ number of 1's.
\begin{lemma}[Richness lemma, Lemma~5 of \cite{mnsw98}]\label{lem:richness}
Let $f: X \times Y \to \{0,1\}$ be a $(u, v)$-rich communication problem. If there exists a deterministic protocol $\pi$ for $f$ where Alice sends $a$ bits and Bob sends $b$ bits, then $f$ contains a submatrix of dimension at least $(u/2^a) \times (v/2^{a+b})$ that only contains 1's.
\end{lemma}

Now we are ready to prove our main theorem.
\begin{theorem}[Lower bound]\label{thm:lower_bound}
Let $d$ be the dimension of the set disjointness problem. Let $\epsilon \in (0, 1/2)$ be the error parameter.

Let $k \leq l < d/3$ be any two integers that satisfy $k \cdot l < d \ln(1/(3\epsilon)) / 3$. Let $\lambda$ be the uniform distribution over $x \subseteq [d]$ of size $k$, and let $\rho$ be the uniform distribution over $y \subseteq [d]$ of size $l$.

For any deterministic protocol $\pi$ over the product distribution $\lambda \times \rho$ with $\epsilon$ error type II error, we have that either Alice sends $a = \Omega(k)$ bits, or Bob sends $b = \Omega(l)$ bits.
\end{theorem}
\begin{corollary}
1. Choosing $k = d/4$ and $l = \ln(1/(3 \epsilon))$, then either Alice sends $a = \Omega(d)$ bits, or Bob sends $b = \Omega(\log(1/\epsilon))$ bits.

2. Choosing $k = l = \sqrt{d \ln(1/(3 \epsilon))}/2$, then the total cost of the protocol is $a+b \geq \Omega(\sqrt{d \log(1/ \epsilon)})$.
\end{corollary}
\begin{proof}[Proof of Theorem~\ref{thm:lower_bound}]
Let $X \subseteq 2^{[d]}$ be the set of $x \subseteq [d]$ of size $k$, and let $Y \subseteq 2^{[d]}$ be the set of $y \subseteq [d]$ of size $l$. We use $f: X \times Y \to \{0,1\}$ to denote the disjointness function, i.e., $f(x,y) = \mathbf{1}_{x \cap y = \emptyset}$. We define a function $f': X \times Y \to \{0, 1\}$ such that $f'(x,y) = \pi(x,y)$. Note that by definition $\pi$ is a deterministic protocol that solves $f'$ with no error. And since $\pi$ has type II error, we know that $f'(x,y) = 0$ when $f(x,y) = 0$.

Since the error probability of $\pi$ is $\epsilon$ for $f$, we have
\[
\Pr_{(x,y) \sim \lambda \times \rho}[\pi(x,y) \neq f(x,y)] \leq \epsilon.
\]
Since $k \leq l < d/3$ and $k \cdot l < d \ln(1/(3\epsilon)) / 3$, using Fact~\ref{fac:disjoint_prob} we have $\Pr_{(x, y) \sim \lambda \times \rho}[f(x,y) = 1] \geq 3 \epsilon$. Combining these two equations we have
\begin{align}\label{eq:prob_pi_neq_f_on_1}
    \Pr_{(x,y) \sim \lambda \times \rho|_{f(x,y) = 1} }[\pi(x,y) \neq f(x,y)] \leq 1/3.
\end{align}
The communication matrix of $f$ has $\binom{d}{l}$ columns and each column has $\binom{d-l}{k}$ number of 1's.
Eq.~\eqref{eq:prob_pi_neq_f_on_1} means that in the communication matrix of $f'$, there are at least $\frac{2}{3} \cdot \binom{d}{l} \cdot \binom{d-l}{k}$ number of 1's. Thus $f'$ must have at least $\binom{d}{l}/3$ columns each contains at least $\binom{d-l}{k}/3$ number of 1's, i.e., $f'$ is $(\binom{d-l}{k}/3, \binom{d}{l}/3)$-rich. So using the richness lemma (Lemma~\ref{lem:richness}) we know that $f'$ has a 1-rectangle of size at least $r \times s$ where $r := \binom{d-l}{k}/2^{a+2}$ and $s := \binom{d}{l}/2^{a+b+2}$.

Since $\pi$ has type II error, i.e., $f'(x,y) = 0$ when $f(x,y) = 0$, the function $f$ also has a 1-rectangle of size at least $r \times s$.

Next we use an identical argument as the proof of Theorem~11 in \cite{mnsw98}. Let $x_1, x_2, \cdots, x_r \subseteq [d]$ and $y_1, y_2, \cdots, y_s \subseteq [d]$ denote the sets that correspond to the rows and columns of the 1-rectangle. Note that $x_i \cap y_j = \emptyset$ for all $i,j$. Let $x = \bigcup_i x_i$ and $y = \bigcup_j y_j$. Since every $x_i$ is a subset of $x$, and $x$ has $\binom{|x|}{k}$ subsets of size $k$, we must have $\binom{|x|}{k} \geq r = \binom{d-l}{k}/2^{a+2}$. Let $t = (d-l) / 2^{(a+2)/k}$. Since $\binom{t}{k} < \binom{d-l}{k} / 2^{a+2}$, we have $|x| \geq t$. Since $y \cap x = \emptyset$, we have $|y| \leq d - t$. Since every $y_j$ is a subset of $y$, and $y$ has $\binom{|y|}{l}$ subsets of size $l$, we must have $\binom{d - t}{l} \geq \binom{|y|}{l} \geq s = \binom{d}{l} / 2^{a+b+2}$. This gives
\begin{align}\label{eq:lb_inequality}
    2^{a+b+2} \geq \frac{\binom{d}{l}}{\binom{d - t}{l}} \geq (\frac{d}{d-t})^{l} \geq (1 + \frac{t}{d})^l = (1 + \frac{d-l}{2^{(a+2)/k} \cdot d})^l \geq (1 + 2^{-(a+2)/k - 1})^l,
\end{align}
where the second step follows from $\frac{d}{d-t} \leq \frac{d-1}{d-t-1}\leq \cdots \leq \frac{d-l+1}{d-t-l+1}$, the fourth step follows from the definition $t = (d-l) / 2^{(a+2)/k}$, the fifth step follows from $(d-l)/d \geq 1/2$ since $l < d/3$.

The inequality of Eq.~\eqref{eq:lb_inequality} holds only if either $a = \Omega(k)$ or $b = \Omega(l)$.
\end{proof}

\section{Missing proofs of Section~\ref{sec:lop_lower_bound}}\label{sec:lower_bound_appendix}
In this section we provide the proof of Lemma~\ref{lem:KL_divergence_lower_bound} and Lemma~\ref{lem:range_tq}.

\begin{lemma}[KL divergence lower bounds, Restatement of Lemma~\ref{lem:KL_divergence_lower_bound}]
For any $t, w\in (0, 1)$, the KL divergence $\dd{t}{w}$ satisfies the following lower bounds:
\begin{itemize}
\item If $t \geq 1.1 w$, then $\dd{t}{w} \geq \Omega(t)$.
\item If $t \leq 0.9 w$, then $\dd{t}{w} \geq \Omega(w)$.
\item If $0.9 w \leq t \leq 1.1 w$ and $w \leq 0.1$, then $\dd{t}{w} \geq \Omega(\frac{(t - w)^2}{w})$.
\end{itemize}
\end{lemma}
\begin{proof}
Recall that $\dd{t}{w} = t \log(\frac{t}{w}) + (1-t) \log(\frac{1-t}{1-w})$. In this proof we consider the function $t \ln(\frac{t}{w}) + (1-t) \ln(\frac{1-t}{1-w}) = \dd{t}{w} \cdot \ln(2)$, and it suffices to prove the claimed lower bounds for this function.

{\bf Case 1 ($t \geq 1.1 w$).} Note that the function $h(t) = t \ln(\frac{t}{w}) + (1-t) \ln(\frac{1-t}{1-w}) - 0.001 t$ is monotonically increasing when $t \geq 1.1w$ since its gradient is
\begin{align*}
h'(t) = \ln(\frac{t}{w}) - \ln(\frac{1-t}{1-w}) - 0.001 > \ln(1.1) - 0.001 > 0.
\end{align*}

We also have that $h(1.1w) \geq 0$ because 
\begin{align*}
h(1.1w) = 1.1w \ln(1.1) + (1-1.1w) \ln(\frac{1-1.1w}{1-w}) - 0.0011 w
\end{align*}
is a function that is positive for any $w \in (0,\frac{1}{1.1})$.
This is because letting $g(w) = h(1.1w)$, we have
\begin{align*}
g'(w) = &~ 1.1 \ln(1.1) - 1.1 \ln(\frac{1-1.1w}{1-w}) - (1-1.1w) \frac{1-w}{1-1.1w} \cdot \frac{0.1}{(1-w)^2} - 0.0011 \\
= &~ 1.1 \ln(1.1) - 0.0011 - 1.1 \ln(\frac{1-1.1w}{1-w}) - \frac{0.1}{1-w} \\
\geq &~ 1.1 \ln(1.1) - 0.0011 + 1.1 \frac{0.1w}{1-w} - \frac{0.1}{1-w} \\
\geq &~ 1.1 \ln(1.1) - 0.0011 - 0.1 > 0,
\end{align*}
where in the third step we used that $\ln(x) \leq x-1$ for all $x > 0$. So $g(w)$ is monotonically increasing, and so $g(w) \geq g(0) = 0$.

Using these two facts, we have that for all $t \geq 1.1w$, $h(t) \geq h(1.1w) \geq 0$, and so $t \ln(\frac{t}{w}) + (1-t) \ln(\frac{1-t}{1-w}) \geq 0.001 t$.

{\bf Case 2 ($t \leq 0.9 w$).} Note that the function $h(t) = t \ln(\frac{t}{w}) + (1-t) \ln(\frac{1-t}{1-w}) - 0.001 w$ is monotonically decreasing when $t \leq 0.9w$ since its gradient is
\begin{align*}
h'(t) = \ln(\frac{t}{w}) - \ln(\frac{1-t}{1-w}) < \ln(0.9) < 0.
\end{align*}

We also have $h(0.9w) > 0$ because 
\begin{align*}
h(0.9w) = 0.9w \ln(0.9) + (1-0.9w) \ln(\frac{1-0.9w}{1-w}) - 0.001 w
\end{align*}
is a function that is positive for any $w \in (0,1)$.
This is because letting $g(w) = h(0.9w)$, we have
\begin{align*}
g'(w) = &~ 0.9 \ln(0.9) - 0.9 \ln(\frac{1-0.9w}{1-w}) + (1-0.9w) \frac{1-w}{1-0.9w} \cdot \frac{0.1}{(1-w)^2} - 0.001 \\
= &~ 0.9 \ln(0.9) - 0.001 - 0.9 \ln(\frac{1-0.9w}{1-w}) +  \frac{0.1}{1-w} \\
\geq &~ 0.9 \ln(0.9) - 0.001 - 0.9 \frac{0.1w}{1-w} + \frac{0.1}{1-w} \\
\geq &~ 0.9 \ln(0.9) - 0.001 +0.1 > 0,
\end{align*}
where in the third step we used that $\ln(x) \leq x-1$ for all $x > 0$. So $g(w)$ is monotonically increasing, and so $g(w) \geq g(0) = 0$.

Using these two facts, we have that for all $t \leq 0.9w$, $h(t) \geq h(0.9w) \geq 0$, and so $t \ln(\frac{t}{w}) + (1-t) \ln(\frac{1-t}{1-w}) \geq 0.001 w$.

{\bf Case 3 ($0.9w \leq t \leq 1.1w$ and $w \leq 0.1$). } For this case we use the fact that $\ln(x) \geq (x-1) - 0.6 (x-1)^2$ for all $x$ such that $|x-1| \leq 0.2$. We have
\begin{align*}
t \ln(\frac{t}{w}) + (1-t) \ln(\frac{1-t}{1-w}) 
= &~ t \ln(1 + \frac{t-w}{w}) + (1-t) \ln(1 - \frac{t-w}{1-w}) \\
\geq &~ t \cdot \frac{t-w}{w} - 0.6 t\cdot (\frac{t-w}{w})^2 - (1-t) \cdot \frac{t-w}{1-w} - 0.6 (1-t) \cdot (\frac{t-w}{1-w})^2 \\
= &~ \frac{(t-w)^2}{w(1-w)} \cdot \Big(1 - 0.6 \frac{t(1-w)}{w} - 0.6 \frac{(1-t)w}{(1-w)} \Big) \\
\geq &~ \frac{(t-w)^2}{w(1-w)} \cdot \Big(1 - 0.66 - 0.6 \frac{0.1}{0.9} \Big) \\
\geq &~ 0.3 \frac{(t-w)^2}{w},
\end{align*}
where in the second step we used the lower bound of $\ln(x)$ since $|\frac{t-w}{1-w}| \leq |\frac{t-w}{w}| \leq 0.1$ because $w \leq 0.1$ and $|t-w|\leq 0.1w$.
\end{proof}

\begin{lemma}[Range of $t_q$, Restatement of Lemma~\ref{lem:range_tq}]
Consider any $c \geq 500$. Let $w_q = 1 - \frac{10}{\sqrt{c}}$, $w_u = \frac{1}{\sqrt{c}}$, and $\alpha = 1 - \frac{20}{\sqrt{c}}$. Let $t_u \leq t_q \in [0,1]$ where $t_u \neq w_u$ satisfies that $\ln(\frac{(t_q-t_u)}{(w_q-w_u)}) = (1-\alpha) \ln(\frac{(1-t_q)}{(1-w_q)}) + \alpha \ln\frac{t_q}{w_q}$. Then $t_q$ must fall into the following ranges:
\begin{itemize}
\item If $t_u \geq w_u$, then $t_q \in [w_q, w_q + t_u - w_u]$.
\item If $t_u < w_u$, then $t_q \in [w_q - (w_u - t_u), w_q)$.
\end{itemize}
\end{lemma}
\begin{proof}
Consider any fixed $t_u \in [0,1] \backslash \{w_u\}$. Define a function
\begin{align*}
h(t_q) = \ln(\frac{(t_q-t_u)}{(w_q-w_u)}) - (1-\alpha) \ln(\frac{(1-t_q)}{(1-w_q)}) - \alpha \ln\frac{t_q}{w_q}.
\end{align*}
Its gradient is
\begin{align*}
h'(t_q) = \frac{1}{t_q - t_u} + \frac{1-\alpha}{1-t_q} - \frac{\alpha}{t_q}.
\end{align*}
We always have $h'(t_q) > 0$ for the domain $t_q \in (t_u, 1)$, so the function $h(t_q)$ is monotonically increasing. %
Next we characterize the range of $t_q^*$ that satisfies $h(t_q^*) = 0$.

{\bf Case 1 ($t_u \geq w_u$).} First note that we have
\[
h(w_q) = \ln(\frac{w_q-t_u}{w_q-w_u}) \leq 0.
\]
We also have
\begin{align*}
h(w_q + t_u - w_u) = &~ - (1-\alpha) \ln(\frac{(1-w_q-(t_u - w_u))}{(1-w_q)}) - \alpha \ln\frac{w_q + (t_u - w_u)}{w_q} \\
\geq &~ (1-\alpha) \frac{t_u-w_u}{1-w_q} - \alpha \frac{t_u-w_u}{w_q} \\
= &~ (t_u - w_u) \cdot (2 - \frac{1-\frac{20}{\sqrt{c}}}{1-\frac{10}{\sqrt{c}}}) 
> 0,
\end{align*}
where in the second step we used that $\ln(x) \leq x-1$ for all $x > 0$, the third step follows from the parameters $w_q = 1 - \frac{10}{\sqrt{c}}$ and $\alpha = 1 - \frac{20}{\sqrt{c}}$.

Combining these two equations, we have that the solution $t_q^*$ of $h(t_q^*) = 0$ must satisfy $t_q^* \in (w_q, w_q + t_u - w_u]$. %

{\bf Case 2 ($t_u < w_u$).} First note that we have
\begin{align*}
h(w_q) = \ln(\frac{w_q-t_u}{w_q-w_u}) > 0.
\end{align*}

Next we prove that $h(w_q + t_u - w_u) \leq 0$.
We have
\begin{align*}
h(w_q + t_u - w_u) = &~ - (1-\alpha) \ln(1 + \frac{w_u - t_u}{1-w_q}) - \alpha \ln(1-\frac{w_u-t_u}{w_q}) \\
\leq &~ - (1-\alpha) \frac{w_u - t_u}{1-w_q} + 0.6 (1-\alpha) (\frac{w_u - t_u}{1-w_q})^2 + \alpha \frac{w_u-t_u}{w_q} + 0.6 \alpha (\frac{w_u-t_u}{w_q})^2 \\
= &~ - 2 \cdot (w_u-t_u) + 0.12 \sqrt{c} (w_u-t_u)^2 + \frac{1-\frac{20}{\sqrt{c}}}{1-\frac{10}{\sqrt{c}}} (w_u - t_u) + 0.6 \frac{1-\frac{20}{\sqrt{c}}}{(1-\frac{10}{\sqrt{c}})^2} (w_u - t_u)^2 \\
\leq &~ (w_u - t_u) \cdot \Big( -2 + 0.12 + 1 + 0.6 \Big) < 0,
\end{align*}
where the second step follows from $\ln(x) \geq (x-1) - 0.6 (x-1)^2$ for all $x$ such that $|x-1| \leq 0.2$, and $|\frac{w_u - t_u}{1-w_q}| \leq 0.2$ and $|\frac{w_u-t_u}{w_q}| \leq 0.2$, the third step follows from $w_q = 1 - \frac{10}{\sqrt{c}}$ and $\alpha = 1 - \frac{20}{\sqrt{c}}$, the fourth step follows from $w_u - t_u \leq w_u = \frac{1}{\sqrt{c}}$.

Combining these two equations, we have that the solution $t_q^*$ of $h(t_q^*) = 0$ must satisfy $t_q^* \in [w_q - (w_u - t_u), w_q)$. %
\end{proof}

\end{document}